\PassOptionsToPackage{svgnames}{xcolor}
\documentclass[acmsmall]{acmart}
\usepackage[T1]{fontenc}
\usepackage[utf8]{inputenc}
\usepackage[english]{babel}

\usepackage{graphicx}

\usepackage{xspace}

\usepackage{amsmath}
\usepackage{amsfonts}
\usepackage{amsthm}
\usepackage{thm-restate}
\usepackage{mathtools}
\usepackage{textcomp}
\usepackage{gensymb}
\usepackage{multicol}

\usepackage{listings}
\usepackage[scaled=0.85]{DejaVuSansMono}

\usepackage{tikz}
\usetikzlibrary{arrows}

\usepackage{caption}
\usepackage{subcaption}
\usepackage[inline,shortlabels]{enumitem}
\setlist{leftmargin=*,noitemsep}
\usepackage{array}
\usepackage{adjustbox}
\usepackage{bm}
\usepackage{pifont}

\usepackage{natbib}
\usepackage{mathpartir} 

\mprset{sep=1.5em}

\usepackage{tikz}
\usetikzlibrary{decorations.text,backgrounds,positioning,shapes,
  shadings,shadows,arrows,decorations.markings,calc,fit,fadings,
  tikzmark
}

\usepackage[nameinlink,capitalize]{cleveref}
\addto\extrasenglish{%

}
\newtheorem{property}{Property}

\author{Gabriel Radanne}
\affiliation{
  \institution{Inria}
  \country{Paris}
}
\email{radanne@informatik.uni-freiburg.de}

\author{Hannes Saffrich}
\affiliation{
  \institution{University of Freiburg}
  \country{Germany}
}
\email{saffrich@informatik.uni-freiburg.de}

\author{Peter Thiemann}
\affiliation{
  \institution{University of Freiburg}
  \country{Germany}
}
\email{thiemann@informatik.uni-freiburg.de}

\newcommand\htag[1]{\shortintertext{\textbf{#1}}}

\newcommand\affe{Affe\xspace}

\usepackage{xcolor}
\definecolor{butter}{HTML}{C4A000}
\definecolor{orange}{HTML}{CE5C00}
\definecolor{chocolate}{HTML}{8F5902}
\definecolor{chameleon}{HTML}{4E9A06}
\definecolor{skyblue}{HTML}{204A87}
\definecolor{plum}{HTML}{5C3566}
\definecolor{scarletred}{HTML}{A40000}
\definecolor{lightalu}{HTML}{BABDB6}
\definecolor{darkalu}{HTML}{2E3436}

\newcommand{\kwstyle}{\bfseries}

\input{lst-rule}
\newcommand\lang{Affe\xspace}
\newcommand\hmx{$\operatorname{HM}(X)$\xspace}

\newcommand\lk{\leq}
\newcommand\mul{Q}
\newcommand\klin{\mathbf{L}}
\newcommand\kaff{\mathbf{A}}
\newcommand\kun{\mathbf{U}}
\newcommand\karr{\operatorname{\rightarrow}}
\newcommand\kvar{\kappa}

\newcommand\Sp{{sp}}

\newcommand\lub\bigvee
\newcommand\glb\bigwedge
\newcommand\Lat{\mathcal L}
\newcommand\CL{{\mathcal C_{\Lat}}}

\newcommand\SType{\ensuremath{\mathcal S}}

\newcommand\T[1]{\mathrm{#1}}
\newcommand\tvar{\alpha}
\newcommand\tcon{\T{T}}
\newcommand\schm{\sigma}
\newcommand\kschm{\theta}
\newcommand\tarr[1]{\operatorname{\xrightarrow{#1}}}
\newcommand\tapp[2]{\T{#1}\ #2}
\newcommand\qual[2]{#1 \operatorname{\Rightarrow} #2}
\newcommand\tyPair[3][{}]{#2 \times^{#1} #3}
\newcommand\tres{R}
\newcommand\tunit{\ensuremath{\mathtt{Unit}}}

\newcommand\generalize[3]{\operatorname{\text{gen}}(#1,#2,#3)}

\newcommand\unif{\mathrm{\psi}}
\newcommand\meet{\sqcup}

\newcommand\mostgeneral{\sqcup}

\newcommand\eletfun{\mathtt{letfun}}
\newcommand\ilam[5]{\lambda[#1 \mid #3 \Rightarrow #4]#5.}
\newcommand\ivar[3]{{#1}[{#2};{#3}]}
\newcommand\iapp[3]{({#2}\ {#3})_{#1}}
\newcommand\ilet[4]{\elet\ #2 =_{#1} #3\ \ein\ #4}
\newcommand\iletfun[7]{\eletfun\ \bvar{#2}{#3} =_{#1} \lam[{#4}]{#5} #6\ \ein\ #7}
\newcommand\ipair[4]{({#3},{#4})_{#1}^{#2}}
\newcommand\imatchin[5]{\ematch_{#2}\ #3 =_{#1} #4\ \ein\ #5}
\newcommand\aiapp[2]{\iapp{}{#1}{#2}}
\newcommand\aipair[3]{\ipair{}{#1}{#2}{#3}}
\newcommand\aimatchin[4]{\imatchin{}{#1}{#2}{#3}{#4}}

\newcommand\lam[2][{}]{\lambda^{#1} #2.}
\newcommand\elet{\mathtt{let}}
\newcommand\ematch{\mathtt{match}}
\newcommand\ein{\mathtt{in}}
\newcommand\letin[3]{\elet\ #1 = #2\ \ein\ #3}
\newcommand\matchin[4][\etransfm]{\ematch_{#1}\ #2 = #3\ \ein\ #4}
\newcommand\introPair[3][{}]{({#2},{#3})^{#1}}

\newcommand\app[2]{(#1\ #2)}
\newcommand\borrow[2][\BORROW]{{\&}^{#1}#2}
\newcommand\reborrow[2][\BORROW]{{\&\&}^{#1}#2}
\newcommand\borrowty[3][\BORROW]{{\&}^{#1}(#2,#3)}
\newcommand\region[3][n]{\{\!|#3|\!\}^{#1}_{#2}}
\newcommand\regionS[1]{\{\!|#1|\!\}}

\newcommand\sborrow[2][\BORROW]{#2{#1}}
\newcommand\etransfm{\phi}
\newcommand\transfm[1]{\etransfm(#1)}

\newcommand\create{\ensuremath{\mathtt{create}}}
\newcommand\rss[1]{[#1]}
\newcommand\observe{\ensuremath{\mathtt{observe}}}
\newcommand\update{\ensuremath{\mathtt{update}}}
\newcommand\destroy{\ensuremath{\mathtt{destroy}}}

\newcommand\Rannot[4][n]{#2 \rightsquigarrow_{#1} #3, #4}
\newcommand\RannotT[2]{#1 \rightsquigarrow #2}
\newcommand\getBorrows[3]{#1 \oplus #2 = #3}

\newcommand\Ctrue{\operatorname{True}}
\newcommand\Cfalse{\operatorname{False}}
\newcommand\Cempty{\cdot}

\newcommand\Cleq[2]{(#1 \leq #2)}
\newcommand\Ceq[2]{(#1 = #2)}
\newcommand\Cand{\wedge}
\newcommand\Cproj[2]{\exists #1.#2}
\newcommand\Weaken{\operatorname{Weak}}
\newcommand\Instance{\operatorname{Inst}}

\newcommand\Eempty{\cdot}
\newcommand\E{\Gamma}

\newcommand\bnone{\emptyset}
\newcommand\svar[3][\BORROW]{[#2 : #3]_{#1}}
\newcommand\bvar[2]{(#1 : #2)}
\newcommand\bbvar[4][\BORROW]{( #2 \div #4)_{#1}^{#3}}

\newcommand\lsplit[4]{#1 \vdash_e #2 = #3 \ltimes #4}
\newcommand\bsplit[4]{#1 \Lleftarrow #2 = #3 \ltimes #4}
\newcommand\lregion[5][n]{#2 \vdash_e #4 \rightsquigarrow_{#1}^{#3} #5}
\newcommand\bregion[5][n]{#2 \Lleftarrow #4 \rightsquigarrow_{#1}^{#3} #5}
\newcommand\fv[1]{\operatorname{fv}(#1)}
\newcommand\Eflat[1]{\Downarrow\!#1}

\newcommand\Sempty{\cdot}
\newcommand\Sone[2]{\left\{#1 \mapsto #2\right\}}
\newcommand\Sunion{\operatorname{\cup}}

\newcommand\Sv{\Sigma}
\newcommand\Sdel[1]{\!\!\mathbin{\setminus}\!\!\{#1\}}
\newcommand\Sonly[1]{\big|_{#1}}
\newcommand\Sadd[1]{\cup \{#1\}}
\newcommand\Sclos[1]{{{\downarrow}#1}}

\newcommand\BAR{\operatorname{|}}

\newcommand\entail[2]{#1 \operatorname{\vdash_e} #2}
\newcommand\equivC{\operatorname{=_e}}
\newcommand\inferS[4]{#1 \BAR #2 \operatorname{\vdash_s} #3 : #4}
\newcommand\inferW[5]{#1 \BAR #2 \BAR #3 \operatorname{\vdash_w} #4 : #5}
\newcommand\inferSK[4]{#1 \BAR #2 \operatorname{\vdash_s} #3 : #4}
\newcommand\inferK[4]{#1 \BAR #2 \operatorname{\vdash_w} #3 : #4}

\newcommand\inferSS[4]{#1 \BAR #2 \operatorname{\vdash_s} #3 \le #4}

\newcommand\normalize[2]{\operatorname{normalize}(#1,#2)}

\newcommand\Dom[1]{\operatorname{dom} (#1)}

\newcommand\subst[3]{#3[#1\rightarrow#2]}

\newcommand\IF[3]{\text{if } (#1) \text{ then } #2 \text{ else } #3}

\newcommand\MBORROW{\kaff}
\newcommand\IBORROW{\kun}
\newcommand\BORROW{\textit{b}}
\newcommand\BQ{\ensuremath\beta}
\newcommand\Addr\rho           
\newcommand\Loc\ell             

\newcommand\Multi[2][{}]{\overline{#2_{#1}}}

\newcommand\addlin[1]{#1}
\newcommand\inP[1]{\textcolor{ForestGreen}{\bm{#1}}}

\newcommand\CType[1]{\operatorname{\text{CType}} (#1)} 
\newcommand\IType[2]{\operatorname{\text{IType}} (\T{#1}, #2)} 
\newcommand\Bcompatible\nwarrow
\newcommand\VEnv\gamma          
\newcommand\Store\delta 
\newcommand\SE\Delta            
\newcommand\Perm\pi

\newcommand\StClosure[4]{\text{STCLOS}(#1, #2, #3, #4)}
\newcommand\StPClosure[6]{\text{STPOLY}(#1, #2, #3, #4, #5, #6)}
\newcommand\StPair[3]{\text{STPAIR}(#1, #2, #3)}
\newcommand\StRes[1]{\text{STRSRC}(#1)}
\newcommand\StFreed{\bullet}

\newcommand\TimeOut{\operatorname{TimeOut}}
\newcommand\Ok[1]{\operatorname{Ok} (#1)}
\newcommand\Nat{\mathbb{N}}
\newcommand\Rawloc[1]{\operatorname{getloc} (#1)}
\newcommand\Reach[2]{\operatorname{reach}_{#1} (#2)}
\newcommand\RS[2]{\operatorname{reach}_0 (#2)}

\newcommand\Linear[2]{\operatorname{lin}_{#1} (#2)}
\newcommand\Affine[2]{\operatorname{aff}_{#1} (#2)}
\newcommand\Unrestricted[2]{\operatorname{unr}_{#1} (#2)}
\newcommand\Crewrite{\leadsto}

\newcommand\REACH\Theta

\newcommand\MultiNumber[2]{#2(#1)}

\newcommand\LinPart[1]{{#1}^\klin}
\newcommand\AffPart[1]{{#1}^\kaff}
\newcommand\UnrPart[1]{{#1}^\kun}

\newcommand\Active[1]{\LinPart{#1}} 
\newcommand\MutableBorrows[1]{\AffPart{#1}} 
\newcommand\ImmutableBorrows[1]{\UnrPart{#1}} 
\newcommand\Suspended[2][{}]{{#2}^{#1}_{\#}}
\newcommand\Disjoint[1]{\operatorname{dis} (#1)}

\newcommand\Matches[2]{#1 \gets #2}        

\begin{document}

\title{Kindly Bent to Free Us}
\begin{abstract}
  Systems programming often requires the manipulation of resources like
  file handles, network connections, or dynamically allocated
  memory. Programmers need to follow certain protocols to handle
  these resources correctly. Violating these protocols causes bugs
  ranging from type mismatches over data races to use-after-free
  errors and memory leaks. These bugs often lead to security  vulnerabilities.

  While statically typed programming languages guarantee type soundness and memory
  safety by design, most of them do not address issues arising
  from improper handling of resources.
  An important step towards handling resources is the adoption of 
  linear and affine types that enforce single-threaded resource usage.
  However, the few languages supporting such types require heavy type
  annotations. 

  We present \lang, an extension of ML that
  manages linearity and affinity properties using kinds and
  constrained types. In addition
  \lang{} supports the exclusive and shared borrowing of affine
  resources, inspired by features of Rust.
  Moreover, \lang{} retains the defining features of the ML family:
  it is an impure, strict, functional expression language with complete principal type
  inference and type abstraction.
  \lang{} does not require any linearity annotations in
  expressions and supports common functional programming idioms.
\end{abstract}

\maketitle

\newcommand\ruleTimeOut{%
  \inferrule[TimeOut]{}{\Store, \Perm, \VEnv \vdash e \Downarrow^0 \TimeOut}
}

\newcommand\ruleSConst[1][i]{%
  \inferrule[SConst]{}{ \Store, \Perm, \VEnv \vdash c \Downarrow^{#1+1} \Ok{\Store, \Perm, c}}
}

\newcommand\ruleSVar[1][i]{%
  \inferrule[SVar]{}{\Store, \Perm, \VEnv \vdash x \Downarrow^{#1+1} \Ok{\Store, \Perm, \VEnv(x)}}
}

\newcommand\ruleSTApp[1][i]{%
  \inferrule[STApp]{
    \Matches \Loc { \VEnv (x)} \\
    \Loc \in \Perm \\
    \Matches {(\VEnv, \ilam {\Multi[i]{\kvar}}{\Multi[j]{\tvar}}Ckx{e})}{ \Store (\Loc)}\\
    \Perm' =  \IF{\entail C {k \le \kun_\infty}}{ \Perm}{\Perm\Sdel\Loc} \\
    \Loc'\notin\Dom{\Store}  \\
    \Store' = \Store[\Loc' \mapsto (\VEnv, \subst{\Multi[j]{\tvar}}{\Multi[j]{t}}{\subst {\Multi[i]{\kvar}}{\Multi[i]{k}}{(\lam[k]xe)}}) ]
  }{\Store, \Perm, \VEnv \vdash  \ivar x{\Multi[i]{k}}{\Multi[j]{\tau}}
    \Downarrow^{#1+1} \Ok{\Store', \Perm'\Sadd{\Loc'}, \Loc'}
  }
}

\newcommand\ruleSPLam[1][i]{%
  \inferrule[SPLam]{
    \Loc'\notin\Dom\Store \\
    \Store' = \Store[\Loc' \mapsto (\VEnv, \ilam
    {\Multi[i]{\kvar}}{\Multi[j]{\tvar}}Ck xe)] \\
    \Perm' = \Perm\Sadd{\Loc'}
  }{
    \Store, \Perm, \VEnv \vdash
    \ilam {\Multi[i]{\kvar}}{\Multi[j]{\tvar}}Ck xe
    \Downarrow^{#1+1} \Ok{ \Store', \Perm', \Loc'}
  }
}

\newcommand\ruleSApp[1][i]{%
  \inferrule[SApp]{
    \Store, \Perm, \VEnv \vdash e_1
    \Downarrow^{#1} \Ok{\Store_1, \Perm_1, r_1} \\
    \Matches\Loc{ r_1} \\
    \Matches{ (\VEnv'',\lam[k]{x}{e})}{ \Store_1 (\Loc)}  \\\\
    \Perm_1' = \IF{\entail {} {k \le \kun}}{\Perm_1}{ \Perm_1\Sdel\Loc}\\
    \Store_1, \Perm_1', \VEnv \vdash e_2
    \Downarrow^{#1} \Ok{ \Store_2, \Perm_2, r_2} \\
    \Store_2, \Perm_2, \VEnv''[x\mapsto r_2] \vdash e \Downarrow^{#1}
    \Ok{\Store_3, \Perm_3, r_3}
  }{\Store, \Perm, \VEnv \vdash \app{e_1}{e_2}
    \Downarrow^{#1+1} \Ok{\Store_3,\Perm_3, r_3}
  }
}

\newcommand\ruleSLet[1][i]{%
  \inferrule[SLet]{
    \Store, \Perm, \VEnv \vdash e_1
    \Downarrow^{#1} \Ok{ \Store_1, \Perm_1, r_1} \\
    \Store_1, \Perm_1, \VEnv[x \mapsto r_1] \vdash e_2
    \Downarrow^{#1} \Ok{ \Store_2, \Perm_2, r_2}
  }{
    \Store, \Perm, \VEnv \vdash \letin{x}{e_1}{e_2}
    \Downarrow^{#1+1} \Ok{\Store_2, \Perm_2, r_2}
  }
}

\newcommand\ruleSPair[1][i]{%
  \inferrule[SPair]{
    \Store, \Perm, \VEnv \vdash e_1
    \Downarrow^{#1} \Ok{ \Store_1, \Perm_1, r_1} \\
    \Store_1, \Perm_1, \VEnv \vdash e_2
    \Downarrow^{#1} \Ok{\Store_2, \Perm_2, r_2} \\\\
    \Loc'\notin\Dom{\Store_2} \\
    \Store_2' = \Store_2[\Loc' \mapsto \introPair[k]{r_1}{ r_2}] \\
    \Perm_2' = \Perm_2 \Sadd{\Loc'}
  }{
    \Store, \Perm, \VEnv \vdash \introPair[k]{e_1}{e_2}
    \Downarrow^{#1+1}
    \Ok{\Store_2', \Perm_2', \Loc'}
  }
}

\newcommand\ruleSMatchLocation[1][i]{%
  \inferrule[SMatchLocation]{
    \Store, \Perm, \VEnv \vdash e
    \Downarrow^{#1} \Ok{ \Store_1, \Perm_1, r_1} \\
    \Matches{\Loc}{r_1}  \\
    \Matches{\introPair[k]{ \Addr_1}{\Addr_2}}{\Store' (\Loc)} \\
    \Perm_1' = \IF{\entail {} {k \le \kun}}{\Perm_1}{\Perm_1\Sdel\Loc} \\
    \Store_1, \Perm_1', \VEnv[x,y \mapsto \Addr_1, \Addr_2] \vdash e_2
    \Downarrow^{#1} \Ok{\Store_2, \Perm_2, r_2}
  }{
    \Store, \Perm, \VEnv \vdash \matchin[\text{id}]{x,y}{e_1}{e_2} \Downarrow^{#1+1}
    \Ok{\Store_2, \Perm_2,  r_2}
  }
}

\newcommand\ruleSMatchBorrow[1][i]{%
  \inferrule[SMatchBorrow]{
    \Store, \Perm, \VEnv \vdash e_1
    \Downarrow^{#1} \Ok{ \Store_1, \Perm_1, r_1} \\
    \Matches{\BORROW\Multi\BORROW\Loc}{r_1} \\
    \Matches{\introPair[k]{ \Addr_1}{\Addr_2}} {    \Store' (\Loc)} \\
    \Addr_1' = \Addr_1\BORROW \\
    \Addr_2' = \Addr_2\BORROW \\\\
    \Perm_1' = (\Perm'\Sdel{\Addr_1,\Addr_1}) \Sadd{\Addr'_2, \Addr'_2} \\
    \Store1, \Perm_1', \VEnv[x,y \mapsto \Addr'_1, \Addr'_2] \vdash e_2
    \Downarrow^{#1} \Ok{ \Store_2, \Perm_2, r_2} \\
    \Perm_2' = (\Perm_2 \Sdel{\Addr'_1, \Addr'_2}) \Sadd{\Addr_1,\Addr_2}
  }{
    \Store, \Perm, \VEnv \vdash \matchin[\&^\BORROW]{x,y}{e_1}{e_2} \Downarrow^{#1+1}
    \Ok{\Store_2, \Perm_2',  r_2}
  }
}

\newcommand\ruleSRegion[1][i]{%
  \inferrule[SRegion]{
    \Matches\Addr{\VEnv (x)} \\
    \Addr \in \Perm \\
    \Store, (\Perm \Sdel\Addr) \Sadd{\sborrow{\Addr}}, \VEnv \vdash e
    \Downarrow^{#1} \Ok{\Store', \Perm', r}
  }{
    \Store, \Perm, \VEnv \vdash \region{\Sone x \BORROW}{e}
    \Downarrow^{#1+1} \Ok{ \Store', (\Perm' \Sdel{\sborrow\Addr})\Sadd\Addr, r}
  }
}

\newcommand\ruleSBorrow[1][i]{%
  \inferrule[SBorrow]{
    \Matches\Addr{\VEnv (x)} \\ \sborrow\Addr \in \Perm
  }{
    \Store, \Perm, \VEnv \vdash \borrow{x}
    \Downarrow^{#1+1} \Ok{ \Store, \Perm, \sborrow\Addr}
  }
}

\newcommand\ruleSCreate[1][i]{%
  \inferrule[SCreate]{
    \Store, \Perm, \VEnv \vdash e
    \Downarrow^{#1} \Ok{ \Store', \Perm', r}\\
    \Loc\notin \Dom{\Store'} }{
    \Store, \Perm,\VEnv \vdash \app\create e
    \Downarrow^{#1+1} \Ok{\Store'[\Loc \mapsto \rss{r}], \Perm'\Sadd\Loc, \Loc}
  }
}

\newcommand\ruleSDestroy[1][i]{%
  \inferrule[SDestroy]{
    \Store, \Perm, \VEnv \vdash e
    \Downarrow^{#1} \Ok{ \Store', \Perm', \Loc} \\
    \Matches{\rss{r}}{\Store' (\Loc)}
  }{
    \Store, \Perm, \VEnv \vdash \app\destroy e \Downarrow^{#1+1}
    \Ok{\Store'[\Loc\mapsto \StFreed], \Perm'\Sdel\Loc, ()}
  }
}

\newcommand\ruleSObserve[1][i]{%
  \inferrule[SObserve]{
    \Store, \Perm, \VEnv \vdash e_1
    \Downarrow^{#1} \Ok{ \Store_1, \Perm_1, r_1} \\
    \Matches\Addr{r_1} \\
    \Matches{\IBORROW\Multi\IBORROW\Multi\MBORROW\Loc}\Addr \\
    \Addr \in \Perm_1 \\
    \Matches{\rss{r}}{\Store_1 (\Loc)}
  }{
    \Store, \Perm, \VEnv \vdash \app\observe e
    \Downarrow^{#1+1} \Ok{ \Store_1, \Perm_1, r}
  }
}

\newcommand\ruleSUpdate[1][i]{%
  \inferrule[SUpdate]{
    \Store, \Perm, \VEnv \vdash e_1
    \Downarrow^{#1} \Ok{ \Store_1, \Perm_1, r_1} \\
    \Matches\Addr{r_1} \\
    \Matches{ \MBORROW\Multi\MBORROW\Loc}\Addr \\
    \Store_1, \Perm_1, \VEnv \vdash e_2
    \Downarrow^{#1} \Ok{ \Store_2, \Perm_2, r_2} \\
    \Addr \in \Perm_2 \\
    \Matches{\rss{r}}{\Store_2 (\Loc)} \\
    \Store_2' = \Store_2[\Loc \mapsto \rss{r_2}]
  }{
    \Store, \Perm, \VEnv \vdash \app{\app\update{e_1}}{e_2}
    \Downarrow^{#1+1} \Ok{\Store_2', \Perm_2 \Sdel{\Addr},  ()}
  }
}


\newcommand\ruleInstance{
  \inferrule[Instance]{
    \schm = \forall \Multi[i]{\kvar} \forall (\Multi[j]{\tvar_j:k}).\
    \qual{C}{\tau} \\\\
    \unif = [\Multi[i]{\kvar_i\mapsto k},\Multi[j]{\tvar_j \mapsto \tau}] 
  }{
    \unif(C),\unif(\tau) = \Instance (\E, \schm)
  }
}

\newcommand\ruleSDConst{
  \inferrule[Const]
  {
    \entail C {\addlin{\Cleq{\E}{\kaff_\infty}}}
  }
  { \inferS{C}{\E}{c}{\CType c} }
}

\newcommand\ruleSDVar{%
  \inferrule[Var]
  { \bvar{x}{
      \schm
    }
    \in \E \\
    C_x , \tau_x = \Instance (\E, \schm) \\\\
    \entail C {C_x \Cand \addlin{\Cleq{\E\Sdel{x}}{\kaff_\infty}}}
  }
  { \inferS{C}{\E}{x}{\tau_x}
  }
}

\newcommand\ruleSDLam{%
  \inferrule[Abs]
  {
    \inferS{C}
    {\E;\bvar{x}{\tau_2}}{e}{\tau_1} \\\\
    \addlin{\entail{C}{\Cleq{\E}{k}}}
  }
  { \inferS{C}{\E}
    {\lam{x}{e}}{\tau_2\tarr{k}\tau_1} }
}

\newcommand\ruleSDPair{
  \inferrule[Pair]
  { \addlin{\lsplit{C}{\E}{\E_1}{\E_2}} \\\\
    \inferS{C}{\E_1}{e_1}{\tau_1} \\\\
    \inferS{C}{\E_2}{e_2}{\tau_2}
  }
  { \inferS{C}{\E}{\introPair{e_1}{e_2}}{\tyPair{\tau_1}{\tau_2}} }
}

\newcommand\ruleSDApp{%
  \inferrule[App]
  {
    \inferS{C}{\E_1}{e_1}{\tau_2 \tarr{k} \tau_1} \\
    \inferS{C}{\E_2}{e_2}{\tau'_2} \\\\
    \addlin{\lsplit{C}{\E}{\E_1}{\E_2}}\\
    \entail C {\Cleq{\tau_2'}{\tau_2}}
  }
  { \inferS{C}
    {\E}{\app{e_1}{e_2}}{\tau_1} }
}

\newcommand\ruleSDVApp{%
  \inferrule[VApp]
  {
    {\E_1} (x_1) = {\tau_2 \tarr{k} \tau'_1} \\
    {\E_2} (x_2) = {\tau'_2}\\
    \entail C {\Cleq{\tau_2'}{\tau_2}}
  }
  { \inferS{C}
    {\E}{\app{x_1}{x_2}}{\tau'_1} }
}

\newcommand\ruleSDRegion{%
  \inferrule[Region]
  { \svar x {\tau_x}^n \in \E \\
    \addlin{ \lregion{C}{x}{\E}{\E'} }\\\\
    \inferS{C}{\E'}{e}{\tau} \\
    \entail C {\Cleq{\tau}{\klin_{n-1}}} \\
  }  { \inferS{C}{\E}{\region{\Sone x \BORROW}{e}}{\tau} }
}

\newcommand\ruleSDBorrow{
  \inferrule[Borrow]
  { \bbvar x k\schm \in \E \\
    C_x, \tau_x = \Instance (\E, \schm) \\\\
    \entail C {C_x \Cand \addlin{\Cleq{\E\Sdel{x}}{\kaff_\infty}}} \\
  }
  { \inferS{C}{\E}{\borrow{x}}{\borrowty{k}{\tau_x}} }
}
\newcommand\ruleSDReBorrow{
  \inferrule[Reborrow]
  { \inferS{C}{\E}{x}{\borrowty{k}{\tau}} }
  { \inferS{C}{\E}{\reborrow{x}}{\borrowty{k}{\tau}} }
}
\newcommand\ruleSDCreate{
  \inferrule[Create]
  { \inferSK{C}{\E}{\tau}{k} \\\\
    \entail C {\Cleq k {\kun_0} \Cand \addlin{\Cleq{\E}{\kaff_\infty}}} }
  { \inferS{C}{\E}{\create}{\tau \tarr{} \tapp\tres\tau } }
}
\newcommand\ruleSDObserve{
  \inferrule[Observe]
  { \inferSK{C}{\E}{\tau}{k} \\\\
    \entail C {\Cleq k {\kun_0}  \Cand \addlin{\Cleq{\E}{\kaff_\infty}}} \\
  }
  { \inferS{C}{\E}{\observe}{\borrowty[\IBORROW]{k'}{\tapp\tres\tau} \tarr{} \tau} }
}
\newcommand\ruleSDUpdate{
  \inferrule[Update]
  { \inferSK{C}{\E}{\tau}{k} \\
    \entail C {\Cleq k {\kun_0}  \Cand \addlin{\Cleq{\E}{\kaff_\infty}}} \\
  }
  { \inferS{C}{\E}{\update}{\borrowty[\MBORROW]{k'}{\tapp\tres\tau} \tarr{} \tau \tarr{\kaff} \tunit } }
}
\newcommand\ruleSDDestroy{
  \inferrule[Destroy]
  { \inferSK{C}{\E}{\tau}{k} \\\\
    \entail C {\Cleq k {\kun_0}  \Cand \addlin{\Cleq{\E}{\kaff_\infty}}} \\
  }
  { \inferS{C}{\E}{\destroy}{\tapp\tres\tau \tarr{} \tunit} }
}

\newcommand\ruleSDLet{
  \inferrule[Let]
  { \inferS{C \Cand D}{\E_1}{e_1}{\tau_1} \\
    (C_\schm,\schm) = \generalize{D}{\E}{\tau_1}\\
    \entail{C}{C_\schm} \\\\
    \inferS{C}{\E;\bvar{x}{\schm}}{e_2}{\tau_2} \\
    \addlin{\lsplit{C}{\E}{\E_1}{\E_2}}\\
  }
  { \inferS{C}
    {\E}{\letin{x}{e_1}{e_2}}{\tau_2} }
}

\newcommand\ruleSDMatchPair{
  \inferrule[MatchPair]
  {
    \inferS{C}{\E_1}{e_1}{\transfm{\tyPair{\tau_1}{\tau'_1}}} \\\\
    \inferS{C}
    {\E_2;
      \bvar{x}{\transfm{\tau_1}};
      \bvar{x'}{\transfm{\tau'_1}}}
    {e_2}{\tau_2} \\\\
    \addlin{\lsplit{C}{\E}{\E_1}{\E_2}}
  }
  { \inferS{C}
    {\E}{\matchin{x,x'}{e_1}{e_2}}{\tau_2} }
}


\newcommand\ruleSDIVar{%
  \inferrule[Var] {
    \bvar{x}{\tau} \in \E \\\\
    \entail C {\addlin{\Cleq{\E\Sdel{x}}{\kaff_\infty}}}
  } {
    \inferS{C}{\E}{x}{\tau}
  }
}

\newcommand\ruleSDIVarInst{%
  \inferrule[VarInst]
  { \bvar{x}{
      \forall \Multi[i]{\kvar} \forall (\Multi[j]{\tvar_j:k}).\ \qual{C_x}{\tau}
    }
    \in \E \\\\
    \unif = [\Multi[i]{\kvar_i\mapsto k},\Multi[j]{\tvar_j \mapsto \tau}] \\\\
    \entail C {\unif(C_x) \Cand \addlin{\Cleq{\E\Sdel{x}}{\kaff_\infty}}}
  }
  { \inferS{C}{\E}{\ivar{x}{\Multi[i]k}{\Multi[j]\tau}}{\unif\tau}
  }
}

\newcommand\ruleSDILam{%
  \inferrule[Abs] {
    \inferS{C} {\E;\bvar{x}{\tau_2}}{e}{\tau_1} \\
    \addlin{\entail{C}{\Cleq{\E}{k}}}
  } {
    \inferS
      {C}
      {\E}
      {\lam[k]xe}
      {\tau_2\tarr{k}\tau_1}
  }
}

\newcommand\ruleSDIPair{
  \inferrule[Pair]
  { \Sp : \addlin{\lsplit{C}{\E}{\E_1}{\E_2}} \\\\
    \inferS{C}{\E_1}{e_1}{\tau_1} \\\\
    \inferS{C}{\E_2}{e_2}{\tau_2}
  }
  { \inferS{C}{\E}{\ipair\Sp{k}{e_1}{e_2}}{\tyPair{\tau_1}{\tau_2}} }
}

\newcommand\ruleSDIApp{%
  \inferrule[App]
  {
    \Sp : \addlin{\lsplit{C}{\E}{\E_1}{\E_2}}\\\\
    \inferS{C}{\E_1}{e_1}{\tau_2 \tarr{k} \tau_1} \\\\
    \inferS{C}{\E_2}{e_2}{\tau'_2} \\\\
    \entail C {\Cleq{\tau_2'}{\tau_2}}
  }
  { \inferS{C}
    {\E}{\iapp\Sp{e_1}{e_2}}{\tau_1} }
}

\newcommand\ruleSDICreate{
  \inferrule[Create]
  { \inferSK{C}{\E}{\tau}{k} \\
    \entail C {\Cleq k {\kun_0}} }
  { \inferS{C}{\E}{\create}{\tau \tarr{} \tapp\tres\tau } }
}
\newcommand\ruleSDIObserve{
  \inferrule[Observe]
  { \inferSK{C}{\E}{\tau}{k} \\
    \entail C {\Cleq k {\kun_0}} \\
    \kvar \text{ fresh}
  }
  { \inferS{C}{\E}{\observe}{\borrowty[\IBORROW]{\kvar}{\tapp\tres\tau} \tarr{} \tau} }
}
\newcommand\ruleSDIUpdate{
  \inferrule[Update]
  { \inferSK{C}{\E}{\tau}{k} \\
    \entail C {\Cleq k {\kun_0}} \\
    \kvar \text{ fresh}
  }
  { \inferS{C}{\E}{\update}{\borrowty[\MBORROW]{\kvar}{\tapp\tres\tau} \tarr{} \tau \tarr{\kaff} \tunit } }
}
\newcommand\ruleSDIDestroy{
  \inferrule[Destroy]
  { \inferSK{C}{\E}{\tau}{k} \\
    \entail C {\Cleq k {\kun_0}} }
  { \inferS{C}{\E}{\destroy}{\tapp\tres\tau \tarr{} \tunit} }
}

\newcommand\ruleSDILet{
  \inferrule[Let] {
    \Sp : \addlin{\lsplit{C}{\E}{\E_1}{\E_2}}\\
    \inferS{C}{\E_1}{e_1}{\tau_1} \\
    \inferS{C}{\E;\bvar{x}{\tau_1}}{e_2}{\tau_2} \\
  } {
    \inferS{C}{\E}{\ilet\Sp{x}{e_1}{e_2}}{\tau_2}
  }
}

\newcommand\ruleSDIPLet{
  \inferrule[PLet] {
    \Sp : \addlin{\lsplit{C}{\E}{\E_1}{\E_2}} \\
    \schm_1 =
      \forall \Multi[i]\kvar
              \bvar{\Multi[j]\tvar}{\Multi[j]k}.\
      \qual{D}{\tau_2\tarr{k}\tau_1} \\
    \inferS
      {C \Cand D}
      {\E_1; \bvar{\Multi[j]\tvar}{\Multi[j]k}; \bvar{x}{\tau_2}}
      {e_1}
      {\tau_1} \\
    \addlin{\entail{C \Cand D}{\Cleq{\E_1}{k}}} \\
    \entail{C}{\Cproj{(\Multi[i]{\kvar},\Multi[j]{\tvar})}{D}}\\
    \inferS{C}{\E;\bvar{f}{\schm_1}}{e_2}{\tau_2} \\
  } {
    \inferS
      {C}
      {\E}
      {\iletfun\Sp{f}{\schm_1}{k}{x}{e_1}{e_2}}
      {\tau_2}
  }
}

\newcommand\ruleSDIMatchPair{
  \inferrule[MatchPair] {
    \Sp : \addlin{\lsplit{C}{\E}{\E_1}{\E_2}} \\
    \inferS
      {C}
      {\E_1}
      {e_1}
      {\transfm{\tyPair{\tau_1}{\tau'_1}}} \\
    \inferS
      {C}
      {\E_2; \bvar{x}{\transfm{\tau_1}}; \bvar{y}{\transfm{\tau'_1}}}
      {e_2}
      {\tau_2} \\
  } {
    \inferS
      {C}
      {\E}
      {\imatchin\Sp\etransfm{x,x'}{e_1}{e_2}}
      {\tau_2}
  }
}


\newcommand\ruleSDAIPair{
  \inferrule[Pair] {
    \bvar{x_1}{\tau_1} \in \Gamma \\\\
    \bvar{x_2}{\tau_2} \in \Gamma \\\\
    \entail C {\addlin{\Cleq{\E\Sdel{x_1,x_2}}{\kaff_\infty}}}
  }
  { \inferS{C}{\E}{\aipair{k}{x_1}{x_2}}{\tyPair{\tau_1}{\tau_2}} }
}

\newcommand\ruleSDAIApp{%
  \inferrule[App] {
    \bvar{x_1}{\tau_2 \tarr{k} \tau_1} \in \Gamma \\\\
    \bvar{x_2}{\tau_2'} \in \Gamma \\\\
    \entail C {\Cleq{\tau_2'}{\tau_2}} \\\\
    \entail C {\addlin{\Cleq{\E\Sdel{x_1,x_2}}{\kaff_\infty}}}
  } {
    \inferS{C}{\E}{\aiapp{x_1}{x_2}}{\tau_1}
  }
}

\newcommand\ruleSDAIMatchPair{
  \inferrule[MatchPair] {
    \Sp : \addlin{\lsplit{C}{\E}{\E_1}{\E_2}} \\
    \E_1 = \bvar z {\transfm{\tyPair{\tau_1}{\tau_1'}}}  \\
    \inferS
      {C}
      {\E_2; \bvar{x}{\transfm{\tau_1}}; \bvar{x'}{\transfm{\tau_1'}}}
      {e_2}
      {\tau_2} \\
  } {
    \inferS
      {C}
      {\E}
      {\aimatchin\etransfm{x,x'}{z}{e_2}}
      {\tau_2}
  }
}


\newcommand\ruleIVar{%
  \inferrule[Var$_I$]
  { \bvar{x}{\forall \Multi[i]{\kvar} \forall (\Multi[j]{\tvar_j:k}).\
    \qual{C}{\tau}}\in \E \\
    \Multi[i]{\kvar'},\Multi[j]{\tvar'} \text{ fresh} \\\\
    (C,\unif) = \normalize{C_x}{[\Multi[i]{\kvar_i\mapsto \kvar'},\Multi[j]{\tvar_j \mapsto \tvar'}]}
  }
  { \inferW
    {\addlin{\bvar{x}{\sigma}}}
    {(C,\unif|_{\fv{\E}})}{\E}{x}{\unif\tau} }
}
\newcommand\ruleIAbs{%
  \inferrule[Abs$_I$]
  { \tvar,
    \kvar\text{ fresh}\\
    \inferW{\Sv_x}{(C',\unif')}
    {\E;\bvar{x}{\tvar}
    }{e}{\tau} \\\\
    D = C'\Cand
    \addlin{\Cleq{\Sv_x \Sdel{x}}{\kvar} \Cand \Weaken_{\bvar{x}{\tvar}}(\Sv_x)} \\\\
    (C,\unif) = \normalize{D}{\unif'}
  }
  { \inferW{\addlin{\Sv_x \Sdel{x}}}{(C,\unif\Sdel{\tvar,\kvar})}{\E}
    {\lam{x}{e}}{\unif(\tvar)\tarr{\unif(\kvar)}\tau} }
}
\newcommand\ruleIApp{%
  \inferrule[App$_I$]
  { \tvar,\kvar\text{ fresh}\\
    \inferW{\Sv_1}{(C_1,\unif_1)}{\E}{e_1}{\tau_1} \\\\
    \addlin{\bsplit{C_s}{\Sv}{\Sv_1}{\Sv_2}}\\
    \inferW{\Sv_2}{(C_2,\unif_2)}{\E}{e_2}{\tau_2} \\\\
    D =
    C_1 \Cand C_2 \Cand \Cleq{\tau_1}{\tau_2\tarr{\kvar}\tvar}
    \Cand \addlin{C_s} \\\\
    \unif' = \unif_1 \mostgeneral \unif_2 \\
    (C,\unif) = \normalize{D}{\unif'}\\
  }
  { \inferW{\addlin{\Sv}}{(C,\unif)}
    {\E}{\app{e_1}{e_2}}{\unif(\tvar)} }
}
\newcommand\ruleILet{%
  \inferrule[Let$_I$]
  { \inferW{\Sv_1}{(C_1,\unif_1)}{\E}{e_1}{\tau_1} \\
    (C_\schm,\sigma) = \generalize{C_1}{\unif_1\E}{\tau_1} \\\\
    \inferW{\Sv_2}{(C_2,\unif_2)}{\E;\bvar{x}{\sigma}}{e_2}{\tau_2} \\
    \addlin{\bsplit{C_s}{\Sv}{\Sv_1}{\Sv_2 \Sdel{x}}}\\
    \unif' = \unif_1 \mostgeneral \unif_2 \\\\
    D =
    C_\schm \Cand C_2 \Cand
    \addlin{C_s \Cand \Weaken_{\bvar{x}{\sigma}}(\Sv_2)}  \\
    (C,\unif) = \normalize{D}{\unif'}\\
  }
  { \inferW{\addlin{\Sv}}{(C,\unif|_{\fv{\E}})}
    {\E}{\letin{x}{e_1}{e_2}}{\unif\tau_2} }
}
\newcommand\ruleIPair{%
  \inferrule[Pair$_I$]
  { \inferW{\Sv_1}{(C_1,\unif_1)}{\E}{e_1}{\tau_1} \\
    \inferW{\Sv_2}{(C_2,\unif_2)}{\E}{e_2}{\tau_2} \\
    \unif' = \unif_1 \mostgeneral \unif_2 \\\\
    \addlin{\bsplit{C_s}{\Sv}{\Sv_1}{\Sv_2}}\\
    D =
    C_1 \Cand C_2 \Cand \addlin{C_s} \\
    (C,\unif) = \normalize{D}{\unif'}\\
  }
  { \inferW{\Sv}{(C,\unif)}{\E}{\introPair{e_1}{e_2}}{\tyPair{\tau_1}{\tau_2}} }
}
\newcommand\ruleIMatch{%
  \inferrule[MatchPair$_I$]
  { \tvar,\kvar,\tvar',\kvar'\text{ fresh}\\
    \inferW{\Sv_1}{(C_1,\unif_1)}{\E}{e_1}{\tau_1} \\
    \E' = \E;
    \bvar{x}{\transfm{\tvar}};\bvar{\tvar}{\kvar};
    \bvar{x'}{\transfm{\tvar'}};\bvar{\tvar'}{\kvar'}\\
    \inferW{\Sv_2}{(C_2,\unif_2)}
    {\E'}{e_2}{\tau_2} \\
    \unif' = \unif_1 \mostgeneral \unif_2 \\
    \addlin{\bsplit{C_s}{\Sv}{\Sv_1}{(\Sv_2 \Sdel{x,x'})}}\\
    D =
    C'_1 \Cand C_2 \Cand \Cleq{\tau_1}{\transfm{\tyPair{\tvar}{\tvar'}}}
    \Cand
    \addlin{C_s
      \Cand \Weaken_{\bvar{x}{\transfm\tvar},\bvar{x'}{\transfm\tvar'}}(\Sv_2)} \\
    (C,\unif) = \normalize{D}{\unif'}\\
  }
  { \inferW{\addlin{\Sv}}{(C,\unif|_{\fv{\E}})}
    {\E}{\matchin{x,x'}{e_1}{e_2}}{\unif\tau_2} }
}
\newcommand\ruleIBorrow{%
  \inferrule[Borrow$_I$]
  { \inferW{\_}{(C,\unif)}{\E}{x}{\tau} \\
    \kvar \text{ fresh}
  }
  { \inferW
    {\addlin{\bbvar x \kvar {\tau}}}
    {(C,\unif)}{\E}{\borrow{x}}{\borrowty{\kvar}{\tau}} }
}
\newcommand\ruleIReBorrow{%
  \inferrule[ReBorrow$_I$]
  { \inferW{\Sv}{(C',\unif')}{\E}{x}{\tau'} \\
    \kvar \text{ fresh} \\\\
    (C,\unif) = \normalize{C' \Cand \Cleq{\tau'}{\borrowty{\kvar}{\tau}}}{\unif'}\\
  }
  { \inferW
    {\addlin{\bbvar x \kvar \tau}}
    {(C,\unif)}{\E}{\borrow{x}}{\borrowty{\kvar}{\tau}} }
}
\newcommand\ruleIRegion{%
  \inferrule[Region$_I$]
  { \inferW{\addlin{\Sv'}}{(C',\unif')}{\E}{e}{\tau} \\
    \addlin{ \bregion{C_r}{x}{\Sv}{\Sv'} }\\\\
    \inferK{(C_\tau,\unif_\tau)}{\E}{\tau}{k_\tau}\\\\
    D = C' \Cand C_\tau \Cand \Cleq{k_\tau}{\klin_{n-1}} \Cand C_r\\\\
    (C,\unif) = \normalize{D}{\unif' \mostgeneral \unif_\tau}\\
  }  { \inferW{\addlin{\Sv}}{(C,\unif)}{\E}{\region{\Sone x \BORROW}{e}}{\tau} }
}

\newcommand\ruleResultConstant{%
  \inferrule{}{ \SE \vdash c : \CType{c} }
}

\newcommand\ruleResultLocation{%
  \inferrule{}{ \SE \vdash \ell : \SE (\ell) }
}

\newcommand\ruleResultBorrow{%
  \inferrule{
    \Multi\BQ \Bcompatible \BORROW_n \\
    \SE \vdash \Loc  : \tau
  }{  \SE \vdash
    \Multi\BQ\Loc : \borrowty[\BORROW]{\BORROW_n}{\tau}}
}

\newcommand\ruleStorableFreed{%
  \inferrule{}{
    \SE \vdash \StFreed : \tau
  }
}

\newcommand\ruleStorableResource{%
  \inferrule{
    \SE \vdash r : \IType{\tcon}{\Multi\tau}
  }{
    \SE \vdash {[r]} : \tapp{\tcon}{\Multi\tau}
  }
}

\newcommand\ruleStorablePair{%
  \inferrule{
    \SE \vdash r_1 : \tau_1 \\
    \SE \vdash r_2 : \tau_2 \\
    \entail\Cempty\Cleq{\tau_1}{k} \Cand \Cleq{\tau_1}{k}
  }{
    \SE \vdash \introPair[k]{r_1}{r_2} : \tyPair[k]{\tau_1}{\tau_2}
  }
}

\newcommand\ruleStorableClosure{%
  \inferrule{
    (\exists \E, C)~ \SE \vdash \VEnv : \E
    \\
    \Disjoint\E
    \\
    \inferS{C}{\E;\bvar x{\tau_2}}{e}{\tau_1}
    \\
    \addlin{\entail{C}{\Cleq{\E}{k}}}
  }{
    \SE \vdash (\VEnv, \lam[k]xe) : \tau_2\tarr{k}\tau_1
  }
}


\section{Introduction}

A large proportion of systems programming is focused on the proper
handling of resources, like file handles, network connections, or
dynamically allocated memory. Each of these resources comes with a
protocol that prescribes the correct use of its API.
For examples, a file handle appears as the result of opening a
file. If it was opened for reading, then read operations will succeed,
but write operations will fail. Once the handle is closed, it cannot
be used for reading or writing, anymore.
Dynamic allocation of memory is similar. An API call returns a
pointer to a memory area, which can then be read and written to until
the area is released by another API call.

In both cases, a resource is created in a certain state and a resource
handle is returned to the program. Depending on this state, certain API calls
can safely be applied to it. Finally, there is another API call to
release the resource, which renders the handle invalid.
Taken to the extreme, each API call changes the state so that a
different set of API calls is enabled afterwards.
Ignoring such life cycle protocols is a common source of errors.

Most type systems provide type soundness and memory safety, but neglect the
protocol aspect. Systems that can support reasoning about protocols
build on linear types \cite{DBLP:journals/tcs/Girard87} and/or
uniqueness types~\cite{DBLP:conf/plilp/BarendsenS95}. A value of linear
type is guaranteed to be consumed
exactly once. That is, a file that has been opened must be closed and
memory that has been allocated must be released. A value of unique
type is guaranteed to have a single reference to it. Thus, memory can
be reused on consumption of the value.

These systems work well if one is prepared to write programs
functionally in resource-passing style. In this style, all operations
in the resource's API take the resource as a parameter and return it
in possibly modified state~\cite{DBLP:journals/jfp/AchtenP95}. In
typestate-oriented programming, they would also modify its
type~\cite{DBLP:conf/oopsla/AldrichSSS09}. Functional session types
represent a popular
example~\cite{DBLP:journals/jfp/GayV10,lindley17:_light_funct_session_types}.

Explicit resource passing places a heavy burden on the programmer and
complicates the program structure. For imperative APIs,
resource-passing style is not an option at all. To this end,
\citet{DBLP:conf/popl/BoylandR05}  proposed the notion of
\emph{borrowing} a resource. The idea is that a linear resource can be
borrowed to a function call. The function can work with a borrow of
the resource, but it cannot release the resource. Only the original
owner of the resource has all rights to it and can release it.

The concepts of ownership and borrowing have grown popular over time
and they form the foundation of the type system of the Rust language
\cite{rust}, which considers any memory-allocated data structure a
resource. Rust supports two kinds of borrows,
shared and exclusive ones.
Exclusive borrows enable modification of the data structure
whereas shared borrows only grant read access.
At any given time, either a single exclusive borrow is active or
any number of shared borrows can be active.
Moreover, Rust makes sure that the lifetime of a
borrow is properly contained in the lifetime of its lender.

The design of Rust is geared towards programmers with a low-level
imperative programming background, like C or C++. Its management of
lifetimes supports the manual way of memory management customary in
these languages very well and makes it safe. However, programmers with
a background in managed languages feel alienated from the lack of garbage
collected data. They would prefer a setting where automatic memory
management with garbage collection was the default, but where they
could seemlessly switch to safe, manual resource management if that
was required.
As a concrete example, consider a functional programmer who wants to
safely interact with a C library. Invoking a C function is easy via
the existing foreign function interface, but managing the underlying
resources like malloc'd storage is not: it cannot be left to the
garbage collector, but proper release of the storage via calls to \texttt{free()} must be
ensured by programming conventions.

Our work provides a safe solution to programmers in this situation. We
propose an extended type system for ML-like languages that comes with
linear and affine types, exclusive and shared borrows, but all that
integrated with full principal type inference, garbage collected data, and
automatic placement of borrowing regions.
In our system, it is a type error to omit the call to release the
storage given a suitably typed API for storage allocation.

The most closely related contenders in this design space are Linear Haskell
\cite{DBLP:journals/pacmpl/BernardyBNJS18}, henceforth LH,
Quill \cite{DBLP:conf/icfp/Morris16}, and ALMS \cite{DBLP:conf/popl/TovP11}.
Compared to LH and Quill, the goals and means are
similar as these systems also permit abstraction over the number of uses of
values and retain type inference, but the details are different.
\begin{enumerate}
\item Multiplicities in LH and Quill are either linear or unrestricted whereas
  we also distinguish affine values.
\item In \lang{} and in Quill multiplicities are directly attached to the type of a
  value. For example, in \lang{} the function type \lstinline/'a-{lin}>'b/
  denotes the type of a \emph{single-use function} that can be called
  just once, whereas the multiplicities in LH choose
  between $\alpha\to\beta$ and $\alpha \multimap\beta$ where the
  latter is a function that promises to \emph{use its argument exactly
    once}.
\item \lang{} makes use of multiplicity contraints (like Quill) and kind
  subsumption (unlike Quill). Kind subsumption results in significantly simpler, more readable
  inferred types.
\item Neither LH nor Quill have borrowing whereas \lang{} supports
  two flavors: affine (exclusive, mutable) and unrestricted (shared, immutable) borrows.
\end{enumerate}
See \cref{sec:related-work} for further in-depth discussion of these
and other related works.

\subsection{First examples}
\label{sec:first-example}

\lstMakeShortInline[keepspaces,basicstyle=\small\ttfamily]@

\begin{figure}[tp]
  \begin{subfigure}[t]{0.35\linewidth}
\begin{lstlisting}
module File : sig
  type t: lin
  val fopen: path -> t
  val write: &!t -> string -{aff}> unit
  val close: t -> unit
end
\end{lstlisting}
    \vspace{-15pt}
    \caption{File API}
    \label{fig:writing-files-api}
  \end{subfigure}~
  \begin{subfigure}[t]{0.3\linewidth}
\begin{lstlisting}
let main () =
  let h = File.fopen "foo" in
  File.write &!h "Hello ";(*@\label{line:hello}*)
  File.write &!h "world!";(*@\label{line:world}*)
  File.close h
\end{lstlisting}
    \vspace{-10pt}
    \caption{File example}
    \label{fig:writing-files-example}
  \end{subfigure}~
  \begin{subfigure}[t]{0.3\linewidth}
\begin{lstlisting}
let main () =
  let h = File.fopen "foo" in
  {| File.write &!h "Hello " |};(*@\label{line:hello}*)
  {| File.write &!h "world!" |};(*@\label{line:world}*)
  File.close h
\end{lstlisting}
    \vspace{-10pt}
    \caption{File example with regions}
    \label{fig:writing-files-example-region}
  \end{subfigure}
  \vspace{-10pt}
  \caption{Writing files}
  \label{fig:writing-files}
\end{figure}

As a first, well-known example we consider a simplified API for
writing files shown in \cref{fig:writing-files-api}.  It introduces a
linear abstract type @File.t@. A call like
@File.fopen "foo"@ returns a linear handle to a newly created
file, which \emph{must} be released later on with @File.close@
as shown in \cref{fig:writing-files-example}. Failing to do so is a
static type error.  To write to the file, we must take an exclusive
borrow @&!h@ of the handle and pass it to the
@File.write@ function. Exclusive borrows are affine:
they must not be duplicated, but there is no requirement to use
them. This affinity shows up in the annotation @-{aff}>@ of the second arrow in
the type of @File.write@: a partial application like
@File.write &!h@ captures the affine borrow and hence the
resulting function is also affine. It
would be an error to use the affine closure twice as in
\begin{lstlisting}[numbers=none]
let w = File.write &!h in w "Hello "; w "world!" (*type error*)
\end{lstlisting}
The remaining arrows in
the API are unrestricted and we write @->@ instead of the
explicitly annotated @-{un}>@.

Every borrow is restricted to a \emph{region}, i.e., a lexically
scoped program fragment from which the borrow must not escape. In
\cref{fig:writing-files-example}, there are two regions visualized in \cref{fig:writing-files-example-region}, one
consisting of \cref{line:hello} and another consisting of
\cref{line:world}. Both are fully contained in the scope of the linear
handle @h@, hence we can take one exclusive borrow @&!h@ in each
region. In both regions the borrow is consumed immediately by passing
it to @File.write@. \lang elaborates regions automatically before type
inference. Alternatively, programmers may mark regions explicitly (See
\cref{sec:imper-progr}).

This example demonstrates three features of our system:
\begin{enumerate}
\item type and region inference without annotations in user code (\cref{fig:writing-files-example}),
\item types carry multiplicity annotations in the form of kinds,
\item resource APIs can be written in direct style as linearity is a
  property of the type @File.t@.
\end{enumerate}

Direct style means that there is a function like @fopen@ that creates
and returns a linear resource. In contrast, LH forces programmers to
use resource-passing style because, in LH, linearity is a property of
a function, rather than a property of a value that restricts the way
that value can be handled (as in \lang). An LH API analogous to @File@
might provide functions like
\begin{itemize}
\item @withFile : path -> (handle o-> Unrestricted r) -> r@, which creates a new
  file handle and takes a continuation that uses the @handle@
  linearly, but returns an unrestricted value\footnote{For technical
    reasons, LH  requires the programmer to use a type like
    \lstinline/Unrestricted/ at this  point.},
\item @writeFile : string -> handle o-> handle@, which
  returns the transformed resource @handle@, and
\item @closeFile : handle o-> unit@, which consumes the @handle@ by
  closing the file.
\end{itemize}

In general, kinds can be polymorphic and constrained. Function
application and composition are the archetypical
examples for functions exploiting that feature.\footnote{Compared to
  Quill \cite{DBLP:conf/icfp/Morris16} the signatures of application and
  composition are simpler because \lang{} supports kind subsumption.}
For application, \lang{} infers the following type.
\begin{lstlisting}[numbers=none]
let app f x = f x
# app : ('a -{'k}> 'b) -> ('a -{'k}> 'b)
\end{lstlisting}
The reading of the inferred type is straightforward. If
@f@ is a @'k@-restricted function, then so is
@app f@. The multiplicities of @'a@ and
@'b@ play no role. As usual in ML-like languages, we
implicitly assume prenex quantification by
$\forall\kappa\forall\alpha\forall\beta$. Internally, the
type checker also quantifies over the kinds of $\alpha$ and $\beta$,
but the full prefix
$\forall\kappa\kappa_1\kappa_2\forall(\alpha:\kappa_1)\forall(\beta:\kappa_2)$
of the type of @app@ is only revealed as much as necessary for
understanding the type.

For @compose@, \lang{} infers this type.
\begin{lstlisting}[numbers=none]
let compose f g x = f (g x)
# compose : ('k <= 'k_1) => ('b -{'k}> 'c) -> ('a -{'k_1}> 'b) -{'k}> ('a -{'k_1}> 'c)
\end{lstlisting}
Like in @app@, the multiplicities of the type variables
@'a,'b,'c@ do not matter. However, the multiplicity
@'k@ of @f@ reappears on the second to last arrow
because @compose f@ is a closure that inherits
@f@'s multiplicity. The multiplicities of @g@  and
@f@ both influence the multiplicity of the last arrow, so
we would expect its annotation to be the least upper bound
$\kappa \sqcup \kappa_1$. Thanks to subsumption of multiplicities, it
is sufficient to assume $\kappa \le \kappa_1$ and @g@'s
actual multiplicity gets subsumed to $\kappa_1$. This constraint
simplification is part of our type inference algorithm. As before,
printing the type scheme only mentions the non-trivial constraint
$\kappa\le\kappa_1$ and omits the prenex quantification over $\kappa,
\kappa_1$ as well as the kinds of @'a,'b,'c@.

\lstDeleteShortInline@

\subsection{Contributions}
\label{sec:contributions}

\begin{itemize}
\item A polymorphic type system that encodes linearity and affinity with
  borrowing in lexical regions. Polymorphism covers types and kinds that
  express multiplicity constraints on the number of uses of a
  value. This type system is a conservative extension of systems for
  existing ML-like languages.
\item Expressive type soundness theorem with respect to a big-step linearity-aware semantics.
\item An extension of the \hmx framework
  \cite{DBLP:journals/tapos/OderskySW99} for constrained type
  inference to equip the type system with full, principal type inference.
\item Soundness proof of the inference algorithm.
\item Automatic inference of regions for borrows.
\item A prototype implementation of the type inference algorithm, including all
  constraint simplification and extended with algebraic datatypes and
  pattern matching,
  available at \url{https://affe.netlify.com/}.
\end{itemize}

As \lang{} is built on top of the \hmx{} framework, which is a general
framework for expressing constraint-based typing and type inference,
the extension of our work with features like typeclasses, ad-hoc overloading,
traits, etc is possible and orthogonal to the presentation in this paper.
While the system is geared towards type inference, it is nonetheless
compatible with type annotations and thereby amenable to extensions
where type inference may no longer be possible.


\section{Linearity, Affinity, and Borrows at Work}
\label{motivation}

\lstMakeShortInline[keepspaces,basicstyle=\small\ttfamily]@

\lang{} supports the resource-passing style
common in functional encodings of session types (e.g.,
\cite{DBLP:journals/jfp/Padovani17}; see also
\cref{sec:session-linearity} in the supplement) as well as other
functional resource handling. But it really shines
when manipulating mutable resources like buffers or connection pools
using a mix of functional and imperative programming styles.
To support this usage pattern of linearity, \lang{} relies on
the notion of borrowing \cite{DBLP:conf/popl/BoylandR05}. Our first
example of linear arrays demonstrates simple borrowing and
imperative programming; the second example demonstrates reborrowing
and the interaction between closures and borrowing by implementing a
Sudoko solver based on a hybrid copy-on-write data structure; the
third example demonstrates advanced uses 
of regions with iterators on linear values and the low-level
primitives needed to implement them.
Further examples are available in \cref{sec:extra-example}.

\subsection{Imperative programming with linear arrays}
\label{sec:imper-progr}

\begin{figure}[tp]
  \centering
\begin{lstlisting}[numbers=none]
module Array : sig
  type ('a : 'k) t : lin
  val create : ('a : un) => int * 'a -> 'a t
  val free : ('a : aff) => 'a t -> unit
  val length : &('a t) -> int
  val get : ('a : un) => &('a t) * int -> 'a
  val set : ('a : aff) => &!('a t) * int * 'a -> unit
  val map : (&'a -> 'b) * &('a t) -> 'b t
  val iter : ('a -> unit) * 'a t -> unit
end
\end{lstlisting}
  \vspace{-15pt}
  \caption{Linear arrays}
  \label{ex:array}
  \label{sig:array}
  \vspace{-10pt}
\end{figure}

The API for mutable linear arrays (\cref{sig:array})
aims to safely handle manual allocation and
deallocation of arrays that may contain affine or linear elements.
A program would first use @create (n, v)@ to create
an array of size @n@ initialized with value
@v@. The value @v@ must be unrestricted as it is duplicated to
initialize all array elements.
To @free@ an array the elements must be affine. Thanks to subkinding,
the type of @free@ is pleasingly simple: any type @'a@ whose kind is
less than or equal to @aff@ is acceptable. 
The @length@ function is always applicable.
The @get@ function is only applicable
if the element type is unrestricted as one element is duplicated.
To @set@ an array element displaces the previous content, which must
be at least affine.

The @map@ function can transform arrays with arbitrary elements. In 
particular, it can turn unrestricted elements into linear (affine)
ones. It takes a borrow of the input array and returns a newly created
output array. As @free@ing requires affine elements, we provide the
@iter@ function which takes a
suitable finalizer and an array with arbitrary elements, which is
consumed. Indeed such an iteration is the only way to free an array
with linear elements. A real-life
API would provide a combination of @map@ and @iter@ as a
``destructive'' map that consumes its input array. Assuming a uniform
representation, such a destructive map might be implemented by
in-place update.

To manage the different kinds of accessing the array we distinguish between constructors, destructors,
observers, and mutators.
Constructors and destructors like @create@ and @free@ manipulate the whole
array. 
The constructor
@create@ yields a linear resource which is consumed by @free@.
During the lifetime of the array resource @a@, we can split
off \emph{shared borrows} @&a@ that provide a read-only
view or \emph{exclusive borrows} @&!a@ for read-write views.
Observer functions such as @length@ and @get@ expect a shared borrow
argument, mutator functions such as @set@ expect
an exclusive borrow.

Each borrow is tied to a region whose lifetime is properly
contained in the lifetime of the resource.
In a region, we can split off as many shared borrows of a resource as
we like, but we can take only one exclusive borrow. In a
subsidiary region, we can take shared borrows of any
borrow or we can take an exclusive borrow of an exclusive borrow from an
enclosing region. Borrows are confined to their regions. Inside the region,
shared borrows are unrestricted (@un@) whereas exclusive
borrows are affine (@aff@).

Using the API 
we can create an
array of Fibonacci numbers in an imperative coding style:
\begin{lstlisting}
let mk_fib_array n =
  let a = create (n, 1) in
  for i = 2 to n - 1 do
    let x = get (&a, i-1) + get (&a, i-2) in(*@\label{ex:array:get}*)
    set (&!a, i, x)(*@\label{ex:array:set}*)
  done;
  a(*@\label{ex:array:return}*)
# mk_fib_array : int -> int Array.t
\end{lstlisting}

After
creation of the array, the presence of a borrow in the for loop
prevents access to the ``raw'' resource inside the loop's
body. In particular, the resource cannot be freed through a
borrow. \cref{ex:array:get} contains two shared borrows 
in the same expression which forms a region by itself (recall that
shared borrows are unrestricted and may thus be duplicated). These borrows
are split off the exclusive borrow used in \cref{ex:array:set} which
belongs to the next enclosing region corresponding to the loop body.
The whole array can be returned in \cref{ex:array:return} because  the
borrows are no longer in scope.
More precisely, here is is an annotated excerpt with regions explicitly
marked by braces @{| ... |}@:
\begin{lstlisting}[firstnumber=3]
  for i = 2 to n - 1 do {|
    let x = {| get (&a, i-1) + get (&a, i-2) |} in(*@\label{ex:array:region:get}*)
    set (&!a, i, x)
  |} done;
\end{lstlisting}

One region consists of the header expression of the @let@ in
\cref{ex:array:region:get}. It is contained in another region
spanning the body of the @for@ loop. \lang guarantees that borrows
never escape the smallest enclosing region. It employs a system of
\emph{indexed kinds} like @aff_r@ and @un_r@ where
$r$ is a positive integer that corresponds to the lexical nesting
depth of regions. For instance, the type of @&!a@ in
\cref{ex:array:set} has kind @aff_1@ whereas the type of
@&a@ in \cref{ex:array:region:get} has kind
@un_2@ and the typing of the inner region is such that types with
kind indexes greater than or equal to $2$ cannot escape.
In the example, borrows cannot escape  because they are consumed
immediately by @get@ and @set@.

\subsection{Solving sudokus with hybrid data-structures}

This section presents an implementation of a
backtracking Sudoku solver using a safe API for persistent arrays that
supports both mutable updates and immutable versioning.
The implementation showcases safe mixing of resource allocation and
deallocation in the presence of exclusive (mutable) and immutable
borrows. It also demonstrates two new aspects: the interaction between
closures and borrows and the notion of reborrowing.

Recently introduced persistent data structures
permit transient mutations where
non-linear uses lead to degraded performance
\cite{DBLP:conf/ml/ConchonF07} or to
dynamic and static checks \cite{DBLP:journals/pacmpl/Puente17}.
In particular, persistent Hash-Array-Mapped-Tries (HAMT) have been used with similar
APIs in several non-pure functional languages (OCaml, Clojure, \dots).
Affine types help formalize the performance contract between the programmer
and the library, while borrows avoid the need to thread state explicitly,
as usually required by an API for immutable data types.

Our implementation of a backtracking Sudoku solver abstracts this scenario.
The solver maintains a two-dimensional array to represent the state of
the game and uses backtracking when there are several choices to proceed.
As choice points may be revisited several times, it seems advantageous
to select a persistent data structure for the array.
However, local changes between choice points may be implemented as
cheap in-place mutations.

\cref{sig:hybarray} contains an API @HYBARRAY@ along with an
implementation @CowArray@
that enables using mutable and immutable modifications to
the board through affine types and borrows.
The signature differs slightly from the @Array@
signature. As our application requires the @get@ function, the array
elements must be unrestricted, but
the structure itself remains linear so as to be implemented in terms of @Array@.
The in-place mutation function @set_mut@ with type
@&!('a t) * int * 'a -> unit@ works on an exclusive borrow  whereas the persistent
@set@ operation has type @&('a t) * int * 'a -> 'a t@. It
takes a shared borrow because it only reads from the
argument array, but returns a fresh,  modified structure.
The module @CowArray@, also in \cref{ex:cow}, contains a very simple
implementation of 
@HYBARRAY@ that represents hybrid arrays
as regular arrays and uses copy-on-write for persistent
modifications. The function
@mapi: (int * &'a -> 'b) * &('a t) -> 'b t@
is a simple variation on @Array.map@ where the mapping function also
takes the position of the element. Recall that @Array.map@ always
creates a new array for the result.

In the example code, we make use of a two-dimensional version
@Matrix@ of the @CowArray@ data\-structure. The only difference is
that the API functions @get@, @set@, and @set_mut@ now take two index
parameters of type @int@ instead of one. The internal working is exactly the same.

\begin{figure}[tp]
  \centering
  \begin{subfigure}[t]{0.45\linewidth}
\begin{lstlisting}[numbers=none]
module type HYBARRAY = sig
  include ARRAY
  val set : &('a t) * int * 'a -> 'a t
  val set_mut : &!('a t) * int * 'a -> unit
end
\end{lstlisting}
  \end{subfigure}\hfill
  \begin{subfigure}[t]{0.55\linewidth}
\begin{lstlisting}[numbers=none]
module CowArray : HYBARRAY = struct
  include Array
  let set (a, i0, x0) =
    Array.mapi ((fun (i, x) -> if i = i0 then x0 else x), a)
  let set_mut = Array.set
end
\end{lstlisting}
  \end{subfigure}
  \vspace{-15pt}
  \caption{Signature and Implementation of hybrid arrays}
  \label{sig:hybarray}
  \label{ex:cow}

\begin{lstlisting}[numbers=none]
val propagate : int -> int -> &!Matrix.t -> int -> unit
val solve : int -> int -> Matrix.t -> unit
\end{lstlisting}

  \begin{subfigure}[t]{0.45\linewidth}
\begin{lstlisting}[linerange=1-13]
type board = IntSet.t Matrix.t(*@\label{line:boardtype}*)

let propagate_line i0 j0 board n =(*@\label{line:propline}*)
  for j = j0+1 to 8 do
    let cell = Matrix.get (&board, i0 , j) in(*@\label{line:reborrow_i}*)
    let cell' = IntSet.remove n cell in
    Matrix.set_mut (&!board, i0, j, cell')(*@\label{line:reborrow_m}*)
  done

let propagate i j board n =
  propagate_line i j &!board n;(*@\label{line:reborrow1}*)
  propagate_column i j &!board n;
  propagate_square i j &!board n(*@\label{line:reborrow3}*)
\end{lstlisting}
  \end{subfigure}\hfill
  \begin{subfigure}[t]{0.55\linewidth}
\begin{lstlisting}[linerange=14-31,firstnumber=last]
let rec solve i j board =(*@\label{line:solve}*)
  begin if is_solved &board then Matrix.print &board else
    let (new_i, new_j) = next_pos (i,j) in
    let try_solution n =(*@\label{line:try_solution}*)
      let new_board =
        Matrix.set (&board, i, j, IntSet.singleton n) in(*@\label{line:immutset}*)
      propagate i j &!new_board n;
      if is_valid &new_board then
        solve new_i new_j new_board
    in
    let cell = Matrix.get (&board, i, j) in
    IntSet.iter try_solution cell;(*@\label{line:iter}*)
  end;
  free board(*@\label{line:free:g}*)
\end{lstlisting}
  \end{subfigure}
  \vspace{-15pt}
  \caption{Excerpt of the Sudoku solver}
  \label{ex:sudoku}
\end{figure}


Our implementation of a Sudoku solver (~\cref{ex:sudoku}) represents
the board as a 2D-matrix (\cref{line:boardtype}).
Each cell contains an integer set of type @IntSet.t@ that represents
admissible solutions so far. This type is immutable, i.e.,
@IntSet.remove@ produces a new value.

The main functions are @solve@ and @propagate@ with the typings shown
on top of \cref{ex:sudoku}. The types of @propagate_line@ etc are the
same as for @propagate@. From the types, we can see that @solve@ takes
ownership of the board, whereas @propagate@ only takes a mutable
borrow. Hence, the @propagate@ functions can only modify or read the board,
whereas @solve@ has full control. 

The Sudoku solver @solve@ iterates over the cells and tries each possible
solution (\cref{line:try_solution}). 
When a value is picked for the current cell, it creates a choice point
in @new_board@, where the current cell is updated with an immutable modification (\cref{line:immutset}), and propagate
the changes with the @propagate@ function.
The @propagate@ function uses direct mutation through an
exclusive borrow of the matrix as it need not preserve the previous
version of the board.
The implementation of @propagate@ is split into three parts
for lines, columns, and square, which are all very similar to function
@propagate_lines@ (\cref{line:propline}).

As the @board@ parameter to the @propagate@ function is an exclusive
borrow, it should be handled 
in an affine manner. To pass it safely to the three helper functions,
the body of @propagate@ \emph{reborrows} (i.e., it takes a borrow of a
borrow) the board three times in 
\cref{line:reborrow1}-\cref{line:reborrow3}.
The function @propagate_line@ also contains two reborrows of the
exclusive borrow argument @board@, an
immutable one (\cref{line:reborrow_i}) and an exclusive one
(\cref{line:reborrow_m}). It demonstrates the facility 
to take immutable borrows from exclusive ones.

The typing ensures that the mutations do not compromise the state at
the choice point, because they operate on a new state @new_board@ created for one
particular branch of the choice.
As the @set@ function only requires an unrestricted shared borrow,
the closure @try_solution@ remains unrestricted even though
it captures the borrow @&board@.
The price is that @try_solution@ cannot escape from
@&board@'s region. In this example, the inferred region corresponds to the
@begin@/@end@ scope. 
Hence, @try_solution@ can be used in
the iteration in \cref{line:iter}.
As @board@ is linear we must free it outside of the region
before returning (\cref{line:free:g}).

While presented for copy-on-write arrays, the API
can easily be adapted to other persistent data structures with
transient mutability such as Relaxed-Radix Balance Vectors (RRB) \cite{DBLP:journals/pacmpl/Puente17}
or persistent HAMTs \cite{bagwell2001ideal,clojurehamt} to provide  a
convenient programming style without compromising performance.





\subsection{Iterators and regions}
\label{tuto:vector}

In the examples so far, regions do not appear in type signatures.
But for certain programming idioms, we want to extend the scope of a
region across a function boundary.
For instance, how should we fold on an array of linear objects?
Here is a first attempt at the type of a fold function:
\begin{lstlisting}[numbers=none]
val fold : ('a:'k) => ('a -> 'b -{'k}> 'b) -> 'a Array.t -> 'b -{lin}> 'b
\end{lstlisting}
This type puts no restrictions on the element type of the array, but
it requires the @fold@ function to consume the array and all its
elements. The last function arrow is linear because the array type (from
\cref{sec:imper-progr}) is linear. 

If we want to work on borrows of linear and affine resources, then the
typing gets more involved because
we must make sure those borrows are not leaked in the result.
We obtain the following signature for @bfold@, the borrowing fold operation:
\begin{lstlisting}[numbers=none]
val bfold : ('b:'k),('k <= lin_r) => (&(aff_r1,'a) -> 'b -{aff_r1}> 'b) -> &('k_1,'a Array.t) -> 'b -{'k_1}> 'b
\end{lstlisting}

The folded function receives a shared borrow of the element in the array.
The typing of the callback ensures
that this borrow is neither captured nor returned by the function.
This encapsulation is implemented with a universally quantified \emph{kind index variable} $r$.
The signature prescribes the type @&(aff_r1,'a)@ for the
shared borrow of the resource with an affine kind at region nesting $r+1$. The return
type of the callback is constrained to kind @'k <= lin_r@. The
important part of this constraint is the $_r$ index, which ensures
that the callback cannot return the borrowed argument from the more
deeply nested scope. 
The input of the fold is a shared borrow of the array,
which ensures that we have the right to share borrows of the inner content and
make multiple concurrent folds.

As an easy example, we fold over an array of files
@all_files : File.t Array.t@ to compute the sum of their sizes:
\begin{lstlisting}[numbers=none]
let total_size_of files = bfold (fun f s -> File.size f + s) files 0
let total_size = total_size_of &all_files
\end{lstlisting}

This approach is not sufficient if we want to mutate the elements of
the array during iteration.
To this end, we need to take an exclusive borrow of the structure to iterate
on:
\begin{lstlisting}[numbers=none]
val iter_mut : (&!(aff_r1,'a) -> unit) -> &!('k_1,'a Array.t) -> unit
\end{lstlisting}

While the distinction between mutable and immutable iteration functions
seems unfortunate, it is typical of
programming with borrows and is also present in the Rust standard library.
It enables the programmer to explicitly state how different iterations
may be composed and optimized.
It also enables different implementations such as
using parallel iterations in the immutable case.
\lang's region variables ensure that the content iterated on can never
be leaked outside of the iteration function.
This pattern is essential in many use cases of linearity such
as pools of linear objects (see \cref{tuto:pool}).

To close this discussion, let's see which primitives are needed to implement
functions like @bfold@ and @iter_mut@. The naive sequential implementations of
both functions boil down to a loop over the index range of the array:
\begin{flushleft}\vspace{-\baselineskip}
  \begin{minipage}[t]{0.49\linewidth}
\begin{lstlisting}
let rec bfold_helper f a z i =
  if i < 0 then z
  else bfold_helper f &&a (f (get_sb &&a i) z) (i-1)
let bfold f a z =
  let l = length &&a - 1 in
  bfold_helper f &&a z l
\end{lstlisting}
  \end{minipage}
  \begin{minipage}[t]{0.49\linewidth}
\begin{lstlisting}
let rec iter_helper f a i =
  if i < 0 then ()
  else (f (get_eb &&!a i); iter_helper f &&!a (i-1))
let iter_mut f a =
  let l = length &&a - 1 in
  iter_helper f &&!a l
\end{lstlisting}
  \end{minipage}
\end{flushleft}
Observe that we are proposing two different primitives
\begin{itemize}
\item
  @get_sb : &('k, 'a Array.t) -> int -{'k}> &('k, 'a)@ \\to get a
  shared borrow from a shared borrow of an array and
\item
  @get_eb : &!('k, 'a Array.t) -> int -{'k}> &!('k, 'a) (* Unsafe! *)@
\\to get an
  exclusive borrow from an exclusive borrow of  an array.
\end{itemize}
These primitives have the same underlying implemention (the same as
@get@ from \cref{sec:imper-progr}). Their types arise from the
intuition that the borrow of a structure should entitle to borrows of
its substructures, roughly, the borrow of an array could be considered as
array of borrows.
However, only @get_sb@ is safe: the shared borrow of the array entitles us to
obtain shared borrows for the elements \emph{in the same region} as
the shared borrow for the array freezes the array inside the
region. This freeze extends to the elements because 
lifetime of the array fully overlaps with the lifetime of its
elements. 
Considering @get_eb@, we see that we may obtain \emph{two different
exclusive borrows} of the same array element inside a region. Clearly,
the exclusive borrow for the element should live in a nested region
where the array is not accessible. Hence, the safe alternative is
to use a different function for obtaining borrows of elements
\begin{lstlisting}[numbers=none]
with_eborrow : ('b:'k),('k <= lin_r) => &!('k_1,'a Array.t) -> int -> (&!(aff_r1,'a) -> 'b) -{'k_1}> 'b
\end{lstlisting}
Hence, the helper function for @iter_mut@ should read like this
\begin{lstlisting}
let rec iter_helper f a i =
  if i < 0 then ()
  else (with_eborrow &!a i f; iter_helper f &!a (i-1))
\end{lstlisting}
which also explains the occurrence of @aff_r1@ in the type of
@iter_mut@.

We conclude that borrows of datastructures create the need for
differently typed access functions that are tailored for different use
cases. The argumentation whether such an access function is safe is
sometimes subtle and gives rise to nonobvious types.

\lstDeleteShortInline@


\section{The \lang language}

\lang is a minimal ML-like language with
let-polymorphism and abstract types. Its type system manages
linearity and borrowing using kinds and (kind-) constrained types.
For ease of presentation, we consider a simplified internal language.
\begin{itemize}[topsep=0pt]
\item Pattern matching is demonstrated on pairs rather than algebraic
  datatypes.
\item There are separate operators for borrowing and reborrowing (taking
  the borrow of a borrow); the surface language unifies these
  operators using ad-hoc polymorphism / typeclasses.
\item Regions are explicit in the code and must be annotated using the
  algorithms presented in \cref{regionannot}.
\item Regions are identified using nesting levels instead of region
  variables.
\end{itemize}

The rest of this section formalizes  \lang: the syntax (\cref{syntax}),
the statics in terms of a region annotation pass (\cref{regionannot}) and
syntax-directed typing (\cref{sdtyping}),
and a dynamics that is linearity- and resource-aware (\cref{sec:sem}).

\subsection{Syntax}
\label{syntax}

\begin{figure}[!tb]
  \input{grammar}
  \vspace{-7pt}
  \caption{Syntax}
  \label{grammar}
  \vspace{-10pt}
\end{figure}

\Cref{grammar} defines the syntax of \lang. Expressions are as usual
in an ML-like language.  The novel aspects are match
specifications, regions, and borrows.

A borrow $\borrow{x}$ is always taken from a variable $x$. The
\emph{borrow annotation} $\BORROW$ indicates whether the borrow is exclusive/affine
($\MBORROW$) or shared/unrestricted ($\IBORROW$). If $x$ is already borrowed,
then we need to use the reborrow expression $\reborrow{x}$.
A region $\region{\Sone x\BORROW}{e}$ is annotated with its nesting $n$, the variable $x$ that may be borrowed in $e$, and the kind of borrow $\BORROW$.
A match is indexed with a \emph{match specification} $\etransfm$ that indicates
whether the match operates on borrows ($\etransfm=\&^b$) or not ($\etransfm=id$).
We consider four primitive operations to manipulate resources:
$\create$, $\observe$, $\update$ and $\destroy$.
They serve a prototypes to demonstrate the typing and handling of
resources. For a concrete type of resource, there are further
arguments and perhaps different versions of the operations. But the
typing and behavior of the operations is analogous to the prototype
operations. Moreover, $\observe$ and $\update$ serve as eliminators
for borrow types.

Many types are indexed with kinds.
A kind $k$ is either a kind variable $\kvar$ or a constant
(linear $\klin$, affine $\kaff$, or unrestricted $\kun$) indexed
by a nesting level $n \in \Nat \cup \{\infty\}$.

A type $\tau$ is either a type variable, a pair type, a function type indexed by a
kind,
a type application $\tapp{\tcon}{\Multi{\tau}}$ of an abstract type constructor
$\T T$, or a borrowed type  $\borrowty{k}{\tau}$.

Type schemes $\schm$ add quantification over kind
variables $\kappa$, kinded type variables $(\alpha:k)$, and constraints
$C$ to a type, where a constraint is a list of inequalities over
kinds.

Abstract type constructors possess kind schemes $\kschm$ which relate
the kinds of the type constructors' arguments to the kind of the
constructed type.

Generally, we write lists with an overbar and (sometimes) an index as
in $\Multi[i]\tau$.

\subsection{Automatic region annotation}
\label{regionannot}

In the surface language (\cref{motivation}) region annotations are optional.
In the internal language, regions must be syntactically explicit and
annotated with a nesting index and a scoped variable.
This section defines a transformation $\RannotT{p}{p'}$ which
automatically inserts region annotations in  programs.
The input $p$ is a program with optional region annotations of the form $\region[]{}{e}$.
The output $p'$ is a program with explicit annotations
of the form $\region{\Sone x\BORROW}{e}$ such that no borrow occurs
outside a region.
We give an informal presentation of our code transformation and defer
the complete definition to \cref{appendix:regionannot}. This code transformation
aims to find, for each borrow $\borrow{x}$, the biggest region satisfying
the following rules:
\begin{enumerate}
\item The region should contain at least $\borrow{x}$.
\item The region must be contained in the scope of $x$.
\item An exclusive borrow $\borrow[\kaff]{x}$ should never share a region with any other borrow of $x$.
\item The variable $x$ cannot occur in the region of $\borrow{x}$.
\end{enumerate}
The transformation starts from each borrow $\borrow x$ and grows its associated
region until enlarging it would include the binding for $x$ or lead to a
conflicting use of $x$ or a conflicting borrow for $x$.
As an example, consider the following program:
\[
  \setlength{\jot}{-1pt}
  \RannotT{
\begin{aligned}
  \lam{a}{}\ &
  \letin{x}{\app{f}{\borrow[\kaff]{a}}}{}\\
  &g\ (\borrow[\kaff]{x});\\
  &f\ (\borrow[\kun]{x})\ (\borrow[\kun]{x})\\
\end{aligned}
}{
\begin{aligned}
  \lam{a}{}\region[1]{\Sone a\MBORROW}{\ &%
  \letin{x}{\app{f}{\borrow[\kaff]{a}}}{}\\
  &\region[2]{\Sone x\MBORROW}{g\ (\borrow[\kaff]{x})};\\
  &\region[2]{\Sone x\IBORROW}{f\ (\borrow[\kun]{x})\ (\borrow[\kun]{x})}\\
  &\!\!\!}
\end{aligned}
}
\]
As variable $a$ only has one borrow, its region covers its whole lexical scope.
Variable $x$ has multiple conflicting borrows and requires more consideration.
We place a first
region around the exclusive borrow and its function application,
and a second region around both shared borrows. This placement of region
is optimal: making any region bigger would cause a borrowing
error.
In particular, it is essential to place the second borrow around \emph{both}
occurrence of $(\borrow[\kun]{x}$: we want the function to receive both
borrows, but its result must not contain any borrow (otherwise it would escape).
Region indices are assigned after placement, so it is trivial to
ensure well-nested regions where the inner regions have
higher indices than the outer ones.

Programmers may also annotate regions explicitly.
The transformation considers an annotation as an
upper bound on the contained regions.
In the following program, a manual annotation has been inserted to ensure
no borrow enters the reference $r$:
\[
  \setlength{\jot}{-1pt}
  \RannotT{
\begin{aligned}
  \letin{&r}{\operatorname{ref}\ 0}{}\\
  \lam{a}{}\ &
  set\ r\ \region[]{}{g\ (\borrow[\kun]{a})};\\
  &f\ (\borrow[\kun]{a})
\end{aligned}
}{
\begin{aligned}
  \letin{&r}{\operatorname{ref}\ 0}{}\\
  \lam{a}{}\ &
  set\ r\ \region[1]{\Sone a\IBORROW}{g\ (\borrow[\kun]{a})};\\
  &\region[1]{\Sone a\IBORROW}{f\ (\borrow[\kun]{a})}
\end{aligned}
}
\]
The rules allow merging the two regions around the borrows
$\borrow[\kun]{a}$. However
the explicit annotation indicates that the region should stay around the closure
passed as argument of $set$. This feature is useful to control captures
by imperative APIs.

The code transformation is purely syntactic and must be used before
typing. It only produces well-nested annotations: if $\region[n]{b}{\dots}$
is nested inside $\region[n']{b'}{\dots}$, then $n > n'$. Furthermore, there
is at most one region per borrow, and exactly one region per exclusive borrow.
In the rest of this article, we assume
that all terms have been annotated by this code transformation and respect
these properties.

\subsection{Typing}
\label{sdtyping}

To avoid distraction, this section focuses on the essential and novel
parts of the type system. A complete
description is available in \cref{appendix:sdtyping}.
%
Here we only discuss the following judgments:
\begin{flushleft}
  \begin{tabular}{l@{ --- }l}
    $\inferS{C}{\E}{e}{\tau}$&
                               Expression $e$ has type $\tau$ in
                               environment $\E$ under constraints
                        $C$. \\
    $\inferSK{C}{\E}{\tau}{k}$&
                                Type $\tau$ has kind $k$ in
                                environment $\E$ under constraints
                                $C$. \\
    $\entail{D}{C}$ & Constraint $D$ entails constraint $C$.\\
    $D \equivC C$ & Constraints $C$ and $D$ are equivalent.
  \end{tabular}
\end{flushleft}

\paragraph{Kinds and constraints}

\affe uses kinds and constrained types to indicate
linear and affine types.
A kind $k$ is either a kind variable, $\kvar$, or a constant $\mul_n$.
The quality $\mul$ describes the use pattern of the type:
unrestricted $\kun$, affine $\kaff$, or linear $\klin$. The level $n \in \Nat \cup \{\infty\}$
describes the nested regions in which the value can be used.
Level $0$ refers to the top-level scope outside any region; we often elide it
and write $\kaff$ for $\kaff_0$. Level
$\infty$ refers to an empty region that is infinitely nested.
For instance, the constraint $\Cleq{\kvar}{\kun_\infty}$ indicates that
$\kvar$ must be unrestricted, but can be local to a region.
Kinds form a lattice described in \cref{sdtyp:lattice}.
Unrestricted values may be used
where affine ones are expected and affine ones are less restricted
than linear ones as reflected in \TirName{Lat-UAL}.
Values usable at level $n$ may be used at any more deeply
nested level $n'$ as defined in the \TirName{Lat-Level} axioms.
Constraints are conjunctions of inequality constraints over this kind
lattice, i.e., they specify upper or lower bounds for
kind variables or relate two kind variables.

\begin{figure}[tp]
  \begin{minipage}{0.65\linewidth}
  \begin{mathpar}
    \inferrule[Lat-UAL]{}{\kun \lk \kaff \lk \klin}
    \and
    \inferrule[Lat-Level]{\mul \lk \mul' \and n \lk n'}{\mul_n \lk_\Lat \mul'_{n'}}
  \end{mathpar}
\end{minipage}~
\begin{minipage}{0.2\linewidth}
  \centering
  \begin{tikzpicture}
    [->,auto,semithick, every node/.style={scale=0.7}]
    \node(U) {$\kun_0$} ;
    \node(A) [above left of=U] {$\kaff_0$} ;
    \node(L) [above left of=A] {$\klin_0$} ;
    \node(Un) [above right of=U] {$\kun_n$} ;
    \node(An) [above left of=Un] {$\kaff_n$} ;
    \node(Ln) [above left of=An] {$\klin_n$} ;
    \node(Uinf) [above right of=Un] {$\kun_\infty$} ;
    \node(Ainf) [above left of=Uinf] {$\kaff_\infty$} ;
    \node(Linf) [above left of=Ainf] {$\klin_\infty$} ;
    \path
    (U) edge (A)
    (A) edge (L)
    (Un) edge (An)
    (An) edge (Ln)
    (Uinf) edge (Ainf)
    (Ainf) edge (Linf)
    ;
    \path[dotted]
    (U) edge (Un)
    (A) edge (An)
    (L) edge (Ln)
    (Un) edge (Uinf)
    (An) edge (Ainf)
    (Ln) edge (Linf)
    ;
  \end{tikzpicture}
\end{minipage}

  \vspace{-10pt}
  \caption{Lattice ordering -- $k \lk_\Lat k'$}
  \label{sdtyp:lattice}
  \vspace{-10pt}
\end{figure}

\paragraph{Environments and bindings}
\label{sdtyping:envs}

\begin{figure}[tp]
  \begin{minipage}{0.36\linewidth}
  \begin{align*}
    \E ::=&\ \Eempty \mid \E;B\tag{Environments}\\
    B ::=&\ \bnone \tag{Empty} \\
    |&\ \bvar{\tvar}{\kschm} \tag{Types}\\
    |&\ \bvar{x}{\schm} \tag{Variables} \\
    |&\ \svar{x}{\schm}^n \tag{Suspended} \\
    |&\ \bbvar x k \schm \tag{Borrows}
  \end{align*}
  \vspace{-15pt}
  \caption{Type environments}
  \label{grammar:env}
  \end{minipage}\hfill
  \begin{minipage}{0.6\linewidth}
    \begin{tabular}
  {@{}>{$}r<{$}@{ $\vdash_e$ }
  >{$}c<{$}@{ $=$ }
  >{$}c<{$}@{ $\ltimes$ }
  >{$}c<{$}r}

  \Cleq{\schm}{\kun_\infty}
  &\bvar{x}{\schm}&\bvar{x}{\schm}&\bvar{x}{\schm}
  &(Both)\\

  {\Cempty}&
             {\bbvar[\IBORROW]{x} k \schm}&
                                            \bbvar[\IBORROW]{x}
                                            k{\schm}&{\bbvar[\IBORROW]{x}
                                                      k {\schm}}
  &(Borrow)
  
  \\[2mm]

  {\Cempty}&{B_x}&{B_x}&{\bnone}
  &(Left)\\
  {\Cempty}&{B_x}&{\bnone}&{B_x}
  &(Right)\\[2mm]

  {\Cempty}&{\bvar x \schm}&{\svar x \schm^n}&{\bvar x \schm}
  &(Susp)\\

  \Cleq{\BORROW'}{\BORROW}&
             {\bbvar x k \schm}&{\svar[\BORROW'] x \schm^n}&{\bbvar x k \schm}
  &(SuspB)\\

  {\Cempty}&
             {\svar{x} \schm^{n'}}&{\svar[\IBORROW] x \schm^n}&{\svar{x} \schm^{n'}}
  &(SuspS)\\
\end{tabular}

    \vspace{-5pt}
    \caption{Splitting rules for bindings -- $\lsplit{C}{B}{B_l}{B_r}$}
    \label{sdtyp:split}
  \end{minipage}
\end{figure}

\begin{figure*}[tp]
    \begin{mathpar}
      \ruleInstance
      \and
      \ruleSDVar
      \and
      \ruleSDLam
      \and
      \ruleSDApp
      \and
      \ruleSDRegion
      \and
      \ruleSDBorrow
      \and
      \ruleSDReBorrow
      \and
      \inferrule[BorrowBinding]{
        \entail{C}{(\BORROW_n\lk k) \wedge (k \lk \BORROW_\infty)}\\\\
        \BORROW\in\left\{\kun,\kaff\right\}
      }{
        \lregion{C}{x}
        {\svar[\BORROW]{x}{\schm}^n}
        {\bbvar x k \schm}
      }
    \end{mathpar}
    \vspace{-10pt}
    \caption{Selected typing rules ($\inferS{C}{\E}{e}{\tau}$)
      and borrowing rules ($\lregion{C}{x}{\E}{\E'}$)}
    \label{selectrules:borrow}
    \label{selectrules:binders}
    \label{sdtyp:app}
    \label{selectrules:region}
    \label{env:rule:borrow}
    \vspace{-5pt}
\end{figure*}

\lang controls the use of variables by supporting new modes of
binding in type environments $\E$, as defined in \cref{grammar:env}.
Environments contain standard bindings of type variables to kind schemes,
$\bvar{\tvar}{\kschm}$, value bindings $\bvar{x}{\schm}$, but also
suspended and borrow bindings.
A suspended binding, $\svar{x}{\schm}^n$, indicates that $x$ is
earmarked for a borrowing use in a nested region
marked with $x$ 
but
cannot be used directly.
A borrow binding,
$\bbvar x k \schm$,
replaces
such a suspended binding on entry to the $x$-region. It indicates
that the borrow $\borrow{x}$ can be used directly. Kind $k$
restricts the lifetime of the borrow to the region (see rules
\TirName{Region} and \TirName{BorrowBinding} in
\cref{selectrules:borrow} and the upcoming discussion of these rules).


Constraints on an environment control substructural properties by
restricting the types of variables.  The constraint $\Cleq{\E}{k}$
stands for the conjunction of $\Cleq\schm k$, for all $\bvar{x}{\schm}$
in $\E$, which in turn means that $k' \lk k$ where $k'$ is the kind of
the variable's type scheme $\schm$.
Borrow bindings follow the same rules, but suspended bindings
are forbidden in an environment $\E$ constrained like that.
This intuitive explanation is sufficient to understand
the \TirName{Var} and \TirName{Abs} rules shown in
\cref{selectrules:binders}.

Rule {\sc Var} looks up the type scheme of the variable $x$ in
the environment $\E$
and instantiates it with $\Instance(\E,\schm)$.
Instantiation follows the \hmx formulation and
takes as input a scheme $\schm$ and an environment $\E$
and returns a constraint $C$ and a type $\tau$.
The rule also
checks that the other bindings in $\E$ can be safely discarded by
imposing the constraint $\Cleq{\E\Sdel{x}}{\kaff_\infty}$.
It enforces that all remaining bindings (except $x$) are affine or
unrestricted and can therefore be discarded.

Rule {\sc Abs} ensures the kind annotation on the arrow type
($\tau_2 \tarr{k} \tau_1$) reflects the restrictions on captured variables
via the constraint $\Cleq\E k$.
If, for instance, any binding in $\E$ is affine, it gives rise to the
constraint $\Cleq{\kaff_n}{k}$ and the arrow kind is at least
affine at nesting level $n$.
Capturing a borrow is perfectly fine: the kind of the borrow is also a
lower bound of the arrow kind $k$ which restricts the closure
to the region of the borrow.
Capturing a suspended binding is forbidden.

\paragraph{Copying and Splitting}
\label{sdtyping:split}

The {\sc App} typing rule in \cref{sdtyp:app} demonstrates how \lang
deals with duplication and dropping of values.
The splitting $\lsplit{C}{\E}{\E_1}{\E_2}$ in the rule decomposes the
type environment $\E$ in two parts, $\E_1$ and $\E_2$, which are used
to typecheck the components of the application.

\cref{sdtyp:split} shows the action of splitting rules on single
bindings. If $x$'s type is unrestricted,
rule {\sc Both} indicates that we can duplicate it.
Similarly, unrestricted borrows can be duplicated with rule
{\sc Borrow}.
{\sc Left} and {\sc Right} rules are always applicable and move a binding
either to the left or right environment.
The rules {\sc Susp}, {\sc SuspB} and {\sc SuspS}
split off suspended bindings to
the left while conserving access to the binding on the right.
A suspended binding can later be turned
into a borrow inside a region. Splitting of suspended bindings is
asymmetric. It must follow the order of execution from left to right,
which means that a resource can be used first as a borrow on the left
and then later as a full resource on the right. The {\sc Susp} rule
works with a full resource, rule {\sc SuspB}
with a borrow, and rule {\sc SuspS} with a suspended binding.

Splitting applies whenever an
expression has multiple subexpressions:  function applications,
let bindings and pairs. In the
expression
$\letin{a}{\text{create}\ 8\ x}
{f\ \region[{}]{\Sone a\IBORROW}{\text{length}\ \borrow[\IBORROW]{a}}\ a}$,
the rule {\sc Susp} splits off a borrow from  the resource
$a$ to use it in the left argument.
As usual, a borrow cannot be active in the same scope as its resource.
The \emph{region} around its use ensures that the borrow in the left argument
does not
escape, which brings us to the next topic.

\paragraph{Regions}
\label{sdtyping:regions}




Borrowing is crucial to support an imperative programming style.
To guarantee the validity of a borrow, its lifetime must be properly contained in its
ancestor's lifetime. \lang ensures proper nesting of lifetimes by using
regions. The expression $\region{\Sone x\BORROW}{e}$ indicates a
region at nesting level $n$ in which a $\BORROW$-borrow can be taken of $x$.

The typing for a region (rule {\sc Region} in \cref{selectrules:region})
replaces suspended bindings by borrow bindings
(rule {\sc BorrowBinding}), typechecks the body
of the region, and ensures that the borrow does not leak outside.
This last check is done with indices that correspond to the nesting
level of the region. The kind $k$ of the borrow is indexed with the level $n$
corresponding to its region, thanks to the constraint $(\BORROW_n\lk
k)$. The constraint $\Cleq{\tau}{\klin_{n-1}}$ ensures that
the return type of the region must live at some enclosing, lower level.

As an example, consider the expression $\region{\Sone x\IBORROW}{f(\borrow[\IBORROW]{c})}$
where $c$ is a linear channel in an environment $\E$.
The first step is to check that $\svar[\IBORROW]{c}{\text{channel}}$
is in $\E$.
When entering the region, rule \TirName{Region} imposes
$\lregion{C}{x}{\E}{\E'}$, which defines $\E'$
corresponding to $\E$ where the suspended binding is replaced by the
borrow binding  $\bbvar[\IBORROW]{c}{k}{\text{channel}}$.
To constrain the borrow to this region we impose the constraint
$\entail{C}{(\kun_n\lk k) \wedge (k \lk\kun_\infty)}$, which affirms
that the borrow is unrestricted, but can only be used in nesting
levels $n$ and higher.
%
%
Rule {\sc Region} also imposes the constraint
$\Cleq{\tau}{\klin_{n-1}}$, which prevents the borrow, having kind $k$ of level
$\geq n$, from escaping the region's body of type $\tau$.




\paragraph{Pattern matching}
\label{sdtyping:matching}

Elimination of pairs is done using a matching construct
$(\matchin{x,x'}{e_1}{e_2})$.
This construct is mostly standard, except it can operate both
on a normal pair and a borrow of a pair.
The intuition is as follows:
A syntactic marker $\etransfm$ indicates if it applies to
a pair ($\etransfm = \operatorname{id}$) or a borrow ($\etransfm = \&^\BORROW$).
If $\etransfm = \operatorname{id}$, the typing simplifies to
the usual elimination of a pair.
Otherwise, $e_1$ is expected to be a borrow of type
$\borrowty{k}{\tyPair{\tau_1}{\tau'_1}}$
and the variables $x$ and $x'$ have type
$\borrowty{k}{\tau_1}$ and $\borrowty{k}{\tau'_1}$,
respectively. Thus, the borrow of a pair is considered as a pair of
borrows of its components.


\paragraph{Resource management}

To demonstrate how \lang{} deals with resources, we introduce an abstract
type $\tapp{\tres}{\tau}$ whose content of type $\tau$ must be
unrestricted ($\kun_0$) and which is equipped with the four
operations introduced in \cref{syntax}:
\begin{itemize}[topsep=0pt]
\item
$\create$:
$\forall{\kvar_\tvar}\bvar{\tvar}{\kvar_\tvar}.\
\qual{\Cleq {\kvar_\tvar} {\kun_0}}
{\tvar \tarr{} \tapp\tres\tvar}$
\item
$\observe$:
$\forall{\kvar\kvar_\tvar}\bvar{\tvar}{\kvar_\tvar}.\
\qual{\Cleq {\kvar_\tvar} {\kun_0}}
{\borrowty[\IBORROW]{\kvar}{\tapp\tres\tvar} \tarr{} \tvar}$
\item
$\update$:
$\forall{\kvar\kvar_\tvar}\bvar{\tvar}{\kvar_\tvar}.\
\qual{\Cleq {\kvar_\tvar} {\kun_0}}
\borrowty[\MBORROW]{\kvar}{\tapp\tres\tvar} \tarr{} \tvar \tarr{\kaff} \tunit$
\item
$\destroy$:
$\forall{\kvar_\tvar}\bvar{\tvar}{\kvar_\tvar}.\
\qual{\Cleq {\kvar_\tvar} {\kun_0}}
\tapp\tres\tvar \tarr{} \tunit$
\end{itemize}

\subsection{Semantics}
\label{sec:sem}
\begin{figure}[tbp]
  \begin{subfigure}[t]{0.5\linewidth}
  \begin{align*}
    \htag{Elaborated expressions}
    e ::=~& x \mid \ivar x {\Multi k} {\Multi\tau} \mid \lam[k]{x}{e} \mid \iapp\Sp{e}{e'} \\
    \mid~& \ipair\Sp{k}{e}{e'} \mid \imatchin\Sp\etransfm{x,y}{e}{e'} \tag{Pairs}\\
    \mid~& \ilet\Sp{x}{e}{e'} \tag{Mono let}\\
    \mid~& \iletfun\Sp{x}\schm{k}{y}{e}{e'} \tag{Poly let}\\
    \mid~& \region{\Sone x\BORROW}{e} \tag{Region}\\
    \mid~& \borrow{x} \mid \reborrow{x} \tag{Borrows}\\
    \mid~& \create \mid \observe \\\mid~& \update \mid \destroy \tag{Resources}\\
    \htag{Splittings}
    \Sp :~& (\lsplit{C}{\E}{\E_l}{\E_r})
\end{align*}
\end{subfigure}\hfill
\begin{subfigure}[t]{0.49\linewidth}
  \begin{align*}
    \htag{Storables}
    w ::=~& \StPClosure \VEnv {\Multi\kvar} C k x e \tag{Poly Closures}\\
    \mid~& \StClosure \VEnv k x e \tag{Closures} \\
    \mid~& \StPair k r {r'} \tag{Pairs} \\
    \mid~& \StRes r \tag{Resources} \\
    \mid~& \StFreed \tag{Freed Resource}
    \htag{Environment}
    \Addr ::=~& \Multi\IBORROW\hspace{0.5mm}\Multi\MBORROW\hspace{0.5mm}\Loc \tag{Locations}
    \\
    \Perm ::=~& \{\} \mid \Perm + \Addr \tag{Permissions}
    \\
    r ::=~& \Addr \mid c \tag{Results}\\
    \VEnv ::=~& \Eempty \mid \VEnv( x \mapsto r) \tag{Enviroments} \\
    \Store ::=~& \Eempty \mid \Store( \Loc \mapsto w) \tag{Stores}
  \end{align*}
\end{subfigure}
  \vspace{-15pt}
\caption{Syntax of internal language}
\label{fig:syntax-internal-language}
  \vspace{-10pt}
\end{figure}


\lstMakeShortInline[keepspaces,basicstyle=\normalsize\normalfont]@

It is straightforward to give a standard semantics for \lang, but such
a semantics would not be very informative. In this section, we give a
big-step semantics that performs explicit bookkeeping of the number of
times a value is used and of the mode in which a reference to a
resource is used (e.g., borrowed or not). This bookkeeping is based on a set of permissions
that regulate the currently allowed mode of access to resources and
closures. It enables us to state and prove a highly informative type
soundness result (see \cref{sec:metatheory}) with expressive invariants that ensure proper
resource usage.

The dynamics of \lang is given in big-step
functional
style~\cite{siek13:_type_safet_three_easy_lemmas,DBLP:conf/esop/OwensMKT16,
  DBLP:conf/popl/AminR17}.  A function
@eval@
manipulates the semantic objects defined in
\cref{fig:syntax-internal-language}.
%
The semantics is defined in terms of \emph{elaborated expressions} $e$
with kind, constraint, and splitting annotations inserted by the typechecker.
A splitting $\Sp$ is evidence of the splitting relation for type environments
used in the typing rules.

Let-polymorphism in the surface language gives rise to elaborated
$\eletfun$ expressions annotated with a type scheme $\schm$ and a kind $k$ indicating their
usage restriction (linear, affine, etc) relative to the variables and
constraints of $\schm$. Their use
gives rise to explicit instantiation of the kind and type variables.
Pairs come with a kind tag $k$ indicating the usage restriction.

Addresses $\Addr$ are composed of a raw location $\Loc$, which is just
a pointer into a store, and a stack of modifiers that indicates the
borrows and reborrows that have been taken from $\Loc$.
Once we have taken an unrestricted borrow (from a raw location or a borrowed
one), then we can take further unrestricted borrows from it, but no more affine
ones.

A permission $\Perm$ is a set of addresses that may be accessed during
evaluation. A well-formed permission contains at most one address for each raw
location.

Non-trivial results are boxed in the  semantics. So, a result
$r$ is either an address or a primitive constant (e.g., a number).

A value environment $\VEnv$  maps variables to results.

A storable $w$ describes the content of a location in the store. There are five
kinds of storables. A \emph{poly closure} represents a polymorphic
function. It consists of an environment and the components of an
elaborated abstraction. A \emph{closure} represents a monomorphic
function in the usual way.
A \emph{resource} contains
a result and the \emph{hole} $\StFreed$ fills a released location.

A store $\Store$ is a partial map from raw locations to
storables. The function
@salloc: store -> storable -> (loc *store)@ is such that
@salloc delta@ @w@ allocates an unused location in @delta@ and fills it with
@w@. It returns the location and the extended store.

The evaluation function is indexed by a step count @i@ so that each
invocation is guaranteed to terminate either with an error, a timeout,
or a result. Its return type is a monad
@'a sem@ which combines error reporting and timeout:
\begin{lstlisting}
type 'a sem = Error of string | TimeOut | Ok of 'a
val eval: store->perm->venv->int->exp->(store *perm *result) sem
\end{lstlisting}
Function @eval@ evaluates the given expression in the context of an initial store, a
permission to use addresses in the store, a value environment, and a
step count. If successful, it returns the final store, the remaining
permissions, and the actual result.

\begin{figure*}
  \begin{minipage}[t]{0.5\linewidth}
\begin{lstlisting}[style=rule,linerange=eval\ header-(**)]
(* with splitting*)
type ident = string

type splitting = ident list * ident list

type loc = int
type borrow = Imm | Mut

type address = Address of borrow list * loc

let loc : loc -> address =
  fun ell -> Address ([], ell)

type result = RADDR of address
            | RCONST of int
            | RVOID

type venv = VEmpty | VUpd of venv * ident * result

type alpha = int                (* type variable *)

type 'a var = V of int

type kind = KVAR of kind var
          | KLIN of int option
          | KAFF of int option
          | KUNR of int option

type kappa = kind var           (* kind variable *)
type kappas = kappa list

type constr = (kind * kind) list

type ty = TYVAR of alpha
        | TYARR of kind * ty * ty
        | TYPAIR of ty * ty
        | TYCON of ty list * ident
        | TYBOR of borrow * ty

type schm = SCHM of kappas * (alpha * kind) list * constr * ty

type exp = Const of int
         | Var of ident
         | Varinst of ident * kind list
         | Lam of kind * ident * exp
         | App of exp * exp * splitting
         | Let of ident * exp * exp * splitting
         | LetFun of ident * schm * kind * ident * exp * exp * splitting
         | Pair of kind * exp * exp * splitting
         | Match of ident * ident * exp * exp * splitting
         | Matchborrow of ident * ident * exp * exp * splitting
         | Region of exp * int * ident * ty * borrow
         | Borrow of borrow * ident
         | Create (* of exp *)
         | Destroy of exp
         | Observe of exp
         | Update of exp * exp * splitting
         (* anf cases *)
         | VApp of ident * ident
         | VPair of kind * ident * ident
         | VMatch of ident * ident * ident * exp * splitting
         | VMatchborrow of ident * ident * ident * exp * splitting
         | VCreate of ident
         | VDestroy of ident
         | VObserve of ident
         | VUpdate of ident * ident

type storable
  = STPOLY of venv * kappas * constr * kind * ident * exp
  | STCLOS of venv * kind * ident * exp
  | STPAIR of kind * result * result
  | STRSRC of result
  | STRELEASED

(* hashtable and next free location *)
module Store = Map.Make(Int)
type store = Store of storable Store.t * loc

type perm = address list

let s_add : perm -> address -> perm =
  fun p a -> a :: p
let s_rem : perm -> address -> perm =
  assert false
let s_sub : perm -> perm -> perm =
  assert false

let (<+>) = s_add
let (<->) = s_rem
let (--) = s_sub
let (++) = (@)

let (<|=) = List.mem

let (!$) = loc

(* end of syntax *)

type 'a sem = Error of string | TimeOut | Ok of 'a

let sembind : 'a sem -> ('a -> 'b sem) -> 'b sem =
  fun sa fsb ->
  match sa with
    Error (s) -> Error (s)
  | TimeOut -> TimeOut
  | Ok (a) -> fsb a

let (let*) = sembind

let borrowed : address -> borrow -> address sem =
  fun (Address (bs, ell)) b ->
  match b with
    Imm ->
     Ok (Address (b :: bs, ell))
  | Mut ->
     match bs with
       [] ->
        Ok (Address ([b], ell))
     | (Mut :: bs) ->
        Ok (Address (b :: bs, ell))
     | (Imm :: bs) ->
        Error ("trying to take mutable borrow of immutable borrow")

let rborrowed : borrow -> address -> address sem =
  fun b a -> borrowed a b

let (<.>) = rborrowed

let checkborrowed : address -> borrow -> bool =
  fun (Address (bs, ell)) b ->
  match bs with
    [] ->
     false
  | (b' :: _) ->
     (b = b')

let (<?>) = checkborrowed

let rec vlookup : venv -> ident -> result sem =
  fun gamma x ->
  match gamma with
    VEmpty -> Error ("variable not found")
  | VUpd (gamma, y, r) ->
     if x = y then Ok r else vlookup gamma x

let (.!()) = vlookup

let rec vrestrict : venv -> ident list -> venv =
  fun gamma vs ->
  match gamma with
    VEmpty -> VEmpty
  | VUpd (gamma, x, r) ->
     let gamma' = vrestrict gamma vs in
     if List.mem x vs then VUpd (gamma', x, r) else gamma'

let vsplit : venv -> splitting -> venv * venv =
  fun gamma (vs1, vs2) ->
  vrestrict gamma vs1, vrestrict gamma vs2

let salloc : store -> storable -> (loc * store) =
  fun delta w ->
  match delta with
    Store (htable, nexta) ->
    (nexta, Store (Store.add nexta w htable, nexta+1))

let slookup : store -> loc -> storable sem =
  fun (Store (htable, _)) ell ->
  match Store.find_opt ell htable with
  | Some x -> Ok x
  | None -> Error "illegal store location"

let supdate : store -> loc -> storable -> store sem =
  fun (Store (htable, limit)) ell w ->
  let htable' = Store.add ell w htable in
  Ok (Store (htable', limit))

let (.*()) = slookup
let (.*()<-) = supdate

let getloc : result -> loc sem =
  fun r ->
  match r with
    RADDR (Address ([], rl)) -> Ok rl
  | _ -> Error ("raw location expected")

let getborrowed_loc : result -> (borrow * borrow list * loc) sem =
  fun r ->
  match r with
    RADDR (Address (b :: bs, ell)) -> Ok (b, bs, ell)
  | _ -> Error ("borrowed address expected")

let getaddress : result -> address sem =
  fun r ->
  match r with
    RADDR (rho) -> Ok rho
  | _ -> Error ("address expected")

let getstpoly : storable -> (venv * kappas * constr * kind * ident * exp) sem =
  fun w ->
  match w with
    STPOLY (gamma, kappas, cstr, k, x, e) ->
      Ok (gamma, kappas, cstr, k, x, e)
  | _ ->
     Error ("expected STPOLY")

let getstclos : storable -> (venv * kind * ident * exp) sem =
  fun w ->
  match w with
    STCLOS (gamma, k, x, e) ->
    Ok (gamma, k, x, e)
  | _ ->
     Error ("expected STCLOS")

let getstpair : storable -> (kind * result * result) sem =
  fun w ->
  match w with
    STPAIR (k, r1, r2) ->
    Ok (k, r1, r2)
  | _ ->
     Error ("expected STPAIR")

let getstrsrc : storable -> result sem =
  fun w ->
  match w with
    STRSRC (r) -> Ok r
  | _ -> Error ("expected STRSRC")

let (let*?) : bool -> (unit -> 'b sem) -> 'b sem =
  fun b f -> if b then f () else Error ("test failed")

let (-:>) : 'a -> 'b -> 'a * 'b =
  fun a b -> a, b
let (.+()) : venv -> ident * result -> venv =
  fun v (i, r) -> VUpd (v, i, r)

type subst
let (-->) : 'a list -> 'a var list -> subst = assert false
let (./{}) : 'a -> subst -> 'a = assert false

let (<=>) : constr -> constr -> bool = assert false
let (<<=) : kind -> kind -> kind * kind =
  fun k1 k2 -> k1, k2

let rec borrowedperm : borrow -> address -> ty -> store -> perm sem =
  fun b rho t delta ->
  let* rho' = b<.>rho in
  let* pi = match t with
      TYPAIR (t_1, t_2) ->
        let Address (bs, ell) = rho' in
        let* w = delta.*(ell) in
        let* (k, r_1, r_2) = getstpair w in
        let* rho_1 = getaddress r_1 in
        let* rho_2 = getaddress r_2 in
        let* pi_1 = borrowedperm b rho_1 t_1 delta in
        let* pi_2 = borrowedperm b rho_2 t_2 delta in
        Ok (rho' :: (pi_1 ++ pi_2))
    | _ ->
        Ok ([rho'])
  in
  Ok pi
   
let rec reach : address -> ty -> store -> perm sem =
  fun rho t delta ->
  match t with
    TYPAIR (t_1, t_2) ->
     let Address (bs, ell) = rho in
     let* w = delta.*(ell) in
     let* (k, r_1, r_2) = getstpair w in
     let* rho_1 = getaddress r_1 in
     let* rho_2 = getaddress r_2 in
     let* pi_1 = reach rho_1 t_1 delta in
     let* pi_2 = reach rho_2 t_2 delta in
     Ok (rho :: (pi_1 ++ pi_2))
  | _ ->
     Ok ([rho])

let withaddr : result sem -> (address -> 'a sem) -> 'a sem =
  fun rsem acont ->
  let* r = rsem in
  match r with
  | RADDR a ->
     acont a
  | _ ->
     Error "address expected"

let (let+) = withaddr
    
let rec mapsem : ('a -> 'b sem) -> 'a list -> 'b list sem =
  fun f xs ->
  match xs with
  | [] ->
     Ok []
  | x :: xs ->
     let* y = f x in
     let* ys = mapsem f xs in
     Ok (y :: ys)

let (<..>) =
  fun b pi ->
  mapsem (rborrowed b) pi

  (* eval header *)
let rec eval
  (delta:store) (pi:perm) (gamma:venv) i e
  : (store * perm * result) sem =
  if i=0 then TimeOut else
  let i' = i - 1 in
  match e with
  (**)
  (* rule const *)
  | Const (c) ->
    Ok (delta, pi, RCONST c)
  (**)
  (* rule var *)
  | Var (x) ->
    let* r = gamma.!(x) in
    Ok (delta, pi, r)
  (**)
  (* rule varinst *)
  | Varinst (x, ks) ->
    let* rx = gamma.!(x) in
    let* ell = getloc rx in
    let*? () = !$ ell <|= pi in
    let* w = delta.*(ell) in
    let* (gamma', kappas', cstr', k', x', e') = getstpoly w in
    let pi' =
      if cstr'./{ks-->kappas'} <=> [(k' <<= KUNR None)]./{ks-->kappas'}
      then pi else pi <-> !$ ell
    in
    let w = STCLOS (gamma', k'./{ks-->kappas'}, x', e'./{ks-->kappas'}) in
    let (ell', delta') = salloc delta w in
    Ok (delta', pi' <+> !$ ell', RADDR !$ ell')
  (**)
  (* rule lam *)
  | Lam (k, x, e) ->
    let w = STCLOS (gamma, k, x, e) in
    let (ell', delta') = salloc delta w in
    let pi' = pi <+> !$ ell' in
    Ok (delta', pi', RADDR !$ ell')
  (**)
  (* rule sapp *)
  | App (e_1, e_2, sp) ->
    let (gamma_1, gamma_2) = vsplit gamma sp in
    let* (delta_1, pi_1, r_1) = eval delta pi gamma_1 i' e_1 in
    let* ell_1 = getloc r_1 in
    let*? () = !$ ell_1 <|= pi_1 in
    let* w = delta_1.*(ell_1) in
    let* (gamma', k', x', e') = getstclos w in
    let pi_1' = if k' <= KUNR None then pi_1 else pi_1 <-> !$ ell_1 in
    let* delta_1' = delta_1.*(ell_1) <- (if k' <= KUNR None then w else STRELEASED) in
    let* (delta_2, pi_2, r_2) = eval delta_1' pi_1' gamma_2 i' e_2 in
    let* (delta_3, pi_3, r_3) = eval delta_2 pi_2 gamma'.+(x'-:> r_2) i' e' in
    Ok (delta_3, pi_3, r_3)
  (**)
  (* rule sappanf *)
  | VApp (x_1, x_2) ->
    let* r_1 = gamma.!(x_1) in
    let* ell_1 = getloc r_1 in
    let*? () = !$ ell_1 <|= pi in
    let* w = delta.*(ell_1) in
    let* (gamma', k', x', e') = getstclos w in
    let pi' = (if k' <= KUNR None then pi else pi <-> !$ ell_1) in
    let* delta' = delta.*(ell_1) <- (if k' <= KUNR None then w else STRELEASED) in
    let* r_2 = gamma.!(x_2) in
    let* (delta_3, pi_3, r_3) = eval delta' pi' gamma'.+(x'-:> r_2) i' e' in
    Ok (delta_3, pi_3, r_3)
  (**)
  (* rule slet *)
  | Let (x, e_1, e_2, sp) ->
    let (gamma_1, gamma_2) = vsplit gamma sp in
    let* (delta_1, pi_1, r_1) = eval delta pi gamma_1 i' e_1 in
    let* (delta_2, pi_2, r_2) = eval delta_1 pi_1 gamma_2.+(x -:> r_1) i' e_2 in
    Ok (delta_2, pi_2, r_2)
  (**)
  (* rule sletfun *)
  | LetFun (f, sigma, k, x, e, e', sp) ->
    let (gamma_1, gamma_2) = vsplit gamma sp in
    let SCHM(kappas, _, cstr, ty) = sigma in
    let w = STPOLY (gamma_1, kappas, cstr, k, x, e') in
    let (ell', delta') = salloc delta w in
    let pi' = pi <+> !$ ell' in
    let* (delta_1, pi_1, r_1) = eval delta' pi' gamma_2.+(f -:> RADDR !$ ell') i' e' in
    Ok (delta_1, pi_1, r_1)
  (**)
  (* rule spair *)
  | Pair (k, e_1, e_2, sp) ->
    let (gamma_1, gamma_2) = vsplit gamma sp in
    let* (delta_1, pi_1, r_1) = eval delta pi gamma_1 i' e_1 in
    let* (delta_2, pi_2, r_2) = eval delta_1 pi_1 gamma_2 i' e_2 in
    let w = STPAIR (k, r_1, r_2) in
    let (ell', delta') = salloc delta_2 w in
    Ok (delta', pi_2 <+> !$ ell', RADDR !$ ell')
  (**)
  (* rule spairanf *)
  | VPair (k, x_1, x_2) ->
    let* r_1 = gamma.!(x_1) in
    let* r_2 = gamma.!(x_2) in
    let w = STPAIR (k, r_1, r_2) in
    let (ell', delta') = salloc delta w in
    let pi' = pi <+> !$ ell' in
    Ok (delta', pi', RADDR !$ ell')
  (**)
  (* rule smatch *)
  | Match (x, x', e_1, e_2, sp) ->
    let (gamma_1, gamma_2) = vsplit gamma sp in
    let* (delta_1, pi_1, r_1) = eval delta pi gamma_1 i' e_1 in
    let* ell = getloc r_1 in
    let* w = delta_1.*(ell) in
    let* (k', r_1', r_2') = getstpair w in
    let pi_1' = (if k' <= KUNR None then pi_1 else pi_1 <-> !$ ell) in
    let gamma_2' = gamma_2.+(x -:> r_1).+(x' -:> r_1') in
    let* (delta_2, pi_2, r_2) = eval delta_1 pi_1' gamma_2' i' e_2 in
    Ok (delta_2, pi_2, r_2)
  (**)
  (* rule smatchanf *)
  | VMatch (x, x', z, e_2, sp) ->
    let (gamma_1, gamma_2) = vsplit gamma sp in
    let* r = gamma_1.!(z) in
    let* ell = getloc r in
    let* w = delta.*(ell) in
    let* (k, r_1, r_1') = getstpair w in
    let pi' = if k <= KUNR None then pi else pi <-> !$ ell in
    let* delta' = delta.*(ell) <- (if k <= KUNR None then w else STRELEASED) in
    let gamma_2' = gamma_2.+(x -:> r_1).+(x' -:> r_1') in
    let* (delta_2, pi_2, r_2) = eval delta' pi' gamma_2' i' e_2 in
    Ok (delta_2, pi_2, r_2)
  (**)
  (* rule matchborrow *)
  | Matchborrow (x, x', e_1, e_2, sp) ->
    let (gamma_1, gamma_2) = vsplit gamma sp in
    let* (delta_1, pi_1, r_1) = eval delta pi gamma_1 i' e_1 in
    let* rho = getaddress r_1 in
    let* (b, _, ell) = getborrowed_loc r_1 in
    let* w = delta_1.*(ell) in
    let* (k', r_1', r_2') = getstpair w in
    let* rho_1 = getaddress r_1' in
    let* rho_2 = getaddress r_2' in
    let* rho_1' = b<.>rho_1 in
    let* rho_2' = b<.>rho_2 in
    let pi_1'' = (if k' <= KUNR None then pi_1 else pi_1 <-> rho) in
    let gamma_2'' = gamma_2.+(x -:> RADDR rho_1').+(x' -:> RADDR rho_2') in
    let* (delta_2, pi_2, r_2) = eval delta_1 pi_1'' gamma_2'' i' e_2 in
    Ok (delta_2, pi_2, r_2)
  (**)
  (* let pi_2' = (((pi_2 <-> rho_1') <-> rho_2') <+> rho_1) <+> rho_2 in *)
  (* rule matchborrowanf *)
  | VMatchborrow (x, x', z, e_2, sp) ->
    let (gamma_1, gamma_2) = vsplit gamma sp in
    let* r_1 = gamma_1.!(z) in
    let* (b, _, ell) = getborrowed_loc r_1 in
    let* w = delta.*(ell) in
    let* (k', r_1', r_2') = getstpair w in
    let* rho = getaddress r_1 in
    let pi'' = (if k' <= KUNR None then pi else pi <-> rho) in
    let delta'' = delta in
    let* rho_1 = getaddress r_1' in
    let* rho_2 = getaddress r_2' in
    let* rho_1' = b<.>rho_1 in
    let* rho_2' = b<.>rho_2 in
    let r_1'' = RADDR rho_1' in
    let r_2'' = RADDR rho_2' in
    let gamma_2'' = gamma_2.+(x -:> r_1'').+(x' -:> r_2'') in
    let* (delta_2, pi_2, r_2) = eval delta'' pi'' gamma_2'' i' e_2 in
    Ok (delta_2, pi_2, r_2)
  (**)
  (* rule sregion *)
  | Region (e, n, x, ty_x, b) ->
     let+ rho = gamma.!(x) in
     let* rho' = b<.>rho in
     let* pi' = reach rho ty_x delta in
     let* pi'' = b<..>pi' in
     let gamma' = gamma.+(x -:>RADDR rho') in
     let pi = (pi ++ pi'') -- pi' in
     let* (delta_1, pi_1, r_1) = eval delta pi gamma' i' e in
     let pi_1 = (pi_1 -- pi'') ++ pi' in
     Ok (delta_1, pi_1, r_1)
  (**)
  (*(* previous *)
  | Region (e_1, n, x, b) ->
    let* rx = gamma.!(x) in
    let* rho = getaddress rx in
    let*? () = rho <|= pi in
    let* rho' = b<.>rho in
    let gamma' = gamma.+(x -:>RADDR rho') in
    let pi' = (pi <-> rho) <+> rho' in
    let* (delta_1, pi_1, r_1) = eval delta pi' gamma' i' e_1 in
    let pi_1' = (pi <-> rho') <+> rho in
    Ok (delta_1, pi_1', r_1)
   *)
  (* rule sborrow *)
  | Borrow (b, x) ->
     let+ rho = gamma.!(x) in
     let*? () = rho <?> b && rho <|= pi in
     Ok (delta, pi, RADDR rho)
  (**)
  (*(* previous *)
  | Borrow (b, x) ->
    let* rx = gamma.!(x) in
    let* rho = getaddress rx in
    let* rho' = b<.>rho in
    let*? () = rho' <|= pi in
    Ok (delta, pi, RADDR rho')
   *)
  (* rule screate *)
  | Create ->
    let w = STRSRC (RCONST 0) in
    let (ell_1, delta_1) = salloc delta w in
    let pi_1 = pi <+> !$ ell_1 in
    Ok (delta_1, pi_1, RADDR !$ ell_1)
  (**)
  (* rule screateanf *)
  | VCreate (x) ->
    let* r = gamma.!(x) in
    let w = STRSRC (r) in
    let (ell_1, delta_1) = salloc delta w in
    let pi_1 = pi <+> !$ ell_1 in
    Ok (delta_1, pi_1, RADDR !$ ell_1)
  (**)
  (* rule sdestroy *)
  | Destroy (e_1) ->
    let* (delta_1, pi_1, r_1) = eval delta pi gamma i' e_1 in
    let* rho = getaddress r_1 in
    let* ell = getloc r_1 in
    let* w = delta_1.*(ell) in
    let* r = getstrsrc w in
    let*? () = rho <|= pi_1 in
    let* delta_1' = delta_1.*(ell) <- STRELEASED in
    let pi_1' = pi_1 <-> !$ ell in
    Ok (delta_1', pi_1', RVOID)
  (**)
  (* rule sdestroyanf *)
  | VDestroy (x) ->
    let* r = gamma.!(x) in
    let* rho = getaddress r in
    let* ell = getloc r in
    let* w = delta.*(ell) in
    let* r = getstrsrc w in
    let*? () = rho <|= pi in
    let* delta_1 = delta.*(ell) <- STRELEASED in
    let pi_1 = pi <-> !$ ell in
    Ok (delta_1, pi_1, RVOID)
  (**)
  (* rule sobserve *)
  | Observe (e_1) ->
    let* (delta_1, pi_1, r_1) = eval delta pi gamma i' e_1 in
    let* rho = getaddress r_1 in
    let*? () = rho <|= pi_1 in
    let* (b, _, ell) = getborrowed_loc r_1 in
    let*? () = (b = Imm) in
    let* w = delta_1.*(ell) in
    let* r = getstrsrc w in
    Ok (delta_1, pi_1, r)
  (**)
  (* rule sobserveanf *)
  | VObserve (x) ->
    let* r = gamma.!(x) in
    let* rho = getaddress r in
    let*? () = rho <|= pi in
    let* (b, _, ell) = getborrowed_loc r in
    let*? () = (b = Imm) in
    let* w = delta.*(ell) in
    let* r' = getstrsrc w in
    Ok (delta, pi, r')
  (**)
  (* rule supdate *)
  | Update (e_1, e_2, sp) ->
    let (gamma_1, gamma_2) = vsplit gamma sp in
    let* (delta_1, pi_1, r_1) = eval delta pi gamma_1 i' e_1 in
    let* rho = getaddress r_1 in
    let* (b, _, ell) = getborrowed_loc r_1 in
    let*? () = (b = Mut) in
    let* (delta_2, pi_2, r_2) = eval delta_1 pi_1 gamma_2 i' e_2 in
    let* w = delta_2.*(ell) in
    let* r = getstrsrc w in
    let*? () = rho <|= pi_2 in
    let* delta_2' = delta_2.*(ell) <- STRSRC (r_2) in
    let pi_2' = pi_2 <-> rho in
    Ok (delta_2', pi_2', RVOID)
  (**)
  (* rule supdateanf *)
  | VUpdate (x_1, x_2) ->
    let* r_1 = gamma.!(x_1) in
    let* rho = getaddress r_1 in
    let* (b, _, ell) = getborrowed_loc r_1 in
    let*? () = (b = Mut) in
    let* r_2 = gamma.!(x_2) in
    let* w = delta.*(ell) in
    let* r = getstrsrc w in
    let*? () = rho <|= pi in
    let* delta' = delta.*(ell) <- STRSRC (r_2) in
    let pi' = pi <-> rho in
    Ok (delta', pi', RVOID)
(**)
\end{lstlisting}
  \lstsemrule{sapp}
  \medskip
  \lstsemrule{sborrow}
  \end{minipage}%
  \begin{minipage}[t]{0.5\linewidth}
  \lstsemrule{varinst}
  \medskip 
  \lstsemrule{sregion}
  \end{minipage}
  \vspace{-1em}
  \caption{Big-step interpretation}
  \label{fig:big-step-interpretation}
  \vspace{-1em}
\end{figure*}


We give some excerpts of the definition of @eval@ in
\cref{fig:big-step-interpretation} and leave the full
definition for \cref{sec:semant-defin}.
The definition uses OCaml syntax with extensive pretty
printing. The pervasive @let*@ operator acts as monadic bind
for the @sem@ monad. The operator
@let*? : bool -> (unit -> 'a sem) -> 'a sem@
 converts a boolean
argument into success or failure in the monad.
\begin{lstlisting}
let (let*?) : bool -> (unit -> 'b sem) -> 'b sem =
  fun b f -> if b then f () else Error ("test failed")
\end{lstlisting}
The function header of @eval@ checks
whether time is up and otherwise proceeds processing the expression.

The @Varinst@ case corresponds to instantiation. It
obtains the variable's value, checks that it is a location, checks the
permission (the @let*?@ clause), obtains the storable $w$ at that
location, and checks that it is a poly closure (STPOLY). Next, it updates the
permission: if the poly closure is unrestricted, then the location
remains in the permission set, otherwise it is removed. Finally, we
allocate a new monomorphic closure, add it to the permissions, and
return the pointer as the result along with the updated store and
permissions.

The @App@ case implements (elaborated) function application.
We first apply the splitting @sp@ to @gamma@ and
evaluate subterm @e_1@ with its part of the environment and the
decremented timer @i'@. The result must be a location that we are
permitted to use. Moreover, there must be a monomorphic STCLOS stored
at that location. The permission to further use this closure  remains
in force only if the closure is unrestricted. Finally, we evaluate the
argument, then the function body, and return its result.


The @Region@ case implements a region. It obtains the address for @x@,
the suspended binding, and extends it with the intended borrow
$\BORROW$. This extension may fail if we try to take an affine borrow
of an unrestricted borrow. Next, we rebind @x@ to the borrow's
address, extend the permission accordingly, and execute the region's
body.  Finally, we withdraw the permission and return the result.

The @Borrow@ case obtains the address for @x@, checks that it is a
borrow of the correct mode $\BORROW$ and whether it is permitted to
use it. It just returns the address.

\lstDeleteShortInline@


\section{Inference}
\label{inference}

An important contribution of \affe is its principal type inference.
Our type inference algorithm is
based on the \hmx framework~\citep{DBLP:journals/tapos/OderskySW99},
a Hindley-Milner type system for a language
with constrained types where constraints are expressed in an arbitrary
theory $X$.
If $X$ has certain properties, then
\hmx guarantees principal type inference.
We apply \hmx to a concrete constraint language which we name $\CL$.
We adapt and extend \hmx's rules to support kind inference,
track linearity, and handle borrows and regions. We
formulate constraint solving and simplification algorithms for
$\CL$. Finally, we prove that the inference algorithm computes
principal types. 

\subsection{Preliminaries}

In the context of inference, it is critical to know which elements
are input and output of inference judgments.
In the following, when presenting a new judgment,
we write input parameters in \textbf{\textcolor{ForestGreen}{bold
    green}}. The remaining parameters are output parameters.

\paragraph{Usage Environments}


To determine if a variable is used in an affine manner, we track
its uses and the associated kinds. In the expression
$f\ x\ x$, $x$ is used twice. If $x$ is of type $\tau$, which is of kind $k$,
we add the constraint $\Cleq{k}{\kun}$.
To infer such constraints, our inference judgment not only
takes an environment as parameter but also returns a \emph{usage
  environment}, denoted $\Sv$, 
which summarizes usages of variables and borrows.
Usage environments are defined like normal environments.
In \cref{sdtyping}, we use relations to split environments and to
transform suspended bindings into borrows inside a region.
These relations take a constraint parameter which validates
the transformations.
In the context of inference, we define new judgments which \emph{infer}
the constraints.
\begin{itemize}[leftmargin=*,topsep=0pt]
\item $\bsplit{C}{\Sv}{\inP{\Sv_1}}{\inP{\Sv_2}}$.
  Given two usage environments $\inP{\Sv_1}$ and $\inP{\Sv_2}$,
  we return $\Sv$, the merged environment, and $C$, a set
  of constraints that must be respected.
\item $\bregion[\inP{n}]{C}{\inP{x}}{\Sv}{\inP{\Sv'}}$.
  Given a usage environment $\inP{\Sv'}$, a nesting level $\inP{n}$,
  and a variable name $\inP{x}$, we return
  $\Sv$ where the borrow binding of $x$ in $\Sv'$, if it exists,
  is replaced by
  a suspended binding. We also return the constraints $C$.
\end{itemize}
Both relations are total and non-ambiguous in term of their input
(i.e., functions), and use
the rules presented in \cref{sdtyping:split,sdtyping:regions}.
The relations used for syntax-directed typing can trivially be defined
in terms of these new relations by using constraint entailment.
All relations are fully described in \cref{typ:extra:envs}.

\paragraph{Constraint Normalization}

The \hmx framework assumes the existence of a function
``$\operatorname{normalize}$'' which takes a constraint $\inP{C}$ and a
substitution $\inP\psi$ and returns a 
simplified constraint $C'$
and an updated substitution $\unif'$.
Normalization returns a normal form such that $\unif'$ is a most general unifier.
For now, we simply
assume the existence of such a function for our constraint system
and defer details to \cref{infer:solving}.

\subsection{Type Inference}

We write $\inferW{\Sv}{(C,\unif)}{\inP{\E}}{\inP{e}}{\tau}$ when
$\inP{e}$\ has type $\tau$ in $\inP{\E}$ under the constraints $C$ and unifier $\unif$
with a usage environment $\Sv$. $\inP\E$ and $\inP{e}$ are the input parameters of our
inference algorithm.
Unlike in the syntax-directed version, $\E$ contains only regular and type bindings.
Suspended and borrow bindings can only be present in $\Sv$.
We revisit some of the syntax-directed rules presented in \cref{sdtyping}
to highlight the novelties
of our inference algorithm and the differences with the syntax-directed
system in \cref{rule:infer:envs}.
The complete type inference rules are shown in \cref{appendix:infer}.

\paragraph{Environments and Bindings}
\label{infer:envs}
\begin{figure*}[tb]
  \begin{mathpar}
    \ruleIVar
    \and
    \ruleIAbs
    \and
    \ruleIRegion
    \and
    \ruleIApp
    \and
    \inferrule{}
    { \Weaken_{\bvar{x}{\sigma}}(\Sv) =
      \text{if } (x\in\Sv) \text{ then} \Ctrue \text{else }
      \Cleq{\sigma}{\kaff_\infty}
    }
  \end{mathpar}
  \vspace{-15pt}
  \caption{Selected inference rules -- $\inferW{\Sv}{(C,\unif)}{\inP{\E}}{\inP{e}}{\tau}$}
  \label{rule:infer:envs}
  \label{rule:infer:envrules}
  \label{rule:infer:let}
\end{figure*}
In the syntax-directed system, the {\sc Var} rule ensure
that linear variables are not discarded at the \emph{leaves}.
In the inference algorithm, we operate in the opposite
direction: we collect data from the leaves and enforce linearity
at \emph{binders}. This policy is reflected in the {\sc Var$_I$} and
{\sc Abs$_I$} rules.
Typing for variables is very similar to traditional Hindley-Milner
type inference. To keep track of linearity, we record
that $x$ was used with the scheme $\schm$ by returning
a usage environment $\Sv = \{ \bvar{x}{\schm} \}$.
This usage environment is in turn used at each binder to enforce proper
usage of linear variable via the $\Weaken$ property as shown for
lambda expressions in the {\sc Abs$_I$} rule.
First, we typecheck the body of the lambda and obtain a usage
environment $\Sv_x$. As in the syntax-directed type system,
we introduce the constraint
$\Cleq{\Sv\Sdel{x}}{\kvar}$ which properly accounts for captures in
the body of the lambda expression. We then introduce the constraint
$\Weaken_{\bvar{x}{\sigma}}(\Sv)$, which fails if we try
to abandon a linear variable. 
The $\Weaken$ constraint is introduced at each binding construct.
Finally, we normalize constraints to ensure
that the inference algorithm always return the simplest possible
constraints and unifiers.

\paragraph{Splitting and Regions}
\label{infer:split}
\label{infer:regions}

Inference versions
of the {\sc App} and {\sc Region} rules
are similar to the original ones, but now \emph{return} the usage
environment $\Sv$.
As such, we use the ``inference'' version of the relations on
the environment,
$\bsplit{C}{\Sv}{\inP{\Sv_1}}{\inP{\Sv_2}}$
and $\bregion{C}{\inP{x}}{\Sv}{\inP{\Sv'}}$,
which returns the necessary constraints.
We then collect all constraints and normalize them.

\subsection{Constraints}
\label{infer:solving}

\newcommand\A{\mathcal A}
\newcommand\SC{\mathcal S}

%

To properly define our type system, we need to define $\CL$, a constraint
system equipped with an entailment relation noted $\operatorname{\vdash_e}$
and a normalizing function.
For concision, we first demonstrate the constraint solving
algorithm with an example. We then state the various properties
that make it suitable for use in the \hmx framework.
The complete constraint system is
defined in \cref{appendix:constraints}.

\subsubsection{Constraints normalization by example}

\label{solving:example}

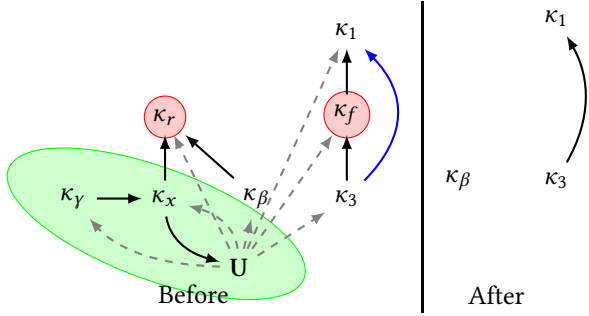
\begin{figure}[tbp]
  \centering
  \begin{tikzpicture}[node distance=6mm,xscale=1.2,yscale=0.7,.every edge/.style=[link,->,>=latex,thick]]
    \node (U) {$\kun$} ;
    \node[above left=of U] (x) {$\kvar_x$} ;
    \node[above=of x] (r) {$\kvar_r$};
    \node[left=of x] (g) {$\kvar_\gamma$} ;
    \node[right=of x] (b) {$\kvar_\beta$} ;
    \node[right=of b] (3) {$\kvar_3$} ;
    \node[above=of 3] (f) {$\kvar_f$} ;
    \node[above=of f] (1) {$\kvar_1$} ;

    \draw (x) to[bend right] (U) ;
    \draw (x) -> (r) ;
    \draw (b) -> (r) ;
    \draw (g) -> (x) ;
    \draw (3) -> (f) ;
    \draw (f) -> (1) ;

    \draw[blue] (3) to[bend right] (1);

    \node at (-0.5,-0.5) {Before};
    
    \begin{scope}[dashed,gray]
      \draw (U) to[bend right] (x) ;
      \draw (U) to (b) ;
      \draw (U) to[bend left] (g) ;
      \draw (U) to (r) ;
      \draw (U) to (1) ;
      \draw (U) to (3) ;
      \draw (U) to (f) ;
    \end{scope}
  
    \begin{scope}[on background layer]
      \node[fill=green!20,draw=green,
      inner sep=-1pt,ellipse,rotate fit=-20,fit=(U.south) (x) (g)] {};
      \node[fill=red!20,draw=red, inner sep=-2pt, circle,fit=(r)] {};
      \node[fill=red!20,draw=red, inner sep=-2pt, circle,fit=(f)] {};
    \end{scope}
  \end{tikzpicture}~
  \vrule~
  \begin{tikzpicture}[node distance=7mm,every edge/.style=[link,->,>=latex,thick]]
    \node (b) {$\kvar_\beta$} ;
    \node[right=of b] (3) {$\kvar_3$} ;
    \node[above=of 3] (f) {} ;
    \node[above=of f] (1) {$\kvar_1$} ;

    \draw (3) to[bend right] (1);

    \node at (0.5,-1.5) {After};
  \end{tikzpicture}
  \vspace{-5pt}
    \caption{Graph representing the example constraints}
    \label{example:graph}
    \label{example:graph:final}
\end{figure}

Consider the expression $\lam{f}{\lam{x}{(\app{f}{x},x)}}$.
The inference algorithm yields the following constraints:
\begin{align*}
  \E &= (\tvar_f : \kvar_f)
  (\tvar_x : \kvar_x)\dots\\
  C &= (\tvar_f \leq \gamma \tarr{\kvar_1} \beta )
  \Cand
  (\gamma \leq \tvar_x)
  \Cand
  (\beta \times  \tvar_x \leq \alpha_r)
  \Cand
  (\kvar_x \leq \kun)
\end{align*}

The first step of the algorithm uses Herbrand unification to obtain
a type skeleton. 
$$
(\gamma \tarr{\kvar_3} \beta) \tarr{\kvar_2} \gamma \tarr{\kvar_1} \beta \times  \gamma$$

In addition, we obtain the following kind constraints: 
\[\begin{aligned}
    &(\kvar_x \leq \kun)
    \Cand
    (\kvar_\gamma \leq \kvar_x)
    \Cand
    (\kvar_x \leq \kvar_r)
    \Cand
    (\kvar_\beta \leq \kvar_r)
    \Cand
    (\kvar_3 \leq \kvar_f)
    \Cand
    (\kvar_f \leq \kvar_1)
\end{aligned}\]

We translate these constraints into a relation whose graph
is shown in \cref{example:graph}.
The algorithm then proceeds as follow:
\begin{itemize}[noitemsep]
\item From the constraints above, we deduce the graph shown
  with plain arrows on the left of \cref{example:graph}.
\item We add all the dashed arrows by saturating
  lattice inequalities. For clarity, we only show $\kun$.
\item We identify the connected component circled in
  {\color{ForestGreen} green}.
  We deduce $\kvar_\gamma = \kvar_x = \kun$.
\item We take the transitive closure, which adds the
  arrow in {\color{blue} blue} from $\kvar_3$ to $\kvar_1$.
\item We remove the remaining nodes not present in the type skeleton (colored in {\color{red} red}): $\kvar_r$ and $\kvar_f$.
\item We clean up the graph (transitive reduction, remove unneeded constants, \dots),
  and obtain the graph shown on the right.
  We deduce $\kvar_3 \leq \kvar_1$.
\end{itemize}

The final constraint is thus
$$\kvar_\gamma = \kvar_x = \kun \Cand \kvar_3 \leq \kvar_1$$
If we were to generalize, we would obtain the type scheme:
$$\forall \kvar_\beta \kvar_1 \kvar_2 \kvar_3
(\gamma : \kun) (\beta : \kvar_\beta).\ %
\qual
{\Cleq{\kvar_3}{\kvar_1}}
{(\gamma \tarr{\kvar_3} \beta) \tarr{\kvar_2} \gamma \tarr{\kvar_1} \beta \times  \gamma}$$

We can further simplify this type by exploiting variance. As $\kvar_1$
and $\kvar_2$ are only used in covariant position, they can be
replaced by their lower bounds, $\kvar_3$ and $\kun$. 
By removing the unused quantifiers, we obtain a much simplified equivalent type:
$$
\forall \kvar
(\gamma : \kun).
{(\gamma \tarr{\kvar} \beta) \tarr{} \gamma \tarr{\kvar} \beta \times  \gamma}$$


\subsubsection{Properties of the constraint system}

To apply \hmx to $\CL$, normalize must compute principal normal forms
and $\CL$ must be regular.

\begin{property}[Principal normal form]
  Normalization computes principal normal forms for $\CL$, i.e.
  given a constraint $D\in\CL$, a substitution $\phi$ and
  $(C,\unif) = \normalize{D}{\phi}$,
  then $\phi\leq\unif$,
  $C \equivC \unif D$ and
  $\unif C = C$.
\end{property}

\begin{property}[Regular constraint system]
  $\CL$ is regular, ie, for $x, x'$ two types or kinds,
  $\entail{}{\Ceq{x}{x'}}$ implies
  $\fv{x} = \fv{x'}$
\end{property}

These properties
are sufficient to state that $HM(\CL)$ provides principal type inference.
The next section shows that these properties carry over to the
inference algorithm for our extension
of \hmx
with kind inference, affine types, and borrows.
This algorithm includes sound and complete constraint simplification.
In addition, we may add ``best-effort'' simplification
rules which help reduce the size of inferred signatures 
\citep{DBLP:conf/aplas/Simonet03}.

\subsection{Soundness and Principality}

The extended inference algorithm is sound
and complete with respect to our extension of \hmx.
%
The first theorem states that inference is sound
with respect to the syntax-directed type system.

\begin{theorem}[Soundness of inference]
  Given a type environment $\E$ containing only value bindings,
  $\E|_\tau$ containing only type bindings, and a term $e$:\\
  if $\inferW{\Sv}{(C,\unif)}{\E;\E_\tau}{e}{\tau}$\\
  then $\inferS{C}{\unif(\Sv;\E_\tau)}{e}{\tau}$, $\unif C = C$ and $\unif \tau = \tau$
\end{theorem}

The syntax-directed derivation holds with the usage environment $\Sv$ instead of the originally provided environment $\E$. Indeed,
$\E$ does not contain suspended and borrow bindings. Those
are discovered on the fly and recorded in $\Sv$. Type bindings
are taken directly from the syntax-directed derivation.

The second theorem states that inference is complete: for any given
syntax-directed typing derivation, our inference algorithm can find
a derivation that gives a type at least as general.


\begin{definition}[Instance relation]
  Given a constraint $C$ and two schemes
  $\schm = \forall \Multi\tvar. \qual{D}{\tau}$ and
  $\schm' = \forall \Multi\tvar'. \qual{D'}{\tau'} $.
  Then $\entail{C}{\schm \preceq \schm'}$
  iff $\entail{C}{\subst{\tvar}{\tau''} D}$
  and $\entail{C\Cand D'}{\Cleq{\subst{\tvar}{\tau''}\tau}{\tau'}}$
\end{definition}

\begin{definition}[Flattened Environment]
A flattened environment,
written as $\Eflat\E$, is the environment
where all the binders are replaced by normal ones. More formally:
\begin{align*}
  \Eflat\E
  =& \left\{\bvar{x}{\tau}\in\E \mid
    \vee \bvar{\borrow{x}}{\borrowty{k}{\tau}}\in\E
    \vee \svar{x}{\tau}^n\in\E
    \right\}
     \cup \left\{ \bvar{\tvar}{k} \mid \bvar{\tvar}{k}\in\E \right\}
\end{align*}
\end{definition}




\begin{theorem}[Principality]
  Let $\inferS{\Ctrue}{\E}{e}{\schm}$ a closed typing judgment.
  Then $\inferW{\Sv}{(C,\unif)}{\Eflat\E}{e}{\tau}$
  such that:
  \begin{align*}
    (\Ctrue,\schm_o) &= \generalize{C}{\unif\E}{\tau}
    &\entail{&}{\schm_o \preceq \schm}
  \end{align*}

\end{theorem}


\section{Metatheory}
\label{sec:metatheory}

\lstMakeShortInline[keepspaces,style=rule]@

There are several connections between the type system and
the operational semantics, which we state as a single type soundness
theorem.
The theorem relies on several standard notions like store typing
$\vdash \Store : \SE$ and agreement of the results in the value environment
with the type environment $\SE \vdash \VEnv:\E$ that we define
formally in~\cref{sec:metatheory:proofs} where we also present selected cases of the
proofs.
The non-standard part is the handling of permissions. With
$\Rawloc\Perm$ we extract the underlying raw locations from the
permissions as in $\Rawloc{\Multi\IBORROW\hspace{0.5mm}\Multi\MBORROW\hspace{0.5mm}\Loc} = \Loc$
and with $\Reach\Store\VEnv$ we transitively trace the
addresses reachable from $\VEnv$ in store $\Store$. We write
$\SE\le\SE'$ and $\Store \le \Store'$ for extending the domain of the
store type and of the store, respectively.
The permission set contains the set
of addresses that can be used during evaluation. It is managed by the
region expression as well as by creation and use of resources as
shown in \cref{sec:sem}.
We distinguish several parts of the value
environment $\VEnv$ that correspond to the different kinds of bindings in the
type environment: $\Active\VEnv$ for active entries of direct
references to linear resources, closures, etc; $\MutableBorrows\VEnv$ for
affine borrows or resources;
$\ImmutableBorrows\VEnv$ for unrestricted values including
unrestricted borrows;
and $\Suspended\VEnv$ for suspended entries. The judgment
$\SE \vdash \VEnv:\E$ is defined in terms of this structure.
We treat
$\Reach\Store\VEnv$ as a multiset to properly discuss linearity and
affinity. We use the notation $\MultiNumber x M$ for the number of
times $x$ occurs in multiset $M$.

\newcommand\resultOk[2]{
  $R#2 = \Ok{\Store#2, \Perm#2, r#2}$
}
\newcommand\resultEnv[2]{
  $\SE#1 \le \SE#2$, $\Store#1 \le \Store#2$, $\vdash \Store#2 : \SE#2$
}
\newcommand\resultPermDom[2]{
  $\Perm#2$ is wellformed and
  $\Rawloc{\Perm#2} \subseteq \Dom{\Store#2} \setminus {\Store#2}^{-1}
  (\StFreed)$.
}
\newcommand\resultReachPerm[2]{
  $\Reach0{r#2} \subseteq \Perm#2$,
  $\Reach{\Store#2}{r#2} \subseteq \Sclos{\Perm#2}$
    $\cap (\Reach{\Store#2}{\VEnv#1} \setminus
    \Reach{\Store#2}{\Suspended{\VEnv#1}}
    \cup \Dom{\Store#2} \setminus \Dom{\Store#1})$.
}
\newcommand\resultFrame[3]{
  Frame: \\
  For all $\Loc \in \Dom{\Store#1} \setminus
  \Rawloc{\Reach{\Store#3}{\VEnv#2}}$ it must be that
  \begin{itemize}
  \item $\Store#3 (\Loc) = \Store#1 (\Loc)$ and
  \item  for any $\Addr$ with $\Rawloc\Addr = \{\Loc\}$,
    $\Addr \in \Perm#1 \Leftrightarrow \Addr\in\Perm#3$.
  \end{itemize}
}
\newcommand\resultImmutables[3]{
  Unrestricted values, resources, and borrows: \\
  For all $\Addr \in
  \Reach{\Store#3}{\ImmutableBorrows{\VEnv#2}, \Suspended[\kun]{\VEnv#2}}$ with
  $\Rawloc\Addr = \{\Loc\}$, it must be that
  $\Loc\in\Dom{\Store#1}$,
  $\Store#3 (\Loc) = \Store#1 (\Loc) \ne \StFreed$
  and $\Addr\in\Perm#3$.
}
\newcommand\resultMutables[3]{
  Affine borrows and resources:\\
  For all $\Addr \in  \Reach{\Store#3}{\MutableBorrows{\VEnv#2}, \Suspended[\kaff]{\VEnv#2}}$
  with $\Rawloc\Addr = \{\Loc\}$, it must be that
  $\Loc\in\Dom{\Store#1}$. If $\Addr\ne\Loc$, then
  $\Store#3 (\Loc) \ne \StFreed$.
  If $\Addr \in \Reach{\Store#3}{\Suspended[\kaff]{\VEnv#2}}$, then
  $\Addr \in \Perm#3$.
}
\newcommand\resultSuspendedXXX[3]{
  Suspended borrows: \\
  For all $\Addr \in \Reach{\Store#3}{\Suspended{\VEnv#2}}$ with
  $\Rawloc\Addr = \{\Loc\}$ it must be that $\Addr \in \Perm#3$
  and  $\Store#3 (\Loc) \ne \StFreed$.
}
\newcommand\resultResources[3]{
  Resources: Let $\REACH#1 = \Reach{\Store#1}{\Active{\VEnv#2}}$.  Let
  $\REACH#3 =\Reach{\Store#3}{\Active{\VEnv#2}}$.
  \\
  For all $\Loc\in \REACH#1$ it must be that
  $\MultiNumber\Loc{\REACH#1}=\MultiNumber\Loc{\REACH#3}=1$,
    $\Loc\notin\Perm#3$, and if $\Store#1(\Loc)$ is a resource, then
    $\Store#3 (\Loc) = \StFreed$.
}
\newcommand\resultThinAir[2]{
  No thin air permission: \\
  $\Perm#2 \subseteq
  \Perm#1 \cup ( \Dom{\Store#2} \setminus \Dom{\Store#1})$.
}
\newcommand\assumeWellformed[1]{
  $\Perm#1$ is wellformed and $\Rawloc{\Perm#1} \subseteq
  \Dom{\Store#1} \setminus {\Store#1}^{-1} (\StFreed)$
}
\newcommand\assumeIncoming[2]{
  Incoming Resources:
  \begin{enumerate}
  \item $\forall \Loc\in \Rawloc{\Reach{\Store#1}{\VEnv#2}}$,
    $\Store#1 (\Loc) \ne \StFreed$.
  \item $\forall \Loc \in \REACH#1
    =\Rawloc{\Reach{\Store#1}{\Active{\VEnv#2},\MutableBorrows{\VEnv#2},
      \Suspended[\kaff]{\VEnv#2}}}$,
    $\MultiNumber\Loc{\REACH#1}= 1$.
  \end{enumerate}
}
\newcommand\assumeReachable[2]{
  $\Reach0{\VEnv#2} \subseteq {\Perm#1}$, 
  $\Reach{\Store#1}{\VEnv#2} \subseteq \Sclos{\Perm#1}$.
}

\begin{restatable}[Type Soundness]{theorem}{SoundnessThm}\label{thm:soundness}
  Suppose that
  \begin{enumerate}[({A}1)]
  \item $\inferS{C}{\E}{e}{\tau}$
  \item\label{item:32} $\SE \vdash \VEnv : \E$
  \item\label{item:33} $\vdash \Store : \SE$
  \item\label{item:11} \assumeWellformed{}
  \item\label{item:12} \assumeReachable{}{}
  \item $\Rawloc{\Active\VEnv}$,
    $\Rawloc{\MutableBorrows\VEnv}$,
    $\Rawloc{\ImmutableBorrows\VEnv}$, and
    $\Rawloc{\Suspended\VEnv}$ are all disjoint
  \item\label{item:15} \assumeIncoming{}{}
  \end{enumerate}
  For all $i\in\Nat$, if
  \ \lstinline[style=rule]{R' = eval} $\Store$ $\Perm$ $\VEnv$ \lstinline[style=rule]{i e}
  and $R'\ne \TimeOut$,
  then
  $\exists$ $\Store'$, $\Perm'$, $r'$, $\SE'$ such that
  \begin{enumerate}[({R}1)]
  \item \resultOk{}{'}
  \item \resultEnv{}{'}
  \item $\SE' \vdash r' : \tau$
  \item \resultPermDom{}{'}
  \item \resultReachPerm{}{'}
  \item \resultFrame{}{}{'}
  \item \resultImmutables{}{}{'}
  \item \resultMutables{}{}{'}
  \item \resultResources{}{}{'}
  \item \resultThinAir{}{'}
  \end{enumerate}
\end{restatable}

The proof of the
theorem is by functional induction on the evaluation judgment, which
is indexed by the strictly decreasing counter $i$.

The assumptions A1-A3 and results R1-R3 state the standard soundness properties
for lambda calculi with references.

The rest of the statement accounts for the substructural properties and
borrowing in the presence of explicit resource management.
%
%
%
Incoming resources are always active (i.e., not freed).
Linear and affine resources as well as suspended affine borrows have
exactly one pointer in the environment.
The Frame condition states that only store locations reachable from
the current environment can change and that all permissions outside
the reachable locations remain the same.
Unrestricted values, resources, and borrows do not change their
underlying resource and do not spend their permission.
Affine borrows and resources may or may not spend their
permission. Borrows are not freed, but resources may be freed.
Incoming suspended borrows have no permission attached to them and
their permission has been retracted on exit of their region.
A linear resource is always freed.
Outgoing permissions are either inherited from the caller or they
refer to newly created values.

\lstDeleteShortInline@


\lstMakeShortInline[keepspaces,basicstyle=\small\ttfamily]@
\section{Limitations and Extensions}

\subsection{Flow sensitivity}

The type system defined so far does not support any
form of flow sensitivity. Therefore, code patterns that rely on
subtle flow-sensitive usage of permissions and linearity will most likely not
typecheck in \lang. For example, the following merge function on linear lists
cannot be expressed directly, because matching against
\lstinline/l1/ and \lstinline/l2/ consumes both lists.

\begin{lstlisting}
let rec merge l1 l2 = match l1, l2 with
  | h1::t1, h2::t2 ->
    if &h1 < &h2 
    then h1::(merge t1 l2) (* Must expand l2 to h2::t2 here *)
    else h2::(merge l1 t2)
  | ....
\end{lstlisting}

Patterns like this require a richer logic, such as provided by
Mezzo~\citep{DBLP:phd/hal/Protzenko14}.
However,
\citet{DBLP:journals/corr/abs-1903-00982} formalize Rust's
notion of non-lexical lifetimes which partially supports such
code patterns. We believe this notion can be adapted to \lang's notion of
regions.

\subsubsection*{Non-Lexical Regions}

The notion of non-lexical lifetimes is 
a recent addition to Rust.
With this feature code is acceptable even if borrowing does not respect
lexical scoping as in this example:

\begin{lstlisting}[numbers=none]
let a = &x in (f a; g &!x)
\end{lstlisting}

This code pattern is dynamically safe because $a$ is not used after
the function call @f a@.
Here, this can be made explicit by transforming
the code to @(let a = &x in f a); g &!x@. However, this is not possible
in programs with branches who uses different dynamic patterns.
Non-lexical lifetimes (NLL) handle such a pattern by removing expressions
that do not mention $a$ from its region; in this example,  NLL removes
the last expression.
In \lang, regions are lexical and marked by the
expression $\region{b}{e}$.
During inference, kind constraints prevent escaping from
a region.

To add support for non-lexical lifetimes, we could replace the
lexical region by an annotation on each expression indicating which borrows are
live in this expression.
When exploring a subexpression, we would compare the annotations, and automatically
apply the $\textsc{Region}$ rule when they differ.
This approach is equivalent to inlining the $\textsc{Region}$ rule in all the other
rules.

Applied to the program above, only the first two expressions would be annotated
to be ``in the region associated with @&x@'', but not the last expression.
Thanks to these annotations, type checking the sequence would check
that the borrow does not escape the left-hand side (i.e., the second
expression @f a@).

\subsection{Capabilities and Identity}
\label{identity}

In \lang the tracking of linearity does
not rely on any notion of ``identity'': the type system cannot specify that two
objects are the same, simply that they share the same usage pattern with
regards to linearity.
A language like Alms~\citep{DBLP:conf/popl/TovP11}, on the other hand,
often relies on a notion of identity to express capabilities.
For instance, the Alms typing
@Array.create : int -> 'a -> \E 'b. ('a, 'b) array@ uses
@'b@ as a  unique identification of the array.
Functions such as @Array.acquire : ('a, 'b) array -> 'b cap@
are used to obtain capabilities to operate on the array.

While such uses are partially covered by borrows and regions,
a notion of identity associated to regions
would enable us to express regions directly in type signatures.
For instance, the @get_eb@ function shown in
\cref{tuto:vector} could be made safe
by creating a restricted inner region on function application,
with the signature:
@&!('k, 'a Array.t) -> int -> \E ('k' < 'k) &!('k', 'a)@

This approach relies on existential types to model identities.
At present, \lang does not support existentials as it would
forgo principal type inference.
However, existentials
are compatible with the HM(X) framework~\citep{DBLP:conf/icfp/Simonet03}
and would make a very desirable addition to \lang. Work on GADTs in
OCaml and Haskell demonstrates
that existential types can be put to use without compromising
inference in the rest of the language, by integrating unpacking
and pattern matching.

\subsection{Ad-hoc Polymorphism and Borrows}

In our formalization, we use two operators, $\borrow{x}$ and $\reborrow{x}$ to
distinguish between borrows and borrows of borrows.
Such a distinction is inconvenient for programming.
Using a typeclass-like mechanism, we can replace these operators
by a single overloaded operator, $\borrow{x}$, which expects $x$ to be @Borrowable@ and
would then desugar to the more precise operators.
A similar solution is used in Rust through the @Borrow@ and @Defer@
traits.
This approach also enables method calls on objects without
explicit borrows, such as @foo.len()@ where len expects a shared borrow.

Ad-hoc polymorphism fits demonstrably in the \hmx framework of constrained
types and preserves all properties of our language such
as principal type inference. Its soundness is orthogonal to linear types
and has been explored in the literature~\citep{DBLP:conf/fpca/OderskyWW95}.

\subsection{A Richer Region System}

\lang requires that each region is identified by an index drawn from a partial
order that is compatible with the nesting of regions.
This order can be implemented in many ways, including region variables
as often used in algebraic effects systems, existentials, etc.

For simplicity, the formalization uses the concrete implementation with 
natural numbers for indices. The proofs only rely on
the existence of a partial order and could be adapted to one of the
more abstract approaches.
In particular, \lang could reuse regions variables provided by
the ongoing work on effect systems
for OCaml~\cite{DBLP:conf/sfp/DolanEHMSW17}.

\subsection{Standard Features}

\subsubsection*{Algebraic Datatypes}

Algebraic data types are a staple of functional programming and fit nicely
in our paradigm. Indeed, it is sufficient to ensure  that the kinds of
the constructor arguments are less than or equal to the kind of the datatype.
Hence, it is forbidden to include affine elements in an unrestricted
datatype, whereas the elements in a linear list may be linear or
unrestricted.
Our prototype implements non-recursive algebraic datatypes with
pattern matching.

\subsubsection*{Branching constructs}

Our formalization of \lang does not cover conditionals. In the
typing rules, a conditional is supported as usual by checking all
branches with the same constraint and typing environment and requiring
the return types to match.
It materializes in the inference algorithm as a straightforward
(symmetric) join relation which is used for all elimination rules on
sum types.
This extension is implemented in our prototype.

\subsection{Concurrency}
\label{sec:concurrency}

While our present semantic model does not consider concurrency, some
design decisions were taken with a possible extension to concurrency
in mind. The main impact of this foresight is the distinction between
exclusive borrows and shared borrows, which materializes in the
metatheory. The intended contract of the shared 
borrow is that it can be duplicated and that the program consistently
observes the same state of the underlying resource inside its region
even in the presence of concurrency.

The exclusive borrow, on the other hand, is propagated according to
the evaluation order with the intention that any suspended binding
split of from an exclusive borrow has finished its action on the
resource before the borrow gets exercised. In the presence of
concurrency, this intended semantics of the exclusive borrow should
guarantee freedom of data races.

That is, if a thread closes over a borrow, that thread should have
terminated before the parent thread leaves the borrow's region. Rust
addresses this lifetime issue with the \lstinline/move/ qualification
for a thread which transfers ownership of the free variables to the
thread. However, moving (in Rust) only applies to the resource itself, but not
to borrows. A more discerning kind system would be needed for \lang to
enable safe sharing of synchronizable resources or borrows analogously
to Rust's \texttt{Sync} trait.

\lstDeleteShortInline@


\section{Related work}
\label{sec:related-work}




\newcommand\YC{\color{Green}}
\newcommand\NC{\color{Red}}
\newcommand\MC{\color{Orange}}
\newcommand\Y{{\YC{\ding{52}}}\xspace}
\newcommand\N{{\NC{\ding{56}}}\xspace}
\newcommand\M{{\MC{\textasciitilde}}\xspace} 
\begin{figure}[tp]
  \begin{center}

    \newcolumntype{R}[2]{%
      >{\adjustbox{angle=#1,lap=\width-(#2)}\bgroup}%
      c%
      <{\egroup}%
    }
    \newcommand*\rot{\multicolumn{1}{R{45}{1em}}}
    
    \begin{tabular}{l|c@{~}cccccccccc}
      \textbf{Language}
      & UAL
      & \rot{State}
      & \rot{Borrows}
      & \rot{Multiplicity Subsumption}
      & \rot{Multiplicity Polymorphism}
      & \rot{Identity}
      & \rot{Concurrency}
      & \rot{Escape hatch}
      & \rot{Inference} & \rot{Formalisation}
      & \rot{Basis}
      \\\hline
      System F\degree~\citep{DBLP:conf/tldi/MazurakZZ10}
      & UL & \N & \N & \N & \N & \N & \N & \N & \N & \YC Coq & System F
      \\
      Alms~\citep{DBLP:conf/popl/TovP11}
      & UA & \Y & \N & \M & \Y & \Y & \Y & \Y & \MC Local & \N & ML \\
      Quill~\citep{DBLP:conf/icfp/Morris16}
      & UL & \N & \N & \N & \Y & \N & \N & \N & \YC Principal & \MC Manual & Qual. types\\
      Lin. Haskell~\citep{DBLP:journals/pacmpl/BernardyBNJS18}
      & UL & \M & \N & \N & \Y & \N & \M & \M & \MC Non-pr. & \MC Manual & Haskell \\
      Mezzo~\citep{DBLP:phd/hal/Protzenko14,DBLP:journals/toplas/BalabonskiPP16}
      & UA & \Y & \M & \M & \Y & \Y & \Y & \Y & \MC Local & \YC Coq & ML \\
      \hline
      Rust~\citep{rust,DBLP:journals/pacmpl/0002JKD18}
      & UA & \Y & \Y & \Y & \M & \N & \Y & \Y & \MC Local & \YC Coq & --- \\
      Plaid~\citep{DBLP:conf/oopsla/AldrichSSS09,DBLP:journals/toplas/GarciaTWA14}
      & UA & \Y & \N & \Y & \Y & \Y & \N & \Y & \N & \MC Manual & Java \\
      \hline
      \lang
      & UAL & \Y & \Y & \Y & \Y & \N & \N & \M & \YC Principal & \MC Manual & ML/HM(X)
      \\
    \end{tabular}
  \end{center}
  \caption{Comparison matrix}
  \label{fig:comparison-matrix}
\end{figure}

The comparison matrix in \cref{fig:comparison-matrix} gives an
overview over the systems discussed in this section.
Each column indicates whether a feature is present (\Y), absent (\N), or
partially supported (\M), i.e., if the feature is limited or can
only be obtained through a non-trivial encoding.
Features are selected according to their relevance for type-based resource
management and programmer convenience.

The column UAL
specifies the substructural features supported (Unrestricted, Affine,
Linear). The columns ``State'' and ``Borrows'' indicate support for the respective
feature. In an ideal world, the presence of linearity and state
indicates that the system is able to support safe manual memory
management as linearity enforces manual deallocation. True affinity
and state only works with garbage collection, which eventually
automatically finalizes an object no longer referenced. In practice, this
distinction is often watered down. For example, Rust automatically
destructs objects at the end of their lifetime, creating the
illusion of affinity while the low-level code is strictly linear.
However, there are ways to consume an object at the source
level without invoking its destructor (using
\lstinline/mem::forget/)\footnote{See
  \url{https://doc.rust-lang.org/nomicon/leaking.html} which contains
  further examples and discussion. Thanks to
  Derek Dreyer and Ralf Jung for pointing this out.}
where the high-level code exhibits linearity, but the low-level code
is affine.

``Multiplicity Subsumption'' indicates that unrestricted elements 
can be promoted to affine and then linear. This promotion applies to
objects, resources, borrows, and closures. ``Multiplicity Polymorphism'' refers to polymorphism over
substructural features: a function can be parameterized over the multiplicity
restriction of an object. For instance, the type of function
composition should express that applies to functions with linear,
affine, and unrestricted multiplicity and returns a function with the
same multiplicity. 
``Identity'' indicates that the language supports a notion
of identity, usually through existential types, as described in
\cref{identity}.
``Concurrency'' indicates whether the language supports
concurrency. For example, the implementation of Linear Haskell
supports state and concurrency, but its theory covers neither.
``Escape hatch'' indicates whether a programmer can (locally) opt out of
resource management through language-integrated means such
as Rust's \texttt{unsafe}. Partial support ``\M'', in the case of \lang
for instance, indicates that this feature is available, but not formalized.
Type ``inference''  can be local, principal, or non-principal (if the
inferred type is not necessarily the most general one). ``Formalization''
refers to the existence of a formal semantics and type soundness
proof. ``Basis'' indicates the heritage or inspiration of the language.

\subsection{Substructural type-systems in functional languages}

Many systems propose combinations of
functional programming and linear types in a practical setting.
The goal of \lang is to combine key ingredients
from these proposals while still preserving
complete type inference.
Many of the following languages support linear or affine types, but rarely
both. In many cases, it is easy to adapt a system to support both, as
\lang does.
None of the following languages support borrows.

System F\degree~\citep{DBLP:conf/tldi/MazurakZZ10}
extends System F with kinds to distinguish
between linear and unrestricted types.
The authors provide
a linearity-aware semantics with a soundness proof.
Unlike \lang, System F\degree{} does not allow
quantification over kinds which limits its expressivity. For instance, it
does not admit a most general type for function composition.
Being based on System F, it does not admit
principal type inference.

Quill~\citep{DBLP:conf/icfp/Morris16} is a Haskell-like language with linear
types.
Quill does not expose a kind language, but
uses the framework of qualified types to govern linearity annotations on arrows.
Its type inference algorithm is proven sound and complete.
\lang infers type signatures for all Quill examples, but often with
simpler types because Quill does not support subkinding.
Quill comes with a linearity-aware semantics and soundness proof.
Quill does not support borrows.


Alms~\citep{DBLP:conf/popl/TovP11} is an ML-like language with rich, kind-based
affine types and ML modules, similar to \lang.
Alms examples often rely on existential types to track the identity
of objects. For instance, consider the signature
\lstinline/Array.create : int -> 'a -> \E 'b. ('a, 'b) array/ where
\lstinline/'b/ uniquely identifies the array.
Due to the reliance on existentials, Alms does not support complete type inference.
Furthermore, Alms does not support borrows and often relies
on explicit capability passing.
In our experience, \affe's limited support for existential types through
regions is sufficient to express many of Alms' examples and leads to
a more convenient programming style for imperative code.
Alms kind structure features unions, intersections and dependent kinds while
\lang uses constrained types.
We believe most of Alms' kind signatures can be expressed equivalently in
our system: for instance the pair type constructor
has kind $\Pi\alpha\Pi\beta. \langle\alpha\rangle \sqcup \langle\beta\rangle$
(where $\alpha$ and $\beta$ are types and $\Pi$ is the dependent function)
in Alms compared to $\kvar\to\kvar\to\kvar$ in \lang thanks
to subkinding.
Finally, Alms provides excellent support for abstraction through
modules by allowing to keep some type unrestricted inside a module, but
exposing it as affine. \lang supports
such programming style thanks to subsumption.

The goal of Linear Haskell~\citep{DBLP:journals/pacmpl/BernardyBNJS18}
(LH) is to retrofit linear types to Haskell.
Unlike the previously discussed approaches, LH relies on ``linear
arrows'', written $\multimap$ as in linear logic, 
which are functions that \emph{use} their argument exactly once.
This design is easy to retrofit on top of an existing compiler
such as GHC, but has proven quite controversial\footnote{
  See the in-depth discussion attached to the GHC proposal for LH on GitHub: \url{https://github.com/ghc-proposals/ghc-proposals/pull/111\#issuecomment-403349707}.}.
Most relevant to \lang:
\begin{itemize}[leftmargin=*]
\item LH does not admit subtyping for arrows and requires
  $\eta$-expansion to pass unrestricted functions in linear
  contexts. This approach is acceptable in a non-strict language such as
  Haskell but changes the semantics in a strict setting.
\item
  While the LH paper specifies a full type system along with a
  linearity-aware soundness proof, there is neither formal description of
  the type inference algorithm nor a proof of the properties of inference.
  Subsequent work~\cite{DBLP:journals/corr/abs-1911-00268}
  formalizes the inference for rank 1 qualified-types.
  However, there is an implementation of the inference as part of GHC.
\item
  LH promotes a continuation-passing style with functions such as
  \lstinline/withFile : path -> (file ->. Unrestricted r) ->. r/
  to ensure linear use of resources. This style leads to problems with
  placing the annotation on, e.g., the IO monad.
  \lang follows System F\degree, Quill, and Alms, all of which support
  resource handling in direct style, where types themselves are
  described as affine or linear. (Of course, continuation-passing
  style is also supported.)
  We expect that the direct approach eases modular reasoning about linearity.
  In particular, using abstraction through modules,
  programmers only need to consider the module
  implementation to ensure that linear
  resources are properly handled.
\end{itemize}

Mezzo~\citep{DBLP:phd/hal/Protzenko14,DBLP:journals/toplas/BalabonskiPP16} is an ML-like language
with a rich capability system which is able to encode numerous
properties akin to separation logic~\citep{DBLP:conf/lics/Reynolds02}.
Mezzo explores the  boundaries of the design space of type systems for
resources. Hence, it is more expressive than \lang, but
much harder to use. The Mezzo typechecker relies on explicit
annotations and it is not known whether type inference for Mezzo is possible.

\citet{DBLP:journals/corr/abs-1803-02796} presents
an extension of OCaml for resource management in the style of C++'s RAII
and Rust's lifetimes. This system assumes
the existence of a linear type system and develops the associated compilation
and runtime infrastructure. We believe our approach is
complementary and aim to combine them in the future.

\subsection{Other substructural type-systems}

\lang uses borrows and regions
which were initially developed in the context of linear and affine
typing for  imperative and
object-oriented
programming~\citep{DBLP:conf/popl/BoylandR05,DBLP:conf/pldi/GrossmanMJHWC02}.

Rust~\citep{rust} is the first
mainstream language that builds on the concepts of borrowing and ownership
to enable safe low-level programming.
\lang is inspired by Rust's borrowing system and transfers some of its
ideas  to a functional setting with type inference, garbage collection, and
an ML-like module system.
Everything is affine in Rust and marker traits like \lstinline/Copy/,
\lstinline/Send/, and \lstinline/Sync/ are used to modulate the characteristics 
of types. \lang relies on kinds to express substructural properties of
types and marker traits may be considered as implementing a
fine-grained kind structure on Rust types.
Rust's lifetime system is more explicit and more expressive than
\lang's regions.
While Rust provides partial lifetime inference, it
does not support full type inference.
Moreover, Rust programmers have full control over memory allocation
and memory layout of objects; they can pass arguments by value or by reference.
These features are crucial for the efficiency goals of Rust.
In contrast, \lang is garbage collected, assumes a uniform object
representation, and all arguments are passed by reference. This choice
forgoes numerous  
issues regarding interior mutability and algebraic data types.
In particular, it
allows us to easily nest mutable references inside objects, regardless
whether they are linear or unrestricted.

In Rust, programmers can implement their low-level abstractions by
using unsafe code fragments. Unsafe code is not typechecked with the
full force of the Rust type system, but with a watered down version
that ignores ownership and lifetimes. This loophole is needed to implement
datastructures like doubly-linked lists or advanced concurrency
abstractions. When unsafe code occurs as part of a 
function body, the Rust typechecker leaves the adherence of the unsafe
code to the function's type signature as a proof obligation to
the programmer. The RustBelt project
\cite{DBLP:journals/pacmpl/0002JKD18} provides a formal 
foundation for creating such proofs by exhibiting a framework for
semantic soundness of the Rust type system in terms of a low-level
core language that incorporates aspects
of concurrency (i.e., data-race freedom). Similar proof obligations
would be needed in \lang to check that an implementation of the module
types or the type of fold shown in \cref{motivation} matches the
semantics of the typings. We aim to develop a suitable framework for
this task for \lang. At present, the metatheory of \lang does not
cover concurrency.

\citet{DBLP:journals/corr/abs-1903-00982}
formalize Rust's ownership discipline from a source-level
perspective. Their approach is purely syntactic and is therefore not
able to reason about unsafe fragments of Rust code. However, their
flow-sensitive type discipline enables soundness proofs for
non-lexical lifetimes, which have been adopted in Rust, but cannot be
expressed in \lang at present.

Vault~\citep{DBLP:conf/pldi/DeLineF01}
and Plaid~\citep{DBLP:conf/oopsla/AldrichSSS09,DBLP:journals/toplas/GarciaTWA14}
leverage typestate and capabilities
to express rich properties in objects and protocols.
These systems are designed for either low-level or object-oriented
programming and do not immediately lend themselves to a more functional
style. While these systems are much more
powerful than \affe's, they require programmer annotations
and do not support inference.
It  would be interesting to extend \lang with limited
forms of typestate as a local, opt-in feature to provide
more expressivity at the cost of inference.

\subsection{Type-system features}
\lang relies on constrained types
to introduce the kind inequalities required for linear types.
\hmx~\citep{DBLP:journals/tapos/OderskySW99} 
allows us to use constrained types in an ML-like language with complete
type inference.
\hmx has been shown to be compatible with subtyping,
bounded quantification and existentials~\citep{DBLP:conf/icfp/Simonet03},
GADTs~\citep{DBLP:journals/toplas/SimonetP07},
and there exists a syntactic soundness proof~\citep{DBLP:journals/entcs/SkalkaP02}.
These results make us confident that the system developed in \lang
could be applied to larger and more complex languages such as OCaml
and the full range of features based on ad-hoc polymorphism.

\lang's  subtyping discipline is similar
to structural subtyping, where the only subtyping (or here, subkinding)
is at the leaves.
Such a discipline is known to be friendly to inference and has been used in many
contexts, including OCaml, and has been combined
with constraints~\citep{DBLP:journals/tapos/OderskySW99,DBLP:conf/sas/TrifonovS96}.
It also admits classical simplification rules
\citep{DBLP:conf/aplas/Simonet03,DBLP:conf/popl/PottierS02} which we partially use
in our constraint solving algorithm.
\affe's novelty is a kind language
sufficiently simple to make
all simplification rules complete, which allows us to keep type signatures simple.


\section{Conclusions}
\label{sec:conclusions}

\lang is an ML-like language extended with sound handling of linear and affine resources. Its main novel feature is the combination of full type inference and a practically useful notion of shared and exclusive borrowing of linear and affine resources.
Although the inferred types are much richer internally than plain ML types, most of that complexity can be hidden from user-level programmers. On the other hand, programmers of libraries dealing with resources have sufficient expressiveness at their fingertips to express many resource management schemes.

The main restriction of the current system is that the lifetime of borrows is determined by lexical scoping. Overcoming this restriction is subject of future work and will probably require extending the type system by some notion of effect, which is currently discussed in the OCaml community. 
Moreover, other systems rely on existential types for extra expressiveness. We chose not to include existentials to preserve complete type inference, but our design can be extended in this direction.
Finally, our matching construct is very simplistic.
Our implementation supports full algebraic data types and we believe it
can be further extended to support manipulating borrows of data-structures and
internal mutability.


\begin{acks}
  This material is based upon work supported by the German Research
  Council, \grantsponsor{sponsor01}{DFG}{https://www.dfg.de/}, project reference number
  \grantnum{sponsor01}{TH 665/11-1}. We are indebted to the anonymous reviewers
  for their thoughtful and constructive comments.
\end{acks}

\bibliography{biblio}

\clearpage
\appendix

\section{Further Examples}
\label{sec:extra-example}
\lstMakeShortInline[keepspaces,basicstyle=\small\ttfamily]@
\subsection{A session on linearity}
\label{sec:session-linearity}

Session typing \cite{Honda1993,DBLP:conf/esop/HondaVK98} is a type
discipline for checking protocols statically. A session type is
ascribed to a communication channel and describes
a sequence of interactions. For instance, the type
\lstinline{int!int!int?end} specifies a protocol
where the program must send two integers, receive an integer, and then
close the channel.
In this context, channels must be used linearly, because every use of
a channel ``consumes'' one interaction and  changes the type of the
channel. That is, after sending two integers on the channel,
the remaining channel has type \lstinline{int?end}.
Here are some standard type operators for session types \SType:
\begin{center}
  \begin{tabular}[t]{rl}
    $\tau ! \SType$ & Send a value of type $\tau$ then continue with $\SType$.\\
    $\tau ? \SType$& Receive a value of type $\tau$ then continue with $\SType$.\\
    $\SType \oplus \SType'$& Internal choice between protocols $\SType$ and $\SType'$.\\
    $\SType \operatorname{\&} \SType'$
                    & Offer a choice between protocols $\SType$ and $\SType'$.
  \end{tabular}
\end{center}

\citet{DBLP:journals/jfp/Padovani17}  has shown how to encode this
style of session typing in ML-like languages, but his implementation
downgrades linearity to a run-time check for affinity. Building on that
encoding we can provide a safe API in \lang that statically enforces
linear handling of channels:
\begin{lstlisting}
type 'S st : lin (*@\label{line:session1}*)
val receive: ('a ? 'S) st -> 'a * 'S st (*@\label{line:session2}*)
val send : ('a ! 'S) st -> 'a -{lin}> 'S st
val create : unit -> 'S st * (dual 'S) st
val close : end st -> unit
\end{lstlisting}

Line~\ref{line:session1} introduces a parameterized abstract
type \lstinline{st} which is linear as indicated
by its kind \lstinline{lin}. Its low-level implementation would wrap a
handle for a socket, for example. The \lstinline{receive}  operation
in Line~\ref{line:session2} takes a channel that is ready to receive a
value of type \lstinline{'a} and returns a pair of the value a the
channel at its residual type \lstinline{'S}. It does not matter
whether \lstinline{'a} is restricted to be linear, in fact
\lstinline{receive} is polymorphic in the kind of \lstinline{'a}.
This kind polymorphism is the
default if no constraints are specified.
The \lstinline{send} operation takes a linear channel and returns a
single-use function that takes a value of type \lstinline{'a} suitable
for sending and returns the channel with updated type.
The \lstinline{create} operation returns a pair of channel
endpoints. They follow dual communication protocols, where the
\lstinline{dual} operator swaps sending and receiving operations.
Finally, \lstinline{close} closes the channel.

In \cref{fig:sessiontype} we show how to use these primitives to
implement client and server for an addition service.
No linearity annotations are needed in the code, as all linearity
properties can be inferred from the linearity of the \texttt{st} type.

The inferred type of the server, \lstinline{add_service}, is
\lstinline{(int ! int ! int ? end) st -> unit}.
The client operates by sending two messages
and receiving the result.
This code is polymorphic in both argument and return types, so it could
be used with any binary operator.
\begin{figure}[!h]
  \begin{subfigure}[t]{0.5\linewidth}
\begin{lstlisting}
let add_service ep =
  let x, ep = receive ep in
  let y, ep = receive ep in
  let ep = send ep (x + y) in
  close ep
# add_service : (int ! int ! int ? end) st -> unit
\end{lstlisting}
    \vspace{-5pt}
    \caption{Addition server}
  \end{subfigure}
  \hfill
  \begin{subfigure}[t]{0.45\linewidth}
\begin{lstlisting}
let op_client ep x y =
  let ep = send ep x in
  let ep = send ep y in
  let result, ep = receive ep in
  close ep;
  result
# op_client : 
   ('a_1 ? 'a_2 ? 'b ! end) st -> 'a_1 -{lin}> 'a_2 -{lin}> 'b
\end{lstlisting}
    \vspace{-5pt}
    \caption{Binary operators client}
  \end{subfigure}
  \vspace{-10pt}
  \caption{Corresponding session type programs in \lang}
  \label{fig:sessiontype}
\end{figure}
Moreover, the \lstinline/op_client/ function can be partially applied
to a channel. Since the closure returned by such a partial application
captures the channel, it can only be used once.  This restriction is
reflected by the arrow of kind \lstinline{lin}, \lstinline/-{lin}>/,
which is the type of a single-use function.
The general form of
arrow types in \lang is \lstinline/-{k}>/, where \lstinline/kk/ is a
kind that restricts the number of uses of the function.  For
convenience, we shorten \lstinline/-{un}>/ to \lstinline/->/.  \lang
infers the
single-use property of the arrows without any user
annotation. In fact, the only difference between the
code presented here and Padovani's examples
\cite{DBLP:journals/jfp/Padovani17} is the kind annotation on the
type definition of \lstinline/st/.

To run client and server, we can create a channel and apply
\lstinline{add_service} to one end and \lstinline{op_client} to the other.
Failure to consume either channel endpoints (\lstinline/a/ or \lstinline/b/)
would result in a type error.
\begin{lstlisting}
let main () =
  let (a, b) = create () in
  fork add_service a;
  op_client b 1 2
# main : unit -> int
\end{lstlisting}



\subsection{Pool of linear resources}
\label{tuto:pool}

We present an interface and implementation of a pool of linear resources where the
extended scope of the region enforces proper use of the resources.

\cref{intf:pool} contains the interface of the @Pool@ module.
A pool is parameterized by its content. The kind of the pool
depends on the content: linear content implies
a linear pool while unrestricted content yields an unrestricted pool.
The functions @Pool.create@ and @Pool.consume@
build/destroy a pool given creators/destructors for the elements
of the pool.
The function @Pool.use@ is the workhorse of the API, which
borrows a resource from the pool to a callback.
It takes a shared borrow of a pool (to enable concurrent access) and a
callback function.
The callback receives a exclusive borrow of an arbitrary resource from the pool.
The typing of the callback ensures
that this borrow is neither captured nor returned by the function.

This encapsulation is implemented with a universally quantified \emph{kind index variable} $r$.
The signature prescribes the type @&!(aff_r1,'a_1)@ for the
exclusive borrow of the resource with an affine kind at region nesting $r+1$. The return
type of the callback is constrained to kind @'k_2 <= aff_r@
so that the callback certainly cannot return the borrowed argument.
In a specific use of @Pool.use@, the index $r$ gets unified
with the current nesting level of regions so that the region for the
callback effectively gets ``inserted'' into the lexical nesting at the callsite.
\cref{ex:pool} shows a simple example using the @Pool@ module.

The implementation in \cref{impl:pool} represents a bag of resources
using a concurrent queue with atomic add and remove operations.
The implementation of the @Pool.create@ and @Pool.consume@
functions is straightforward.
The function @Pool.use@ first draws
an element from the pool (or creates a fresh element),
passes it down to the callback function @f@, and returns
it to the pool afterwards.
For clarity,
we explicitly delimit the region in \cref{line:pool:region} to ensure that
the return value of @f &!o@ does not capture @&!o@.
In practice, the type checker inserts this region automatically.

\begin{figure}[tp]
  \centering
  \begin{subfigure}[t]{0.5\linewidth}
\begin{lstlisting}
type ('a:'k) pool : 'k
create : (unit -> 'a) -> 'a pool
consume : ('a -> unit) -> 'a pool -> unit
use : ('a_1:'k_1),('a_2:'k_2),('k_2 <= aff_r) =>
  &('a_1 pool) -> (&!(aff_r1,'a_1) -{lin}> 'a_2) -{'k_1}> 'a_2
\end{lstlisting}
    \vspace{-8pt}
    \caption{Signature}
    \label{intf:pool}

\begin{lstlisting}
(*Using the pool in queries.*)
let create_user pool name =
  Pool.use &pool (fun connection ->
    Db.insert "users" [("name", name)] connection)

let uri = "postgresql://localhost:5432"
let main users =
  (*Create a database connection pool.*)
  let pool = Pool.create (fun _ -> Db.connect uri) in
  List.parallel_iter (create_user &pool) users;
  Pool.consume (fun c -> Db.close c)
\end{lstlisting}
    \vspace{-8pt}
    \caption{Example of use}
    \label{ex:pool}
  \end{subfigure}
  \hfill
  \begin{subfigure}[t]{0.45\linewidth}
\begin{lstlisting}
type ('a:'k) pool : 'k =
  { spawn: unit -> 'a; queue: 'a CQueue.t }
let create spawn = 
  { spawn ; queue = CQueue.create () }
let consume f c = CQueue.iter f c.queue
let use { spawn ; queue } f =
  let o = match CQueue.pop &queue with
    | Some x -> x
    | None () -> spawn ()
  in
  let r = {| f &!o |} in(*@\label{line:pool:region}*)
  CQueue.push o &queue;
  r
\end{lstlisting}
    \vspace{-8pt}
    \caption{Implementation}
    \label{impl:pool}
  \end{subfigure}

  \caption{The \texttt{Pool} module}
  \label{fig:pool}
\end{figure}


\lstDeleteShortInline@

\section{Automatic region annotation}
\label{appendix:regionannot}

We now define our automatic region annotation
which is presented in \cref{regionannot}.
First, we extend the region annotation to $\region{S}{E}$ where $S$ is
a map from variables to borrow indicator $b$. This annotation, defined below, is equivalent to nested
region annotations for each individual variable.
\begin{align*}
  \region{\Sone x\BORROW;S}{e} &= \region{\Sone x\BORROW}{\region{S}{e}}& \region{\emptyset}{e} &= e\\
\end{align*}
\begin{figure*}[!tb]
  \centering
  \centering
\begin{tabular}
  {@{}>{$}c<{$}@{ $\oplus$ }
  >{$}c<{$}@{ $=$ }
  >{$}c<{$}@{ $,$ }
  >{$}c<{$}@{ $,$ }
  >{$}c<{$}l}
  {\Sone{x}{b}}
  &{\Cempty}
  &\Cempty&\Sone{x}{b}&\Cempty
  &AnnotRegion-Left\\
  {\Cempty}
  &{\Sone{x}{b}}
  &\Cempty&\Sone{x}{b}&\Cempty
  &AnnotRegion-Right\\
  {\Sone{x}{\IBORROW}}
  &{\Sone{x}{\IBORROW}}
  &\Cempty&\Sone{x}{\IBORROW}&\Cempty
  &AnnotRegion-Immut\\
  {\Sone{x}{\IBORROW}}
  &{\Sone{x}{\MBORROW}}
  &\Sone{x}{\IBORROW}&\Sone{x}{\MBORROW}&\Cempty
  &AnnotRegion-MutLeft\\
  {\Sone{x}{\MBORROW}}
  &{\Sone{x}{b}}
  &\Sone{x}{\MBORROW}&\Cempty&\Sone{x}{b}
  &AnnotRegion-MutRight
\end{tabular}

\hrulefill
\begin{mathpar}
  \inferrule
  { e = \borrow{x} \mid \reborrow{x}}
  { \Rannot{e}{e}{\Sone{x}{b}} }

  \inferrule{e = c\ |\ x}
  { \Rannot{e}{e}{\Sempty} }

  \inferrule
  { \forall i,\ \Rannot[n+1]{e_i}{e'_i}{B_i} \\
    \getBorrows{B_1}{B_2}{S_1,S,S_2}
  }
  { \Rannot{\app{e_1}{e_2}}{\app{\region[n+1]{S_1}{e'_1}}{\region[n+1]{S_2}{e'_2}}}{S} }

  \inferrule
  { \forall i,\ \Rannot[n+1]{e_i}{e'_i}{B_i} \\
    \getBorrows{B_1}{(B_2\Sdel{x})}{S_1,S,S_2} \\
    S'_2 = S_2\Sunion B_2\Sonly{x}
  }
  { \Rannot
    {\letin{x}{e_1}{e_2}}
    {\letin{x}{\region[n+1]{S_1}{e'_1}}{\region[n+1]{S'_2}{e'_2}}}{S} }

  \inferrule[Rewrite-Lam]
  { \Rannot[n+1]{e}{e'}{B} \\
    B_x = B\Sonly{x}
  }
  { \Rannot{\lam{x}{e}}{\lam{x}{\region[n+1]{B_x}{e'}}}{B\Sdel{x}} }

  \inferrule
  { \forall i,\ \Rannot[n+1]{e_i}{e'_i}{B_i} \\
    \getBorrows{B_1}{(B_2\Sdel{x,y})}{S_1,S,S_2} \\
    S'_2 = S_2\Sunion B_2\Sonly{x,y}
  }
  { \Rannot
    {\matchin{x,y}{e_1}{e_2}}
    {\matchin{x,y}{\region[n+1]{S_1}{e'_1}}{\region[n+1]{S'_2}{e'_2}}}{S} }

  \inferrule[Rewrite-Region]
  { \Rannot[n+1]{e}{e'}{B} }
  { \Rannot{\regionS{e}}{\region[n+1]{B}{e'}}{\Sempty} }

  \inferrule[Rewrite-Pair]
  { \forall i,\ \Rannot[n+1]{e_i}{e'_i}{B_i} \\
    \getBorrows{B_1}{B_2}{S_1,S,S_2}
  }
  { \Rannot
    {\introPair{e_1}{e_2}}
    {\introPair{\region[n+1]{S_1}{e'_1}}{\region[n+1]{S_2}{e'_2}}}
    {S} }

  \inferrule[AnnotRegion]{
    \forall i,\ \getBorrows{b_i}{b'_i}{S^l_i,S^m_i,S^r_i}
  }{
    \getBorrows{\Multi[i]{b}}{\Multi[i]{b'}}
    {\cup_i S^l_i, \cup_i S^m_i, \cup_i S^r_i}
  }

  \inferrule[Rewrite-Top]
  { \Rannot[1]{e}{e'}{S} }
  { \RannotT{e}{\region[1]{S}{e'}} }
\end{mathpar}


  \caption{Automatic region annotation --- $\RannotT{\inP{e}}{e'}$}
  \label{fig:region-annotation}
\end{figure*}
\Cref{fig:region-annotation} define a rewriting relation $\RannotT{e}{e'}$
which indicates that an optionally annotated term $e$ can be rewritten
as a fully annotated term $e'$.
Through the rule \textsc{Rewrite-Top}, this is defined
in term of an inductively defined relation
$\Rannot{e}{e'}{S}$ where $n$ is the current nesting and $S$ is a set of
variable that are not yet enclosed in a region.
The base cases are constants, variables and borrows.
The general idea is to start from the leaves of the syntax tree, create a
region for each borrow, and enlarge the region as much as possible.
This is implemented by a depth-first walk of the syntax
tree which collects each variable that has a corresponding borrow.
At each step, it rewrites the inner subterms,
consider which borrow must be enclosed by a region now, and
return the others for later enclosing. Binders force immediate
enclosing of the bound variables, as demonstrated in rule \textsc{Rewrite-Lam}.
For nodes with multiple children, we
use a scope merge operator to decide if regions should be placed and where.
This is shown in rule \textsc{Rewrite-Pair}.
The merge operator, written $\getBorrows{B_l}{B_r}{(S_l,S,S_r)}$, takes
the sets $B_l$ and $B_r$ returned by rewriting the subterms
and returns three sets: $S_l$ and $S_r$ indicates the variables
that should be immediately enclosed by a region on the left and right
subterms and $S$ indicates the set of the yet-to-be-enclosed variables.
As an example, the rule \textsc{AnnotRegion-MutLeft} is applied
when there is an shared borrow and a exclusive borrow. In that case, a
region is created to enclose the shared borrow, while the exclusive
borrow is left to be closed later. This is coherent with the rules
for environment splitting and suspended bindings from \cref{sdtyping}.
Explicitly annotated regions are handled specially through
rule \textsc{Rewrite-Region}. In that case, we assume that all inner
borrows should be enclosed immediately.

\section{Constraints}
\label{appendix:constraints}

We place our constraint system in a more general setting.
We define the constraint solver in terms of an arbitrary commutative bounded
lattice $(\mathcal L, \lk_\Lat)$, i.e.,
a lattice which has a minimal and a maximal element ($l^\top$ and $l^\bot$)
and where meet and joins are commutative.
We write lattice elements as $l$ and $\glb_i l_i$ (resp. $\lub_i l_i$)
for the greatest lower bound (resp. least upper bound) in $\mathcal L$.
The lattice for \lang (see \cref{sdtyping}) is a bounded lattice with
$l^\top = \klin_\infty$ and $l^\bot = \kun_0$.


\begin{figure}[tb]
  \centering
  \begin{align*}
    C &::= \Cleq{\tau_1}{\tau_2}
        \mid \Cleq{k_1}{k_2}
        \mid C_1 \Cand C_2
        \mid \Cproj{\tvar}{C}
        \mid \Cproj{\kvar}{C}
  \end{align*}
  \caption{The constraint language}
  \label{grammar:constraint}
  \input{infer/entails}
  \caption{Base entailment rules -- $\entail{\inP{C}}{\inP{D}}$ }
  \label{rules:entail}
\end{figure}

Let $\CL$ be the set of constraints in such a lattice $\mathcal L$.
The full grammar of constraints is shown in \cref{grammar:constraint}.
Constraints are made of kind inequalities, conjunctions and
projections along with type unification
constraints. Since types might contain kinds (for instance, on the arrows),
type unification is oriented and written as $\leq$.
For simplicity, we consider all type constructors
invariant in their parameters
and define $\Ceq{\tau}{\tau'}$ as $\Cleq{\tau}{\tau'} \wedge \Cleq{\tau'}{\tau}$.

Entailment is denoted by $\entail{C}{D}$,
where $D$ is a consequence of the constraints $C$.
We say that $C$ and $D$ are equivalent, $C \equivC D$,
when $\entail{C}{D}$ and $\entail{D}{C}$.
%

We directly reuse the following definitions from \hmx.

\begin{property}[Cylindric constraint system]
  A cylindric constraint system is a constraint system
  such that, for any constraint $C$:
\begin{align*}
  \entail{C&}{\Cproj{x}{C}}
  &\entail{C}{D} &\implies \entail{\Cproj{x}C}{\Cproj{x}D}\\
  \Cproj{x}{(C\wedge \Cproj{x}D)}&\equivC \Cproj{x}{C} \wedge \Cproj{x}D
  & \Cproj{x}{\Cproj{y}D}&\equivC \Cproj{y}{\Cproj{x}D}
\end{align*}
\end{property}

\begin{property}[Term rewriting system]
A term rewriting system is a system where, for every
types $\tau$,$\tau'$, there exists an equality predicates $\Ceq{\tau}{\tau'}$
which is symmetric, reflexive, transitive, stable under substitution and such that,
for any predicate $P$:
\begin{align*}
  &\entail{\Ceq{x}{y}\wedge \Cproj{x}C \wedge \Ceq{x}{y}}{C}\\
  &\subst{x}{\tau}{P} \equivC \Cproj{x}P \wedge \Ceq{x}{\tau}
                        \text{ where } x\notin\fv{\tau}
\end{align*}
\end{property}

\begin{definition}[Constraint system with lattice]
$\CL$ is defined as the smallest cylindric term constraint system that
satisfies the axiom shown in \cref{rules:entail}.
\end{definition}


We define the set of solved formed
$\mathcal S$ as the quotient set of $\CL$ by $\equivC$.
We will show later that such constraints are in fact only composed of
kind inequalities, and thus correspond to the syntactic constraints
used in type and kind schemes.
We now define our normalization procedure $\normalize{C_0}{\unif_0}$, where
$C_0\in \CL$ is a set of constraints and $\unif_0$ is a substitution.
It returns a constraint $C \in \mathcal S$ in
solved form and a unifier $\unif$.
The main idea of the algorithm is to first remove all the type equalities
by using regular Herbrand unification. After that, we only have
a set of inequalities among kinds, which we can consider as a relation.
We can then saturate the relation,
unify all kinds that are in the same equivalence classes to obtain
a most general unifier on kind variables,
remove all existentially quantified variables and
then minimize back the relation and apply various
simplification rules to make the resulting type easier to understand to users.

More precisely, we apply the following steps:
\begin{enumerate}
\item Solve all type equality constraints through Herbrand unification and
  gather all existential quantifications at the front of the constraint.
  We obtain a constraint $C^k = \exists \Multi\kvar,\ \Cleq{k_j}{k'_j}_j$ and
  a substitution $\unif_\tau$.

  We write $\mathcal R$ for the relation $\Cleq{k_j}{k'_j}_j$,
  $\mathcal G$ the underlying directed graph and $V$ its vertices.

\item Saturate the lattice equalities in $\mathcal R$.

  More precisely, for each kind variable $\kvar \in V$,
  for each constant $l_i$ (resp. $l_j$) such that
  there is a path from $l_i$ to $\kvar$ (resp. from $\kvar$ to $l_j$) in $\mathcal G$,
  add an edge from $\lub l_i$ to $\kvar$
  (resp. from $\kvar$ to $\glb l_j$).
  This step is well defined since $\mathcal L$ is a bounded lattice
  and $\lub\emptyset$ and $\glb\emptyset$ are well defined.

  We also complement $\mathcal R$ with $(\leq)$ by adding an edge
  between related constants.
\item
  At this point, we can easily check for satisfiability: A constraint
  is satisfiable (in the given environment) if and only if,
  for any constants $l_1$ and $l_2$ such that
  there is a path from $l_1$ to $l_2$ in $\mathcal G$, then $l_1\lk_\Lat l_2$.
  If this is not the case, we return \textbf{fail}.

\item For each strongly connected component in $\mathcal G$, unify all its vertices and replace it by a representative.
  We write $\unif_k$ for the substitution that replaces a kind variable by
  its representative.
  The representative of a strongly connected component $g$ can be determined as follows:
  \begin{itemize}
  \item If $g$ does not contain any constant, then the representative
    is a fresh kind variable.
  \item If $g$ contains exactly one constant, it is the representative.
  \item Otherwise, the initial constraint $C_0$ is not satisfiable.
  \end{itemize}
  Note that this step will also detect all unsatisfiable constraints.
\item Take the transitive closure of $\mathcal R$.
\item Remove all the vertices corresponding to the kind variables $\Multi\kvar$
  that are existentially quantified in $C^k$.
\item Take the transitive reduction of $\mathcal R$.
\item Remove the extremums of $\mathcal L$ and the edges of $(\leq)$
  from $\mathcal R$.
\item Return $C = \left\{ k \leq k' \mid k \operatorname{\mathcal R}k' \right\}$
  and $\unif =  \unif_\tau \meet \unif_k$.
\end{enumerate}

An example of this algorithm in action is shown in \cref{solving:example}.
Our algorithm is complete, computes principal normal forms,
and already simplifies constraints significantly
(thanks to steps 6, 7 and 8).
It can be extended with further simplification phases.
In particular, our implementation and all the signatures presented in
\cref{motivation} use a variance-based simplification
where all covariant (resp. contravariant) variables are replaced by their
lower (resp. upper) bounds.
All the simplification mechanisms presented
here, including the variance-based one, are complete.
It is also possible to add ``best-effort'' simplification
rules which help reduce the size of inferred signatures even further
\citep{DBLP:conf/aplas/Simonet03}.



\subsection{Principal constraint system}

We now prove that $\CL$ supports all the properties necessary for
principal type inference, as defined by \hmx.
We first prove that constraint solving
does compute normal forms, and that such normal forms are unique.

\begin{lemma}[Principal normal form]
  \label{lemma:normalform}
  Given a constraint $D\in\CL$, a substitution $\phi$ and
  $(C,\unif) = \normalize{D}{\phi}$,
  then $\phi\leq\unif$,
  $C \equivC \unif D$ and
  $\unif C = C$.
\end{lemma}
\begin{proof}
  Let us partition $\phi$ into a part which affects type variables,
  $\phi_\tau$, and a part which affects kind variables, $\phi_k$.

  We write $(C^k,\unif_\tau)$ for the result of
  the modified Herbrand unification on $(D,\phi)$ in step (1).
  Herbrand unification computes the most general
  unifier. Our modified Herbrand unification only output additional
  kind constraints for kind on the arrows and does not change
  the result of the unification. Thus, we have
  $\phi_\tau\leq\unif_\tau$,
  $C^k \equivC \unif_\tau D$ and
  $\unif_\tau C^k = C^k$.

  Let $C^{k+}$ be the result after step (2), we trivially have that
  $\fv{C^{k+}} = \fv{C^k}$ and that $C^{k+} \equivC C^k$.

  Let $C^{A}$ and $\unif_k$ be the results after step (4).
  By definition, we have $\unif_k C^{k+} \equivC C^{A}$ and
  $\unif_k C^{A} = C^{A}$. Since $\phi_k$ has already be applied to $C$ before
  unifying the strongly connected components,
  we have that $\phi_k\leq\unif_k$.

  Let $\unif = \unif_\tau \meet \unif_k$. Since $\unif_\tau$ and $\unif_k$
  have disjoint supports,
  we have $C^{A} = \unif_\tau C^{A} \equivC \unif C^{k+} \equivC \unif D$
  and $\unif C^{A} = C^{A}$.
  Furthermore, $\phi_\tau \meet \phi_k \leq \unif_\tau \meet \unif_k$.

  Steps (5) to (9) all preserve the free variables and the equivalence
  of constraints, which concludes.
\end{proof}

\begin{lemma}[Uniqueness]
  Given $(C_1,\unif_1)$ and $(C_2,\unif_2)$ such that
  $\unif_1 C_1 \equivC \unif_2 C_2$,\\
  then
  $\normalize{C_1}{\unif_1}$ and $\normalize{C_2}{\unif_2}$
  are identical up to $\alpha$-renaming.
\end{lemma}
\begin{proof}
  In \cref{lemma:normalform}, we have showed that all the steps of the
  normalization procedure preserve equivalence.
  Since $\unif_1 C_1 \equivC \unif_2 C_2$, equivalence between
  the two results of the normalization procedures is preserved for all steps.

  We write $P(C_a)$ if for all $C = (k, k)'$
  such that $\entail{C_a}{C}$ and $\nvdash_eC$,
  we have $C \in {\mathcal R}_a$.

  Let us write $C_1'$ and $C_2'$ for the constraints after step (4). $P(C_1')$ and
  $P(C_2')$ hold. Indeed, since $C_1'$ and $C_2'$ are only composed
  of existential quantifications and kind inequalities, the only rules
  that applies are transitivity and lattice inequalities.
  After step (2) and (5), the associated relations are fully saturated for these
  two rules, hence all inequalities that can be deduced from $C_a'$ are already
  present in the relation.

  The property $P$ is preserved by step (6) since we only remove
  inequalities that involve existentially quantified variables. Such
  inequalities could not be picked in $P$.

  Let us write $C''_a$ for $a\in\{1,2\}$ the constraints after step (5).
  Since there are no more existential variables,
  we have $C''_a=(k_i,k'_i)_i=\mathcal R''_a$.
  For any $C=(k,k')$ such that $\entail{}{C}$ and $\entail{C''_a}{C}$,
  then $C \in (\leq) \subset {\mathcal R}''_a$.
  Indeed, the only trivial inequalities in our system are equalities of the form
  $(\kvar, \kvar)$, which were removed in step (4) and the lattice inequalities.

  Let us consider $C = (k,k') \in \mathcal R''_1$.
  Since $C''_1\equivC C''_2$, we have $\entail{C''_2}{C}$.
  If $\nvdash_e{C}$, by $P(C''_2)$ we have that
  $C\in R''_2$.
  If $\entail{}{C}$, then $C\in (\leq) \subset R''_2$.
  We conclude that $R''_1 \subset R''_2$.
  By symmetry, $R''_1 = R''_2$ and $C''_1 = C''_2$.

  This equality is preserved by step (7) and (8)
  since the transitive reduction of a directed acyclic graph is unique,
  which concludes.
\end{proof}

We can now prove all the necessary high level properties.

\begin{lemma}
  For all $C\in\mathcal S$, $\entail{C}{x = x}$ implies
  $\entail{}{x = x}$.
\end{lemma}
\begin{proof}
  By definition of $\operatorname{normalize}$, We have $C = \Multi{\Cleq{k}{k'}}$
  such that the underlying relation has no cycles.
  Thus, we can not deduce neither kind nor type equalities from $C$.
\end{proof}

\begin{property}[Regular constraint system]
  $\CL$ is regular, ie, for $x, x'$ two types or kinds,
  $\entail{}{\Ceq{x}{x'}}$ implies
  $\fv{x} = \fv{x'}$
\end{property}
\begin{proof}
  The only equalities possibles are between variables (via symmetry) or
  between constants.
\end{proof}

Finally, we can conclude with all the properties we need for
\hmx:

\begin{theorem}[Principal constraints]
  $\CL$ has the principal constraint property,\\
  $\operatorname{normalize}$ computes principal normal forms for $\CL$
  and $\CL$ is regular.
\end{theorem}

This is sufficient to show that $HM(\CL)$ is principal. However,
we do not use \hmx directly but an extended version with kind inference,
linear and affine types, and borrow.
We extend the proofs of \hmx to such a system in \cref{appendix:infer}.

\section{Syntax-directed typing}
\label{appendix:sdtyping}

\begin{figure}[bt]
  \centering
  \input{sdtyping/sd-kinds}
  \caption{Syntax-directed kinding rule --
    $\inferSK{C}{\E}{\tau}{k}$}
  \label{rules:sd-kinding}
\end{figure}

\begin{figure}[!tbp]
  \input{sdtyping/sd-splitting}
  \caption{Splitting ---
    environments $\lsplit C\E{\E_l}{\E_r}$;
    inference $\bsplit C\E{\inP{\E_l}}{\inP{\E_r}}$;
    binders $\bsplit Cb{\inP{b_r}}{\inP{b_l}}$}
  \label{fig:sd-splitting}
\end{figure}
\begin{figure}[!tbp]
  \input{sdtyping/sd-borrowing}
  \caption{Borrowing ---
    environments $\lregion[n]{C}{x}{\E}{\E'}$;
    inference $\bregion[\inP n]{C}{\inP x}{\E}{\inP{\E'}}$;
    binders $\bregion[\inP n]{C}{\inP x}{b}{\inP{b'}}$}
  \label{fig:sd-borrowing}
\end{figure}
\begin{figure*}[tbp]
  \input{sdtyping/sd-environments}
  \caption{Rewriting constraints on environments --- $\inP{\Cleq{\E}{k}}\Crewrite  C$}
  \label{fig:contraints-environments-types}
  \input{sdtyping/sd-typing}
  \caption{Syntax-directed typing rules --
    $\inferS{C}{\E}{e}{\tau}$}
  \label{fig:syntax-directed-typing}
\end{figure*}
\begin{figure*}[!btp]
  \input{sdtyping/sd-typing-internal}
  \caption{Syntax-directed typing rules for internal language --
    $\inferS{C}{\E}{e}{\tau}$}
  \label{fig:syntax-directed-typing-internal}
\end{figure*}

\subsection{Kinding}

We write $\inferSK{C}{\E}{\tau}{k}$
if $\tau$ has kind $k$ in environment $\E$ under constraints $C$.
The rules are shown in \cref{rules:sd-kinding}.
Kinds and types follow a small calculus with variables ($\tvar$,\dots),
functions (type constructors $\T{t}$), application ($\tapp{\tcon}{\Multi{\tau}}$)
and primitives such as types for arrows ($\tau\tarr{k}\tau'$) and
borrows ($\borrowty{k}{\tau}$).
Kind checking can thus be done in a fairly straightforward, syntax-directed
fashion by simply following
the syntax of the types. Kind arrows can only appear when looking
up the kind scheme of a type constructor $\T t$. Kind arrows are forbidden
in any other contexts.

\subsection{Environments}
\label{typ:extra:envs}

In \cref{sdtyping}, we only gave a partial description of
the splitting and borrowing relations on environments,
$\lsplit C\E\E\E$ and $\lregion C{x}\E\E$.
The complete definitions are shown on \cref{fig:sd-splitting,fig:sd-borrowing}.
All the definitions are made in term of the inference version, which
returns fresh constraints. The solving version then simply
uses entailment, as shown in rule {\sc ESplit-Check} and
{\sc EBorrow-Check}.
The remaining new rules are dedicated to iterating over the environment.

\cref{fig:contraints-environments-types} defines
the rewriting relation on environment constraints,
$\Cleq\E k \Crewrite C$, which rewrites a constraint of the form
$\Cleq\E k$ into $C$. It proceeds by iterating over the environment
and expanding the constraints for each binding.
Suspended bindings are rejected ({\sc ConstrSusp}).
Borrow bindings directly use the annotated kind ({\sc ConstrBorrow}).
Other bindings use the underlying type scheme ({\sc ConstrBinding}).
Type schemes are constrained by first inferring the kind, and then
emitting the constraint ({\sc ConstrSD} and {\sc ConstrI}).

\subsection{Typing}

The rules for syntax-directed typing are shown in \cref{fig:syntax-directed-typing}
and follow the presentation given in \cref{sdtyping}.
As usual in HM type systems, introduction of type-schemes
is included in the {\sc Let} rule via generalization.
We define $\generalize{C}{\E}{\tau} =
(\Cproj{\Multi{\kvar},\Multi{\tvar}}{C},
\forall \Multi{\kvar},\Multi{\tvar}.\qual{C}{\tau})$
where
$\Multi{\kvar},\Multi{\tvar} = (\fv{\tau}\cup\fv{C})\setminus\fv{\E}$.
The typing rules specific to the internal language are shown in
\cref{fig:syntax-directed-typing-internal}.


\clearpage
\section{Type Inference}
\label{appendix:infer}

\begin{figure*}[tbp]
  \centering
  \input{infer/kinds}
  \caption{Kind inference rules -- $\inferK{(C,\unif)}{{\inP\E}}{{\inP\tau}}{k}$}
  \label{rules:kinding}
\end{figure*}

\begin{figure*}[!tbp]
  \input{infer/typing}
  \caption{Type inference rules --
    $\inferW{\Sigma}{(C,\psi)}{{\inP\E}}{{\inP e}}{\tau}$ }
  \label{rules:typing:full}
\end{figure*}

In this appendix, we provide the complete type inference rules
and show that our type inference algorithm is sound and complete.
The constraints rules are already shown in \cref{inference}.
Kind inference is presented in \cref{rules:kinding:full}
and the detailed treatment of let-bindings in
\cref{infer:let}.
The type inference rules are shown in \cref{rules:typing:full}.
The various theorems and their proofs are direct adaptations
of the equivalent statements for \hmx \citep{sulzmann1997proofs}.

\subsection{Kind Inference}
\label{rules:kinding:full}

We write $\inferK{(C,\unif)}{{\inP\E}}{{\inP\tau}}{k}$ when type $\tau$ has kind $k$
in environment $\E$ under constraints $C$ and unifier $\unif$.
$\inP\E$ and $\inP\tau$ are the input parameters of
our inference procedure.
We present the kind inference algorithm as a set of rules in
\cref{rules:kinding}.
Higher-kinds are not generally supported
and can only appear by looking-up the kind scheme of a type constructor,
for use in the type application rule {\sc KApp}.
Type variables must be of a simple kind in rule {\sc KVar}.
Kind schemes are instantiated in the {\sc KVar} rules by creating
fresh kind variables and the associated substitution.
{\sc KArr} and {\sc KBorrow} simply returns the kind of the primitive
arrow and borrow types.
The $\operatorname{normalize}$ function is used every time several constraints
must be composed in order to simplify the constraint and return a most general
unifier.

\subsection{Generalization and constraints}
\label{infer:let}

The {\sc Let} rule combines several ingredients previously seen:
since let expressions are binders, we use $\Weaken$ on the bound
identifier $x$. Since let-expressions contain
two subexpressions, we use the environment splitting relation,
$\bsplit{C_s}{\Sv}{\Sv_1}{(\Sv_2 \Sdel{x})}$. We remove the $x$ from
the right environment, since it is not available outside of the expression
$e_2$, and should not be part of the returned usage environment.

As per tradition in ML languages, generalization is performed
on let-bindings.
Following \hmx, we write $(C_\schm, \schm) = \generalize{C}{\E}{\tau}$
for the pair of a constraint and a scheme resulting from
generalization. The definition is provided in \cref{rule:infer:let}.
The type scheme $\schm$ is created by quantifying over all the appropriate
free variables and the current constraints.
The generated constraint $C_\schm$ uses a new projection operator,
$\Cproj{x}{D}$ where $x$ can be either a type or a kind variable, which
allow the creation of local variables inside the constraints.
This allows us to encapsulate all the quantified variables in the global constraints.
It also reflects the fact that there
should exist at least one solution for $C$ for the scheme to be valid.
\citet{DBLP:journals/tapos/OderskySW99} give a detailed account
on the projection operators in HM inference.



\newcommand\mcase[1]{\noindent\textbf{Case }#1\\\noindent}

\subsection{Soundness}

\begin{lemma}
  \label{lemma:constrsubst}
  Given constraints $C$ and $D$ and substitution $\unif$, if $\entail{D}{C}$
  then $\entail{\unif D}{\unif C}$.
  \begin{proof}
    By induction over the entailment judgment.
  \end{proof}
\end{lemma}

\begin{lemma}
  \label{lemma:constrimply}
  Given a typing derivation $\inferS{C}{\E}{e}{\tau}$ and
  a constraint $D \in \mathcal S$ in solved form such that $\entail{D}{C}$, then
  $\inferS{D}{\E}{e}{\tau}$
  \begin{proof}
    By induction over the typing derivation
  \end{proof}
\end{lemma}

\begin{lemma}
  \label{lemma:typ:weakening}
  Given a type environment $\E$, $\E' \subset \E$, a term $e$ and a variable $x\in\E$,\\
  if $\inferS{C}{\E'}{e}{\tau}$
  then $\inferS{C\Cand \Weaken_{x}(\E')}{\E';\bvar{x}{\E(x)}}{e}{\tau}$

  \begin{proof}
    Trivial if $x \in \Sv$. Otherwise, by induction over the typing derivation.
  \end{proof}
\end{lemma}


We define the flattening $\Eflat\E$ of an environment $\E$, as the environment
where all the binders are replaced by normal ones. More formally:
\begin{align*}
  \Eflat\E
  =& \left\{ \bvar x \schm \mid \bvar{x}{\schm} \in\E 
    \vee \bbvar {x}{k}{\schm}\in\E
    \vee \svar{x}{\schm}^n\in\E
    \right\}\\
  &\cup \left\{ \bvar{\tvar}{k} \mid \bvar{\tvar}{k}\in\E \right\}
\end{align*}

\begin{lemma}
  \label{lemma:env:flat}
  Given a type environment $\E$ and a term $e$ such
  that $\inferW{\Sv}{(C,\unif)}{\E}{e}{\tau}$
  then $\Eflat\Sv \subset \E$.
  \begin{proof}
    By induction over the typing derivation.
  \end{proof}
\end{lemma}


\begin{theorem}[Soundness of inference]
  Given a type environment $\E$ containing only value bindings,
  $\E|_\tau$ containing only a type binding and a term $e$,\\
  if $\inferW{\Sv}{(C,\unif)}{\E;\E_\tau}{e}{\tau}$
  then $\inferS{C}{\unif(\Sv;\E_\tau)}{e}{\tau}$, $\unif C = C$ and $\unif \tau = \tau$
\begin{proof}
  We proceed by induction over the derivation of $\vdash_w$.
  Most of the cases follow the proofs from \hmx closely.
  For brevity, we showcase three rules: the treatment of binders
  and weakening, where our inference algorithm differ significantly
  from the syntax-directed rule, and the $Pair$ case
  which showcase the treatment of environment splitting.

  \mcase{$\ruleIVar$}

  We have $\Sv = \Sone{x}{\sigma}$.
  Without loss of generality, we can consider $\unif_x = \unif'|_{\fv{\E}} = \unif'|_{\fv{\sigma}}$.
  Since $\Sv\Sdel{x}$ is empty and by definition of normalize, we
  have that
  $\entail C \unif'(C_x) \Cand \Cleq{\unif_x\Sv\Sdel{x}}{\kaff_\infty}$,
  $\unif' \leq \unif$ and $\unif'C = C$.
  By definition, $\unif_x\unif' = \unif'$.
  By rule $Var$, we obtain
  $\inferS{C}{\unif_x(\Sv;\E_\tau)}{x}{\unif_x\unif'\tau}$, which concludes.
  \\

  \mcase{$\ruleIAbs$}

  By induction, we have
  $\inferS{C'}{\Sv_x;\E_\tau}{e}{\tau}$, $\unif C = C$
  and $\unif \tau = \tau$.\\
  By definition of normalize and \cref{lemma:constrsubst}, we have
  $\entail{C}{C'\Cand\Cleq{\unif'\Sv}{\unif'\kvar} \Cand \Weaken_{x}(\unif'\Sv_x)}$ and
  $\unif \leq \unif'$.
  By \cref{lemma:constrsubst}, we have $\entail{C}{\Cleq{\Sv}{\unif\kvar}}$.

  We now consider two cases:
  \begin{enumerate}[leftmargin=*,noitemsep,topsep=5pt]
  \item If $x\in\Sv_x$, then $\Weaken_{x}(\unif\Sv_x) = \Ctrue$
    and by \cref{lemma:env:flat}, $\Sv_x = \Sv;\bvar{x}{\alpha}$.
    We can deduce\\
    $\inferS{C'\Cand \Weaken_{\bvar{x}{\tvar}}(\unif\Sv_x)}{\unif\Sv;\E_\tau;\bvar{x}{\unif(\tvar)}}{e}{\tau}$.
  \item If $x\notin\Sv_x$, then $\Sv = \Sv_x$ and
    $\Weaken_{\bvar{x}{\tvar}}(\unif\Sv_x) = \Cleq{\unif\tvar}{\kaff_\infty}$.
    By \cref{lemma:typ:weakening},
    we have\\
    $\inferS{C'\Cand \Weaken_{\bvar{x}{\tvar}}(\unif\Sv_x)}{\unif\Sv;\E_\tau;\bvar{x}{\unif(\tvar)}}{e}{\tau}$
  \end{enumerate}
  By \cref{lemma:constrimply}, we
  have $\inferS{C}{\unif(\Sv;\E_\tau);\bvar{x}{\unif(\tvar)}}{e}{\tau}$.\\
  By rule $Abs$, we obtain
  $\inferS{C}{\unif(\Sv;\E_\tau)}{\lam{x}{e}}{\unif(\tvar)\tarr{\unif(\kvar)}\tau}$ which concludes.
  \\

  \mcase{$\ruleIPair$}

  By induction, we have
  $\inferS{C_1}{\unif_1(\Sv_1;\E^1_\tau)}{e_1}{\tau_1}$, $\unif_1 C_1 = C_1$,
  and $\unif_1 \tau_1 = \tau_1$
  and
  $\inferS{C_2}{\unif_2(\Sv_2;\E^2_\tau)}{e_2}{\tau_2}$, $\unif_2 C_2 = C_2$
  and $\unif_2 \tau_2 = \tau_2$.
  Wlog, we can rename the type $\E^1_\tau$ and $\E^2_\tau$ to be disjoint
  and define $\E_\tau = \E^1_\tau \cup \E^2_\tau$.
  By normalization, $\entail{C}{D}$, $\unif \leq \unif'$ and $\unif C = C$.
  By \cref{lemma:constrimply} and by substitution, we have
  $\inferS{C}{\unif\Sv_1}{e_1}{\unif\tau_1}$ and
  $\inferS{C}{\unif\Sv_2}{e_2}{\unif\tau_2}$.
  We directly have that
  $\bsplit{\unif C_s}{\unif \Sv}{\Sv_1}{\unif\Sv_2}$
  and by \cref{lemma:constrimply}, $\entail{\unif C}{\unif C_s}$.
  By rule $Pair$, we obtain
  $\inferS{C}
  {\unif(\Sv;\E_\tau)}{\app{e_1}{e_2}}{\unif(\tyPair{\tvar_1}{\tvar_2}))}$,
  which concludes.
  \\





\end{proof}
\end{theorem}

\subsection{Completeness}

We now state our algorithm is complete: for any given
syntax-directed typing derivation, our inference algorithm can find
a derivation that gives a type at least as general.
For this, we need first to provide a few additional definitions.

\begin{definition}[More general unifier]
  Given a set of variable $U$ and $\unif$, $\unif'$ and $\phi$
  substitutions. \\
  Then
  $\unif \leq^{\phi}_{U} \unif'$ iff $(\phi \circ \unif)|_{U} = \unif'|_U$.
\end{definition}

\begin{definition}[Instance relation]
  Given a constraints $C$ and two schemes
  $\schm = \forall \Multi\tvar. \qual{D}{\tau}$ and
  $\schm' = \forall \Multi\tvar'. \qual{D'}{\tau'} $.
  Then $\entail{C}{\schm \preceq \schm'}$
  iff $\entail{C}{\subst{\tvar}{\tau''} D}$
  and $\entail{C\Cand D'}{\Cleq{\subst{\tvar}{\tau''}\tau}{\tau'}}$
\end{definition}

We also extend the instance relation to environments $\E$.

We now describe the interactions between splitting and the
various other operations.

\begin{lemma}
  \label{split:flat}
  Given $\bsplit{C}{\E}{\E_1}{\E_2}$, Then $\Eflat\E = \Eflat\E_1 \cup \Eflat\E_2$.
  \begin{proof}
    By induction over the splitting derivation.
  \end{proof}
\end{lemma}

\begin{lemma}
  \label{split:gen}
  Given $\lsplit{C}{\E_1}{\E_2}{\E_3}$,
  $\bsplit{C'}{\E'_1}{\E'_2}{\E'_3}$
  and $\unif$ such that
  $\E'_i\subset\E''_i$ and $\entail{}{\unif\E''_i \preceq \E_i}$
  for $i\in\{1;2;3\}$.

  Then $\entail{C}{\unif C'}$.
\begin{proof}
  By induction over the derivation of $\bsplit{C'}{\E'_1}{\E'_2}{\E'_3}$.
\end{proof}
\end{lemma}

We can arbitrarily extend the initial typing environment in an inference
derivation, since
it is not used to check linearity.

\begin{lemma}
  \label{infer:extend}
  Given $\inferW{\Sv}{(C,\unif)}{\Eflat\E}{e}{\tau}$ and $\E'$ such
  that $\E \subseteq \E'$, then $\inferW{\Sv}{(C,\unif)}{\Eflat\E'}{e}{\tau}$
  \begin{proof}
    By induction over the type inference derivation.
  \end{proof}
\end{lemma}

Finally, we present the completeness theorem.

\begin{theorem}[Completeness]
  Given $\inferS{C'}{\E'}{e}{\tau'}$ and
  $\entail{}{\unif'\E \preceq \E'}$.
  Then $$\inferW{\Sv}{(C,\unif)}{\Eflat\E}{e}{\tau}$$
  for some environment $\Sv$,
  substitution $\unif$, constraint $C$ and type $\tau$ such
  that
  \begin{align*}
    \unif &\leq^{\phi}_{\fv{\E}} \unif'
    &\entail{C'&}{\phi C}
    &\entail{&}{\phi \schm \preceq \schm'}
    &\Sv&\subset\E
  \end{align*}
  where $\schm' = \generalize{C'}{\E'}{\tau'}$
  and $\schm = \generalize{C}{\E}{\tau}$
\end{theorem}
\begin{proof}
  Most of the difficulty of this proof comes from proper handling of
  instanciation and generalization for type-schemes.
  This part is already proven
  by \citet{sulzmann1997proofs} in the context
  of \hmx. As before, we will only present few cases
  which highlights the handling of bindings and environments.
  For clarity, we will only present the part of the proof that
  only directly relate to the new aspect introduced by \lang.
  \\



  \mcase{$
    \inferrule[Abs]
    {
      \inferS{C'}
      {\E'_x;\bvar{x}{\tau'_2}}{e}{\tau'_1} \\
      \addlin{\entail{C}{\Cleq{\E'_x}{k}}}
    }
    { \inferS{C'}{\E'_x}
      {\lam[k]{x}{e}}{\tau'_2\tarr{k}\tau'_1} }
    $ and $\entail{}{\unif'\E \preceq \E'}$.}

  Let us pick $\tvar$ and $\kvar$ fresh.
  Wlog, we choose $\unif'(\tvar) = \tau_2$ and $\unif'(\kvar) = k$
  so that $\entail{}{\unif'\E_x \preceq \E'_x}$.
  By induction:
  \begin{align*}
    \inferW{\Sv_x}{(C,\unif)}{\Eflat\E_x;\bvar{x}{\tvar}&}{e}{\tau}
    &\unif &\leq^{\phi}_{\fv{\E_x}\cup\{\tvar;\kvar\}} \unif'
  \end{align*}
  \begin{align*}
    \entail{C'&}{\phi C}
    &\entail{&}{\phi \schm \preceq \schm'}
    &\Sv_x&\subset\E_x;\bvar{x}{\tvar}
  \end{align*}
  \begin{align*}
    \schm' &= \generalize{C'}{\E_x';\bvar{x}{\tau'_2}}{\tau'_1}
    &\schm &= \generalize{C}{\E_x;\bvar{x}{\tvar}}{\tau_1}
  \end{align*}

  Let $C_a = C\Cand \Cleq{\Sv}{\kvar} \Cand \Weaken_{\bvar{x}{\tvar}}(\Sv_x)$
  and
  By definition, $\unif_D\Sdel{\alpha;\kvar} \leq^{\phi_D}_{\fv{\E_x}} \unif$
  which means we have
  $\unif_D\Sdel{\alpha;\kvar} \leq^{\phi\circ\phi_D}_{\fv{\E_x}} \unif'$.
  We also have that $\Sv_x\Sdel{x}\subset \E_x$.

  Since $\entail{C}{\Cleq{\E'_x}{k}}$, we have $\entail{C}{\unif'\Cleq{\Sv}{\kvar}}$.
  If $x\in\Sv_x$, then $\Weaken_{\bvar{x}{\tvar}}(\Sv_x) = \Ctrue$.
  Otherwise we can show by induction
  that $\entail{C'}{\unif'\Weaken_{\bvar{x}{\tvar}}(\Sv_x)}$.
  We also have $\unif C = C$, which gives us $\entail{C'}{\unif'(C_a)}$.
  We can deduce
  $\entail{C'}{\unif'(C_a)}$.\\
  This means $(C',\unif')$ is a normal form of $C_a$, so a principal normal form
  exists. Let $(D,\unif_D) = \normalize{C_a}{\unif\Sdel{\alpha;\kvar}}$.
  By the property of principal normal forms,
  we have $\entail{C'}{\rho D}$ and
  $\unif_D \leq^{\rho}_{\fv{\E_x}} \unif'$.

  By application of {\sc Abs$_I$}, we have
  $\inferW{\Sv_x\Sdel{x}}{(C,\unif_D\Sdel{\tvar,\kvar})}{\Eflat\E_x}
  {\lam{x}{e}}{\unif_D(\tvar)\tarr{\unif_D(\kvar)}\tau_1}$.\\
  The rest of the proof proceeds as in the original \hmx proof.
  \\


  \mcase{$
    \inferrule[Pair]
    {
      \inferS{C'}{\E'_1}{e_1}{\tau'_1} \\
      \inferS{C'}{\E'_2}{e_2}{\tau'_1} \\
      \lsplit{C}{\E'}{\E'_1}{\E'_2}
    }
    { \inferS{C'}
      {\E'}{\app{e_1}{e_2}}{\tyPair{\tau'_1}{\tau'_2}} }
    $\\
    and $\entail{}{\unif'\E \preceq \E'}$}

  The only new elements compared to \hmx is
  the environment splitting.
  By induction:
  \begin{align*}
    \inferW{\Sv_1}{(C_1,\unif_1)}{\Eflat\E_1&}{e}{\tau_1}
    &\unif_1 &\leq^{\phi_2}_{\fv{\E}} \unif'_1
    &\entail{C'&}{\phi C_1}
    &\entail{&}{\phi_2 \schm_1 \preceq \schm'_1}
  \end{align*}
  \begin{align*}
    \Sv_1&\subset\E_1
    &\schm'_1 &= \generalize{C'}{\E'_1}{\tau'_1}
    &\schm_1 &= \generalize{C}{\E_1}{\tau_1}
  \end{align*}
  and
  \begin{align*}
    \inferW{\Sv_2}{(C_2,\unif_2)}{\Eflat\E_2&}{e}{\tau_2}
    &\unif_2 &\leq^{\phi_2}_{\fv{\E}} \unif'_2
    &\entail{C'&}{\phi C_2}
  \end{align*}
  \begin{align*}
    \entail{&}{\phi_2 \schm_2 \preceq \schm'_2}
    &\Sv_2&\subset\E_2
  \end{align*}
  \begin{align*}
    \schm'_2 &= \generalize{C'}{\E'_2}{\tau'_2}
    &\schm_2 &= \generalize{C}{\E_2}{\tau_2}
  \end{align*}

  By \cref{split:flat,infer:extend}, we have
  \begin{align*}
    \inferW{\Sv_1}{(C_1,\unif_1)}{\Eflat\E&}{e}{\tau_1}
    &\inferW{\Sv_2}{(C_2,\unif_2)}{\Eflat\E&}{e}{\tau_2}
  \end{align*}

  Let $\bsplit{C_s}{\Sv}{\Sv_1}{\Sv_2}$.
  We know that $\entail{}{\unif'\E \preceq \E'}$,
  $\entail{}{\unif'_i\E_i \preceq \E'_i}$ and $\Sv_i\subset\E_i$.
  By \cref{split:gen},
  we have $\entail{C}{\unif' C_s}$.
  The rest of the proof follows \hmx.

\end{proof}

\begin{corollary}[Principality]
  Let $\inferS{\Ctrue}{\E}{e}{\schm}$ a closed typing judgment.\\
  Then $\inferW{\Sv}{(C,\unif)}{\Eflat\E}{e}{\tau}$
  such that:
  \begin{align*}
    (\Ctrue,\schm_o) &= \generalize{C}{\unif\E}{\tau}
    &\entail{&}{\schm_o \preceq \schm}
  \end{align*}

\end{corollary}



\clearpage
\section{Semantics Definitions}
\label{sec:semant-defin}



\begin{figure*}
  \begin{minipage}[t]{0.52\linewidth}
\begin{lstlisting}[style=rule,linerange=eval\ header-(**)]
(* with splitting*)
type ident = string

type splitting = ident list * ident list

type loc = int
type borrow = Imm | Mut

type address = Address of borrow list * loc

let loc : loc -> address =
  fun ell -> Address ([], ell)

type result = RADDR of address
            | RCONST of int
            | RVOID

type venv = VEmpty | VUpd of venv * ident * result

type alpha = int                (* type variable *)

type 'a var = V of int

type kind = KVAR of kind var
          | KLIN of int option
          | KAFF of int option
          | KUNR of int option

type kappa = kind var           (* kind variable *)
type kappas = kappa list

type constr = (kind * kind) list

type ty = TYVAR of alpha
        | TYARR of kind * ty * ty
        | TYPAIR of ty * ty
        | TYCON of ty list * ident
        | TYBOR of borrow * ty

type schm = SCHM of kappas * (alpha * kind) list * constr * ty

type exp = Const of int
         | Var of ident
         | Varinst of ident * kind list
         | Lam of kind * ident * exp
         | App of exp * exp * splitting
         | Let of ident * exp * exp * splitting
         | LetFun of ident * schm * kind * ident * exp * exp * splitting
         | Pair of kind * exp * exp * splitting
         | Match of ident * ident * exp * exp * splitting
         | Matchborrow of ident * ident * exp * exp * splitting
         | Region of exp * int * ident * ty * borrow
         | Borrow of borrow * ident
         | Create (* of exp *)
         | Destroy of exp
         | Observe of exp
         | Update of exp * exp * splitting
         (* anf cases *)
         | VApp of ident * ident
         | VPair of kind * ident * ident
         | VMatch of ident * ident * ident * exp * splitting
         | VMatchborrow of ident * ident * ident * exp * splitting
         | VCreate of ident
         | VDestroy of ident
         | VObserve of ident
         | VUpdate of ident * ident

type storable
  = STPOLY of venv * kappas * constr * kind * ident * exp
  | STCLOS of venv * kind * ident * exp
  | STPAIR of kind * result * result
  | STRSRC of result
  | STRELEASED

(* hashtable and next free location *)
module Store = Map.Make(Int)
type store = Store of storable Store.t * loc

type perm = address list

let s_add : perm -> address -> perm =
  fun p a -> a :: p
let s_rem : perm -> address -> perm =
  assert false
let s_sub : perm -> perm -> perm =
  assert false

let (<+>) = s_add
let (<->) = s_rem
let (--) = s_sub
let (++) = (@)

let (<|=) = List.mem

let (!$) = loc

(* end of syntax *)

type 'a sem = Error of string | TimeOut | Ok of 'a

let sembind : 'a sem -> ('a -> 'b sem) -> 'b sem =
  fun sa fsb ->
  match sa with
    Error (s) -> Error (s)
  | TimeOut -> TimeOut
  | Ok (a) -> fsb a

let (let*) = sembind

let borrowed : address -> borrow -> address sem =
  fun (Address (bs, ell)) b ->
  match b with
    Imm ->
     Ok (Address (b :: bs, ell))
  | Mut ->
     match bs with
       [] ->
        Ok (Address ([b], ell))
     | (Mut :: bs) ->
        Ok (Address (b :: bs, ell))
     | (Imm :: bs) ->
        Error ("trying to take mutable borrow of immutable borrow")

let rborrowed : borrow -> address -> address sem =
  fun b a -> borrowed a b

let (<.>) = rborrowed

let checkborrowed : address -> borrow -> bool =
  fun (Address (bs, ell)) b ->
  match bs with
    [] ->
     false
  | (b' :: _) ->
     (b = b')

let (<?>) = checkborrowed

let rec vlookup : venv -> ident -> result sem =
  fun gamma x ->
  match gamma with
    VEmpty -> Error ("variable not found")
  | VUpd (gamma, y, r) ->
     if x = y then Ok r else vlookup gamma x

let (.!()) = vlookup

let rec vrestrict : venv -> ident list -> venv =
  fun gamma vs ->
  match gamma with
    VEmpty -> VEmpty
  | VUpd (gamma, x, r) ->
     let gamma' = vrestrict gamma vs in
     if List.mem x vs then VUpd (gamma', x, r) else gamma'

let vsplit : venv -> splitting -> venv * venv =
  fun gamma (vs1, vs2) ->
  vrestrict gamma vs1, vrestrict gamma vs2

let salloc : store -> storable -> (loc * store) =
  fun delta w ->
  match delta with
    Store (htable, nexta) ->
    (nexta, Store (Store.add nexta w htable, nexta+1))

let slookup : store -> loc -> storable sem =
  fun (Store (htable, _)) ell ->
  match Store.find_opt ell htable with
  | Some x -> Ok x
  | None -> Error "illegal store location"

let supdate : store -> loc -> storable -> store sem =
  fun (Store (htable, limit)) ell w ->
  let htable' = Store.add ell w htable in
  Ok (Store (htable', limit))

let (.*()) = slookup
let (.*()<-) = supdate

let getloc : result -> loc sem =
  fun r ->
  match r with
    RADDR (Address ([], rl)) -> Ok rl
  | _ -> Error ("raw location expected")

let getborrowed_loc : result -> (borrow * borrow list * loc) sem =
  fun r ->
  match r with
    RADDR (Address (b :: bs, ell)) -> Ok (b, bs, ell)
  | _ -> Error ("borrowed address expected")

let getaddress : result -> address sem =
  fun r ->
  match r with
    RADDR (rho) -> Ok rho
  | _ -> Error ("address expected")

let getstpoly : storable -> (venv * kappas * constr * kind * ident * exp) sem =
  fun w ->
  match w with
    STPOLY (gamma, kappas, cstr, k, x, e) ->
      Ok (gamma, kappas, cstr, k, x, e)
  | _ ->
     Error ("expected STPOLY")

let getstclos : storable -> (venv * kind * ident * exp) sem =
  fun w ->
  match w with
    STCLOS (gamma, k, x, e) ->
    Ok (gamma, k, x, e)
  | _ ->
     Error ("expected STCLOS")

let getstpair : storable -> (kind * result * result) sem =
  fun w ->
  match w with
    STPAIR (k, r1, r2) ->
    Ok (k, r1, r2)
  | _ ->
     Error ("expected STPAIR")

let getstrsrc : storable -> result sem =
  fun w ->
  match w with
    STRSRC (r) -> Ok r
  | _ -> Error ("expected STRSRC")

let (let*?) : bool -> (unit -> 'b sem) -> 'b sem =
  fun b f -> if b then f () else Error ("test failed")

let (-:>) : 'a -> 'b -> 'a * 'b =
  fun a b -> a, b
let (.+()) : venv -> ident * result -> venv =
  fun v (i, r) -> VUpd (v, i, r)

type subst
let (-->) : 'a list -> 'a var list -> subst = assert false
let (./{}) : 'a -> subst -> 'a = assert false

let (<=>) : constr -> constr -> bool = assert false
let (<<=) : kind -> kind -> kind * kind =
  fun k1 k2 -> k1, k2

let rec borrowedperm : borrow -> address -> ty -> store -> perm sem =
  fun b rho t delta ->
  let* rho' = b<.>rho in
  let* pi = match t with
      TYPAIR (t_1, t_2) ->
        let Address (bs, ell) = rho' in
        let* w = delta.*(ell) in
        let* (k, r_1, r_2) = getstpair w in
        let* rho_1 = getaddress r_1 in
        let* rho_2 = getaddress r_2 in
        let* pi_1 = borrowedperm b rho_1 t_1 delta in
        let* pi_2 = borrowedperm b rho_2 t_2 delta in
        Ok (rho' :: (pi_1 ++ pi_2))
    | _ ->
        Ok ([rho'])
  in
  Ok pi
   
let rec reach : address -> ty -> store -> perm sem =
  fun rho t delta ->
  match t with
    TYPAIR (t_1, t_2) ->
     let Address (bs, ell) = rho in
     let* w = delta.*(ell) in
     let* (k, r_1, r_2) = getstpair w in
     let* rho_1 = getaddress r_1 in
     let* rho_2 = getaddress r_2 in
     let* pi_1 = reach rho_1 t_1 delta in
     let* pi_2 = reach rho_2 t_2 delta in
     Ok (rho :: (pi_1 ++ pi_2))
  | _ ->
     Ok ([rho])

let withaddr : result sem -> (address -> 'a sem) -> 'a sem =
  fun rsem acont ->
  let* r = rsem in
  match r with
  | RADDR a ->
     acont a
  | _ ->
     Error "address expected"

let (let+) = withaddr
    
let rec mapsem : ('a -> 'b sem) -> 'a list -> 'b list sem =
  fun f xs ->
  match xs with
  | [] ->
     Ok []
  | x :: xs ->
     let* y = f x in
     let* ys = mapsem f xs in
     Ok (y :: ys)

let (<..>) =
  fun b pi ->
  mapsem (rborrowed b) pi

  (* eval header *)
let rec eval
  (delta:store) (pi:perm) (gamma:venv) i e
  : (store * perm * result) sem =
  if i=0 then TimeOut else
  let i' = i - 1 in
  match e with
  (**)
  (* rule const *)
  | Const (c) ->
    Ok (delta, pi, RCONST c)
  (**)
  (* rule var *)
  | Var (x) ->
    let* r = gamma.!(x) in
    Ok (delta, pi, r)
  (**)
  (* rule varinst *)
  | Varinst (x, ks) ->
    let* rx = gamma.!(x) in
    let* ell = getloc rx in
    let*? () = !$ ell <|= pi in
    let* w = delta.*(ell) in
    let* (gamma', kappas', cstr', k', x', e') = getstpoly w in
    let pi' =
      if cstr'./{ks-->kappas'} <=> [(k' <<= KUNR None)]./{ks-->kappas'}
      then pi else pi <-> !$ ell
    in
    let w = STCLOS (gamma', k'./{ks-->kappas'}, x', e'./{ks-->kappas'}) in
    let (ell', delta') = salloc delta w in
    Ok (delta', pi' <+> !$ ell', RADDR !$ ell')
  (**)
  (* rule lam *)
  | Lam (k, x, e) ->
    let w = STCLOS (gamma, k, x, e) in
    let (ell', delta') = salloc delta w in
    let pi' = pi <+> !$ ell' in
    Ok (delta', pi', RADDR !$ ell')
  (**)
  (* rule sapp *)
  | App (e_1, e_2, sp) ->
    let (gamma_1, gamma_2) = vsplit gamma sp in
    let* (delta_1, pi_1, r_1) = eval delta pi gamma_1 i' e_1 in
    let* ell_1 = getloc r_1 in
    let*? () = !$ ell_1 <|= pi_1 in
    let* w = delta_1.*(ell_1) in
    let* (gamma', k', x', e') = getstclos w in
    let pi_1' = if k' <= KUNR None then pi_1 else pi_1 <-> !$ ell_1 in
    let* delta_1' = delta_1.*(ell_1) <- (if k' <= KUNR None then w else STRELEASED) in
    let* (delta_2, pi_2, r_2) = eval delta_1' pi_1' gamma_2 i' e_2 in
    let* (delta_3, pi_3, r_3) = eval delta_2 pi_2 gamma'.+(x'-:> r_2) i' e' in
    Ok (delta_3, pi_3, r_3)
  (**)
  (* rule sappanf *)
  | VApp (x_1, x_2) ->
    let* r_1 = gamma.!(x_1) in
    let* ell_1 = getloc r_1 in
    let*? () = !$ ell_1 <|= pi in
    let* w = delta.*(ell_1) in
    let* (gamma', k', x', e') = getstclos w in
    let pi' = (if k' <= KUNR None then pi else pi <-> !$ ell_1) in
    let* delta' = delta.*(ell_1) <- (if k' <= KUNR None then w else STRELEASED) in
    let* r_2 = gamma.!(x_2) in
    let* (delta_3, pi_3, r_3) = eval delta' pi' gamma'.+(x'-:> r_2) i' e' in
    Ok (delta_3, pi_3, r_3)
  (**)
  (* rule slet *)
  | Let (x, e_1, e_2, sp) ->
    let (gamma_1, gamma_2) = vsplit gamma sp in
    let* (delta_1, pi_1, r_1) = eval delta pi gamma_1 i' e_1 in
    let* (delta_2, pi_2, r_2) = eval delta_1 pi_1 gamma_2.+(x -:> r_1) i' e_2 in
    Ok (delta_2, pi_2, r_2)
  (**)
  (* rule sletfun *)
  | LetFun (f, sigma, k, x, e, e', sp) ->
    let (gamma_1, gamma_2) = vsplit gamma sp in
    let SCHM(kappas, _, cstr, ty) = sigma in
    let w = STPOLY (gamma_1, kappas, cstr, k, x, e') in
    let (ell', delta') = salloc delta w in
    let pi' = pi <+> !$ ell' in
    let* (delta_1, pi_1, r_1) = eval delta' pi' gamma_2.+(f -:> RADDR !$ ell') i' e' in
    Ok (delta_1, pi_1, r_1)
  (**)
  (* rule spair *)
  | Pair (k, e_1, e_2, sp) ->
    let (gamma_1, gamma_2) = vsplit gamma sp in
    let* (delta_1, pi_1, r_1) = eval delta pi gamma_1 i' e_1 in
    let* (delta_2, pi_2, r_2) = eval delta_1 pi_1 gamma_2 i' e_2 in
    let w = STPAIR (k, r_1, r_2) in
    let (ell', delta') = salloc delta_2 w in
    Ok (delta', pi_2 <+> !$ ell', RADDR !$ ell')
  (**)
  (* rule spairanf *)
  | VPair (k, x_1, x_2) ->
    let* r_1 = gamma.!(x_1) in
    let* r_2 = gamma.!(x_2) in
    let w = STPAIR (k, r_1, r_2) in
    let (ell', delta') = salloc delta w in
    let pi' = pi <+> !$ ell' in
    Ok (delta', pi', RADDR !$ ell')
  (**)
  (* rule smatch *)
  | Match (x, x', e_1, e_2, sp) ->
    let (gamma_1, gamma_2) = vsplit gamma sp in
    let* (delta_1, pi_1, r_1) = eval delta pi gamma_1 i' e_1 in
    let* ell = getloc r_1 in
    let* w = delta_1.*(ell) in
    let* (k', r_1', r_2') = getstpair w in
    let pi_1' = (if k' <= KUNR None then pi_1 else pi_1 <-> !$ ell) in
    let gamma_2' = gamma_2.+(x -:> r_1).+(x' -:> r_1') in
    let* (delta_2, pi_2, r_2) = eval delta_1 pi_1' gamma_2' i' e_2 in
    Ok (delta_2, pi_2, r_2)
  (**)
  (* rule smatchanf *)
  | VMatch (x, x', z, e_2, sp) ->
    let (gamma_1, gamma_2) = vsplit gamma sp in
    let* r = gamma_1.!(z) in
    let* ell = getloc r in
    let* w = delta.*(ell) in
    let* (k, r_1, r_1') = getstpair w in
    let pi' = if k <= KUNR None then pi else pi <-> !$ ell in
    let* delta' = delta.*(ell) <- (if k <= KUNR None then w else STRELEASED) in
    let gamma_2' = gamma_2.+(x -:> r_1).+(x' -:> r_1') in
    let* (delta_2, pi_2, r_2) = eval delta' pi' gamma_2' i' e_2 in
    Ok (delta_2, pi_2, r_2)
  (**)
  (* rule matchborrow *)
  | Matchborrow (x, x', e_1, e_2, sp) ->
    let (gamma_1, gamma_2) = vsplit gamma sp in
    let* (delta_1, pi_1, r_1) = eval delta pi gamma_1 i' e_1 in
    let* rho = getaddress r_1 in
    let* (b, _, ell) = getborrowed_loc r_1 in
    let* w = delta_1.*(ell) in
    let* (k', r_1', r_2') = getstpair w in
    let* rho_1 = getaddress r_1' in
    let* rho_2 = getaddress r_2' in
    let* rho_1' = b<.>rho_1 in
    let* rho_2' = b<.>rho_2 in
    let pi_1'' = (if k' <= KUNR None then pi_1 else pi_1 <-> rho) in
    let gamma_2'' = gamma_2.+(x -:> RADDR rho_1').+(x' -:> RADDR rho_2') in
    let* (delta_2, pi_2, r_2) = eval delta_1 pi_1'' gamma_2'' i' e_2 in
    Ok (delta_2, pi_2, r_2)
  (**)
  (* let pi_2' = (((pi_2 <-> rho_1') <-> rho_2') <+> rho_1) <+> rho_2 in *)
  (* rule matchborrowanf *)
  | VMatchborrow (x, x', z, e_2, sp) ->
    let (gamma_1, gamma_2) = vsplit gamma sp in
    let* r_1 = gamma_1.!(z) in
    let* (b, _, ell) = getborrowed_loc r_1 in
    let* w = delta.*(ell) in
    let* (k', r_1', r_2') = getstpair w in
    let* rho = getaddress r_1 in
    let pi'' = (if k' <= KUNR None then pi else pi <-> rho) in
    let delta'' = delta in
    let* rho_1 = getaddress r_1' in
    let* rho_2 = getaddress r_2' in
    let* rho_1' = b<.>rho_1 in
    let* rho_2' = b<.>rho_2 in
    let r_1'' = RADDR rho_1' in
    let r_2'' = RADDR rho_2' in
    let gamma_2'' = gamma_2.+(x -:> r_1'').+(x' -:> r_2'') in
    let* (delta_2, pi_2, r_2) = eval delta'' pi'' gamma_2'' i' e_2 in
    Ok (delta_2, pi_2, r_2)
  (**)
  (* rule sregion *)
  | Region (e, n, x, ty_x, b) ->
     let+ rho = gamma.!(x) in
     let* rho' = b<.>rho in
     let* pi' = reach rho ty_x delta in
     let* pi'' = b<..>pi' in
     let gamma' = gamma.+(x -:>RADDR rho') in
     let pi = (pi ++ pi'') -- pi' in
     let* (delta_1, pi_1, r_1) = eval delta pi gamma' i' e in
     let pi_1 = (pi_1 -- pi'') ++ pi' in
     Ok (delta_1, pi_1, r_1)
  (**)
  (*(* previous *)
  | Region (e_1, n, x, b) ->
    let* rx = gamma.!(x) in
    let* rho = getaddress rx in
    let*? () = rho <|= pi in
    let* rho' = b<.>rho in
    let gamma' = gamma.+(x -:>RADDR rho') in
    let pi' = (pi <-> rho) <+> rho' in
    let* (delta_1, pi_1, r_1) = eval delta pi' gamma' i' e_1 in
    let pi_1' = (pi <-> rho') <+> rho in
    Ok (delta_1, pi_1', r_1)
   *)
  (* rule sborrow *)
  | Borrow (b, x) ->
     let+ rho = gamma.!(x) in
     let*? () = rho <?> b && rho <|= pi in
     Ok (delta, pi, RADDR rho)
  (**)
  (*(* previous *)
  | Borrow (b, x) ->
    let* rx = gamma.!(x) in
    let* rho = getaddress rx in
    let* rho' = b<.>rho in
    let*? () = rho' <|= pi in
    Ok (delta, pi, RADDR rho')
   *)
  (* rule screate *)
  | Create ->
    let w = STRSRC (RCONST 0) in
    let (ell_1, delta_1) = salloc delta w in
    let pi_1 = pi <+> !$ ell_1 in
    Ok (delta_1, pi_1, RADDR !$ ell_1)
  (**)
  (* rule screateanf *)
  | VCreate (x) ->
    let* r = gamma.!(x) in
    let w = STRSRC (r) in
    let (ell_1, delta_1) = salloc delta w in
    let pi_1 = pi <+> !$ ell_1 in
    Ok (delta_1, pi_1, RADDR !$ ell_1)
  (**)
  (* rule sdestroy *)
  | Destroy (e_1) ->
    let* (delta_1, pi_1, r_1) = eval delta pi gamma i' e_1 in
    let* rho = getaddress r_1 in
    let* ell = getloc r_1 in
    let* w = delta_1.*(ell) in
    let* r = getstrsrc w in
    let*? () = rho <|= pi_1 in
    let* delta_1' = delta_1.*(ell) <- STRELEASED in
    let pi_1' = pi_1 <-> !$ ell in
    Ok (delta_1', pi_1', RVOID)
  (**)
  (* rule sdestroyanf *)
  | VDestroy (x) ->
    let* r = gamma.!(x) in
    let* rho = getaddress r in
    let* ell = getloc r in
    let* w = delta.*(ell) in
    let* r = getstrsrc w in
    let*? () = rho <|= pi in
    let* delta_1 = delta.*(ell) <- STRELEASED in
    let pi_1 = pi <-> !$ ell in
    Ok (delta_1, pi_1, RVOID)
  (**)
  (* rule sobserve *)
  | Observe (e_1) ->
    let* (delta_1, pi_1, r_1) = eval delta pi gamma i' e_1 in
    let* rho = getaddress r_1 in
    let*? () = rho <|= pi_1 in
    let* (b, _, ell) = getborrowed_loc r_1 in
    let*? () = (b = Imm) in
    let* w = delta_1.*(ell) in
    let* r = getstrsrc w in
    Ok (delta_1, pi_1, r)
  (**)
  (* rule sobserveanf *)
  | VObserve (x) ->
    let* r = gamma.!(x) in
    let* rho = getaddress r in
    let*? () = rho <|= pi in
    let* (b, _, ell) = getborrowed_loc r in
    let*? () = (b = Imm) in
    let* w = delta.*(ell) in
    let* r' = getstrsrc w in
    Ok (delta, pi, r')
  (**)
  (* rule supdate *)
  | Update (e_1, e_2, sp) ->
    let (gamma_1, gamma_2) = vsplit gamma sp in
    let* (delta_1, pi_1, r_1) = eval delta pi gamma_1 i' e_1 in
    let* rho = getaddress r_1 in
    let* (b, _, ell) = getborrowed_loc r_1 in
    let*? () = (b = Mut) in
    let* (delta_2, pi_2, r_2) = eval delta_1 pi_1 gamma_2 i' e_2 in
    let* w = delta_2.*(ell) in
    let* r = getstrsrc w in
    let*? () = rho <|= pi_2 in
    let* delta_2' = delta_2.*(ell) <- STRSRC (r_2) in
    let pi_2' = pi_2 <-> rho in
    Ok (delta_2', pi_2', RVOID)
  (**)
  (* rule supdateanf *)
  | VUpdate (x_1, x_2) ->
    let* r_1 = gamma.!(x_1) in
    let* rho = getaddress r_1 in
    let* (b, _, ell) = getborrowed_loc r_1 in
    let*? () = (b = Mut) in
    let* r_2 = gamma.!(x_2) in
    let* w = delta.*(ell) in
    let* r = getstrsrc w in
    let*? () = rho <|= pi in
    let* delta' = delta.*(ell) <- STRSRC (r_2) in
    let pi' = pi <-> rho in
    Ok (delta', pi', RVOID)
(**)
\end{lstlisting}
  \lstsemrule{const}
  \medskip
  \lstsemrule{var}
  \medskip
  \lstsemrule{varinst}
  \medskip
  \lstsemrule{lam}
  \medskip
  \lstsemrule{sapp}
  \medskip
  \lstsemrule{slet}
  \end{minipage}
  \begin{minipage}[t]{0.47\linewidth}
  \lstsemrule{sletfun}
  \medskip
  \lstsemrule{spair}
  \medskip
  \lstsemrule{smatch}
  \medskip
  \lstsemrule{matchborrow}
  \end{minipage}
  \caption{Big-step interpretation}
  \label{fig:full-big-step-interpretation}
\end{figure*}

\begin{figure*}[tp]
  \begin{minipage}[t]{0.49\linewidth}
  \lstsemrule{sregion}
  \medskip
  \lstsemrule{sborrow}
  \medskip
  \lstsemrule{sdestroy}
  \end{minipage}
  \begin{minipage}[t]{0.49\linewidth}
  \lstsemrule{screate}
  \medskip
  \lstsemrule{sobserve}
  \medskip
  \lstsemrule{supdate}
  \end{minipage}
  \caption{Big-step interpretation (resources)}
  \label{fig:full-big-step-interpretation-resources}
\end{figure*}


\cref{fig:full-big-step-interpretation} presents  the full big-step
interpretation. \cref{fig:full-big-step-interpretation-resources}
contains the cases for resources.


\section{Proofs for Metatheory}
\label{sec:metatheory:proofs}

\begin{itemize}
\item
  For simplicity, we only consider terms in A-normal forms following the grammar:
  \begin{align*}
    e ::=~& \ldots \mid \aiapp{x}{x'} \mid \aipair{k}{x}{x'} \mid \aimatchin\etransfm{x,y}{z}{e}
  \end{align*}
  Typing and semantics rules are unchanged.
\item Borrow qualifiers  $\BQ ::= \IBORROW_n \mid \MBORROW_n$ where
  $n\ge0$ is a region level. A vector of borrow qualifiers $\Multi\BQ$
  is wellformed if all $\IBORROW$s come before all $\MBORROW$s in the vector. 
\item Borrow compatibility
  $\Multi\BQ \Bcompatible \BQ$,
  \begin{mathpar}
  \inferrule{}{
    \BQ_n\Multi\BQ \Bcompatible \BQ_n
  }
  %
  \end{mathpar}
\item Store typing $ \vdash \Store : \SE$,
  \begin{mathpar}
    \inferrule{
      (\forall \Loc \in \Dom\Store)~~
      \SE \vdash \Store (\Loc) : \SE (\Loc)
    }{ \vdash \Store : \SE }
  \end{mathpar}
\item Relating storables to type schemes $\SE \vdash w : \schm$

  We write $\Disjoint\E$ for $\Active\VEnv$ and $\MutableBorrows\VEnv$
  and $\ImmutableBorrows\VEnv$ and $\Suspended\VEnv$ are all disjoint.
  \begin{mathpar}
  \inferrule{
    (\exists \E)~ \SE \vdash \VEnv : \E
    \\
    \Disjoint\E
    \\
    \inferS{C}{\E; \bvar x{\tau_2}}{e}{\tau_1}
    \\
    \Multi\tvar = \fv{\tau_1,\tau_2} \setminus \fv{\E}
  }{
    \SE \vdash (\VEnv, \ilam {\Multi\kvar}{\Multi\tvar}Ckx{e})
    : \forall\Multi\kvar\forall\Multi{\bvar{\tvar}{k}}.(\qual{C}{\tau_2\tarr{k}\tau_1})
  }
  \end{mathpar}
\item Relating storables to types $ \SE \vdash w : \tau$
  \begin{mathpar}
    \ruleStorableClosure
    
    \ruleStorablePair

    \ruleStorableResource

    \ruleStorableFreed
  \end{mathpar}
\item Relating results to type schemes $\SE \vdash r : \schm$
  \begin{mathpar}
    \ruleResultConstant

    \ruleResultLocation

    \ruleResultBorrow
  \end{mathpar}
\item
We write $\Affine\SE\Loc$ to express that $\Loc$ points to a resource
that requires at least affine treatment. Borrow types do not appear in
store types as the store only knows about the actual resources.

Define  $\Affine\SE\Loc$ if one of the following cases holds:
\begin{itemize}
\item $\SE (\Loc) =
  \forall\Multi\kvar\forall\Multi{\bvar{\tvar}{k}}.(\qual{C}{\tau_2\tarr{k}\tau_1})$
  and $C \wedge (k \lk \kun_\infty)$ is contradictory;
\item $\SE (\Loc) = \tau_2\tarr{k}{\tau_1}$ and $\Cleq{\kaff}{k}$;
\item $\SE (\Loc) = \tyPair[k]{\tau_1}{\tau_2}$ and $\Cleq \kaff
  k$;
\item $\SE (\Loc) = \tapp{\tcon}{\Multi\tau}$.
\end{itemize}
\item
We write $\Linear\SE\Loc$ to express that $\Loc$ points to a linear
resource.

Define  $\Linear\SE\Loc$ if one of the following cases holds:
\begin{itemize}
\item $\SE (\Loc) =
  \forall\Multi\kvar\forall\Multi{\bvar{\tvar}{k}}.(\qual{C}{\tau_2\tarr{k}\tau_1})$
  and $C \wedge (k \lk \kaff_\infty)$ is contradictory;
\item $\SE (\Loc) = \tau_2\tarr{k}{\tau_1}$ and $\Cleq{\klin}{k}$;
\item $\SE (\Loc) = \tyPair[k]{\tau_1}{\tau_2}$ and $\Cleq \klin
  k$;
\item $\SE (\Loc) = \tapp{\tcon}{\Multi\tau}$.
\end{itemize}
\item
It remains to characterize unrestricted resources.
Define $\Unrestricted\SE\Loc$ if neither $\Affine\SE\Loc$ nor
$\Linear\SE\Loc$ holds.
\item
  Relating environments to contexts\\
  $\SE \vdash \Active\VEnv, \MutableBorrows\VEnv,
\ImmutableBorrows\VEnv, \Suspended[\kaff]\VEnv, \Suspended[\kun]\VEnv : \E$.

Here we consider an
environment $\VEnv = (\Active\VEnv, \MutableBorrows\VEnv,
\ImmutableBorrows\VEnv, \Suspended\VEnv)$ as a tuple
consisting of the active entries in $\Active\VEnv$ and the
entries for exclusive borrows in $\MutableBorrows\VEnv$ and for
shared borrows in $\ImmutableBorrows\VEnv$, and suspended entries
in $\Suspended\VEnv = \Suspended[\kaff]\VEnv, \Suspended[\kun]\VEnv$
for affine and unrestricted entries. The
suspended entries cannot be used directly, but they can be activated
by appropriate borrowing on entry to a region.
\end{itemize}

\begin{mathpar}
  \inferrule{}{\SE \vdash \Sempty, \Sempty, \Sempty, \Sempty : \Eempty}
  \\
  \inferrule{
    \SE \vdash \Active\VEnv, \MutableBorrows\VEnv,
    \ImmutableBorrows\VEnv, \Suspended\VEnv : \E
    \\ \SE \vdash r : \schm
    \\ \Linear\SE r}
  {\SE \vdash \Active\VEnv[ x\mapsto r], \MutableBorrows\VEnv ,
    \ImmutableBorrows\VEnv, \Suspended\VEnv : \E;\bvar x\schm }

  \inferrule{
    \SE \vdash \Active\VEnv, \MutableBorrows\VEnv,
    \ImmutableBorrows\VEnv, \Suspended\VEnv : \E
    \\ \SE \vdash r : \schm
    \\ \Affine\SE r}
  {\SE \vdash \Active\VEnv, \MutableBorrows\VEnv[ x\mapsto r] ,
    \ImmutableBorrows\VEnv, \Suspended\VEnv : \E;\bvar x\schm }

  \inferrule{
    \SE \vdash \Active\VEnv, \MutableBorrows\VEnv,
    \ImmutableBorrows\VEnv, \Suspended\VEnv : \E
    \\ \SE \vdash r : \schm
    \\ \Unrestricted\SE r}
  {\SE \vdash \Active\VEnv, \MutableBorrows\VEnv,
    \ImmutableBorrows\VEnv[ x\mapsto r], \Suspended\VEnv : \E;\bvar
    x\schm }

  \inferrule{\SE \vdash \Active\VEnv, \MutableBorrows\VEnv,
    \ImmutableBorrows\VEnv, \Suspended[\kaff]\VEnv, \Suspended[\kun]\VEnv : \E \\
    \SE \vdash r : \schm}
  { \SE \vdash \Active\VEnv, \MutableBorrows\VEnv,
    \ImmutableBorrows\VEnv, \Suspended[\kaff]\VEnv, \Suspended[\kun]\VEnv[ x\mapsto r] :
    \E;\svar[\IBORROW] x\schm^n }

  \inferrule{\SE \vdash \Active\VEnv, \MutableBorrows\VEnv,
    \ImmutableBorrows\VEnv, \Suspended[\kaff]\VEnv, \Suspended[\kun]\VEnv : \E \\
    \SE \vdash r : \schm}
  { \SE \vdash \Active\VEnv, \MutableBorrows\VEnv,
    \ImmutableBorrows\VEnv, \Suspended[\kaff]\VEnv[ x\mapsto r], \Suspended[\kun]\VEnv :
    \E;\svar[\MBORROW] x\schm^n }

  \inferrule{\SE \vdash \Active\VEnv, \MutableBorrows\VEnv,
    \ImmutableBorrows\VEnv, \Suspended\VEnv : \E \\ \SE \vdash
    \IBORROW\Addr : \schm}
  {\SE \vdash \Active\VEnv, \MutableBorrows\VEnv,
    \ImmutableBorrows\VEnv[ x\mapsto \IBORROW\Addr], \Suspended\VEnv :
    \E;\bbvar[\IBORROW] x k{ \schm} }

  \inferrule{\SE \vdash \Active\VEnv, \MutableBorrows\VEnv,
    \ImmutableBorrows\VEnv, \Suspended\VEnv : \E \\ \SE \vdash
    \MBORROW\Addr : \schm}
  {\SE \vdash \Active\VEnv, \MutableBorrows\VEnv[ x\mapsto \MBORROW\Addr],
    \ImmutableBorrows\VEnv, \Suspended\VEnv
    : \E;\bbvar[\MBORROW] x k{ \schm} }
\end{mathpar}
\paragraph{Extending environments and stores}
\begin{mathpar}
  \inferrule{}{\SE \le \SE}

  \inferrule{\SE \le \SE' \\ \Loc \notin \Dom\Store}{\SE \le \SE' (\Loc : \schm)}
\\
  \inferrule{}{\Store\le\Store}

  \inferrule{\Store \le \Store' \\ \Loc \notin \Dom\Store
  }{\Store \le \Store'[ \Loc \mapsto w] }
\end{mathpar}

\begin{lemma}[Store Weakening]\label{lemma:store-weakening}
  $\SE \vdash \VEnv : \E$ and $\SE \le \SE'$ implies $\SE' \vdash
  \VEnv : \E$.
\end{lemma}

\begin{lemma}[Store Extension]\label{lemma:store-extension-transitive}
  \begin{itemize}
  \item If $\SE_1 \le \SE_2$ and $\SE_2 \le \SE_3$, then $\SE_1 \le
    \SE_3$.
  \item If $\Store_1 \le \Store_2$ and $\Store_2 \le \Store_3$, then
    $\Store_1 \le \Store_3$.
  \end{itemize}
\end{lemma}

We write $\Rawloc\cdot$ for the function that extracts a multiset of
\emph{raw locations} from a result or from the range of the variable
environment.

\begin{align*}
  \Rawloc{\Multi\BQ\Loc} &= \{\Loc\} \\
  \Rawloc{c} &= \{ \} \\
  \Rawloc\Eempty &= \{\} \\
  \Rawloc{\VEnv( x \mapsto r)} &= \Rawloc\VEnv \cup \Rawloc r
\end{align*}

We write $\Reach\Store\VEnv$ for the multiset of all \emph{addresses}
reachable from $\Rawloc\VEnv$
assuming that $\Rawloc\VEnv \subseteq \Dom\Store$\footnote{In
  mixed comparisons between a multiset and a set, we tacitly convert
  a multiset $M$ to its supporting set $\{ x \mid \MultiNumber x M \ne 0\}$.}.
The function $\Reach\Store\cdot$ is defined in
two steps. First a helper function
for results, storables, and environments.

\begin{align*}
  \RS\Store\Eempty &= \Eempty \\
  \RS\Store{\VEnv (x \mapsto r)} &= \RS\Store\VEnv \cup
                                      \RS\Store r \\
  \RS\Store\Addr &= \{ \Addr \}  \\
  \RS\Store c &= \{ \} \\
  \RS\Store{\StPClosure \VEnv {\Multi\kvar} C k x e} &=
                                       \RS\Store\VEnv
  \\
  \RS\Store{\StClosure \VEnv k x e} &=
                                                   \RS\Store\VEnv
  \\
  \RS\Store{\StPair k {r_1} {r_2}} &=
                                                   \RS\Store{r_1}
                                                   \cup \RS\Store{r_2}
  \\
  \RS\Store{\StRes r} &=
                                   \RS\Store r
  \\
  \RS\Store{\StFreed} &= \{ \}
\end{align*}

This multiset is closed transitively by store lookup. We define
$\Reach\Store\VEnv$ as the smallest multiset $\REACH$ that fulfills
the following inequations. We assume a nonstandard
model of multisets such that an element $\Loc$ may occur infinitely often as in
$\MultiNumber\Loc\REACH = \infty$.
\begin{align*}
  \REACH &\supseteq \RS\Store\VEnv \\
  \REACH &\supseteq \RS\Store w & \text{if }
                                     \Multi\BQ\Loc
                                     \in \REACH \wedge w = \Store (\Loc)
\end{align*}

\begin{definition}[Wellformed permission]
  A permission $\Perm$ is \emph{wellformed} if it contains at most one
  address for each raw location.
\end{definition}
\begin{definition}[Permission closure]
  The closure of a permission $\Sclos\Perm$ is the set of addresses
  reachable from $\Perm$ by stripping an arbitrary number of borrows
  from it. It is the homomorphic extension of the closure
  $\Sclos\Addr$ for a single address.
  \begin{align*}
    \Sclos\Loc & = \{ \Loc \} & \Sclos{(\BQ\Addr)} & = \{ \BQ\Addr \}  \cup \Sclos\Addr
  \end{align*}
\end{definition}

\begin{lemma}[Containment]\label{lemma:containment}
  Suppose that $\vdash \Store : \SE$,
  $\SE \vdash r : \tau$, $\entail C
  {\Cleq{\tau}{k} \Cand \Cleq{k}{\klin_{m-1}}}$.
  Then $\Reach\Store r$ cannot contain addresses $\Addr$ such that
  $\Addr = \BORROW_n\Addr'$ with $n\ge m$.
\end{lemma}
\begin{proof}
  By inversion of result typing there are three cases.

  \textbf{Case }$\ruleResultConstant$. Immediate: reachable set
  is empty.

  \textbf{Case }$\ruleResultBorrow$. The typing constraint enforces
  that $n < m$.

  \textbf{Case }$\ruleResultLocation$. We need to continue by
  dereferencing $\Loc$ and inverting store typing.

  \textbf{Case }$\ruleStorableFreed$. Trivial.

  \textbf{Case }$\ruleStorableResource$. We assume the implementation
  type of a result to be unrestricted.

  \textbf{Case }$\ruleStorablePair$.

  The typing constraint yields that $k \le \klin_{m-1}$.
  By induction and transitivity of $\le$, we find that
  $\Reach\Store{r_1}$ and $\Reach\Store{r_2}$ cannot contain offending addresses.

  \textbf{Case }$\ruleStorableClosure$.

  The typing constraint yields that $k\le\klin_{m-1}$.
  By transitivity of $\le$ and $\SE \vdash \VEnv:\E$, we find that the
  types of all addresses in 
  $\VEnv$ have types bounded by $\klin_{m-1}$ and, by induction, they
  cannot contain offending addresses.
\end{proof}
\clearpage{}
\lstMakeShortInline[keepspaces,style=rule]@

\SoundnessThm

Some explanations are in order for the resource-related assumptions
and statements.

Incoming resources are always active (i.e., not freed).
Linear and affine resources as well as suspended affine borrows have
exactly one pointer in the environment.

The Frame condition states that only store locations reachable from
the current environment can change and that all permissions outside
the reachable locations remain the same.

Unrestricted values, resources, and borrows do not change their
underlying resource and do not spend their permission.

Affine borrows and resources may or may not spend their
permission. Borrows are not freed, but resources may be freed. The
permissions for suspended entries remain intact.


A linear resource is always freed.

Outgoing permissions are either inherited from the caller or they
refer to newly created values.

\newpage
\begin{proof}
  By induction on the evaluation of
  @eval \Store \Perm \VEnv i e@.

  The base case is trivial as
  @eval \Store \Perm \VEnv 0 e = \TimeOut@.

  For $i>0$ consider the different cases for expressions. For lack of
  spacetime, we only give details on some important cases.

  \textbf{Case $e$ of}
  \lstsemrule{slet}
  We need to invert rule \TirName{Let} for monomorphic let:
  \begin{gather*}
    \ruleSDILet
  \end{gather*}
  As $\Sp$ is the evidence for the splitting judgment and @vsplit@
  distributes values according to $\Sp$, we obtain
  \begin{gather}
    \label{eq:16}
    \SE \vdash \VEnv_1 : \E_1
    \\\label{eq:17}
    \SE \vdash \VEnv_2 : \E_2
  \end{gather}
  Moreover (using $\uplus$ for disjoint union),
  \begin{itemize}
  \item $\Active\VEnv = \Active{\VEnv_1} \uplus\Active{\VEnv_2}$,
  \item $\MutableBorrows\VEnv = \MutableBorrows{\VEnv_1} \uplus
    \MutableBorrows{\VEnv_2}$,
  \item $\ImmutableBorrows\VEnv =
    \ImmutableBorrows{\VEnv_1} = \ImmutableBorrows{\VEnv_2}$,
  \item $\Suspended\VEnv = \Suspended{\VEnv_1} \uplus
    \Suspended{\VEnv_2}$ (this splitting does not distinguish potentially
    unrestricted or affine bindings)
  \end{itemize}
  We establish the assumptions for the call
  @eval \Store \Perm \VEnv_1  i' e_1@.
  \begin{enumerate}[({A1-}1)]
  \item From inversion: $\inferS{C \Cand D}{\E_1}{e_1}{\tau_1} $
  \item From~\eqref{eq:16}: $\SE \vdash \VEnv_1 : \E_1$
  \item From assumption
  \item From assumption
  \item From assumption because $\VEnv_1$ is projected from $\VEnv$.
  \item From assumption because $\VEnv_1$ is projected from $\VEnv$.
  \item From assumption because $\VEnv_1$ is projected from $\VEnv$.
  \end{enumerate}
  Hence, we can apply the induction hypothesis and obtain
  \begin{enumerate}[({R1-}1)]
  \item \resultOk{}{_1}
  \item\label{item:8} \resultEnv{}{_1}
  \item\label{item:9} $\SE_1 \vdash r_1 : \tau_1$
  \item\label{item:14} \resultPermDom{}{_1}
  \item\label{item:13} \resultReachPerm{}{_1}
  \item\label{item:16} \resultFrame{}{_1}{_1}
  \item\label{item:24} \resultImmutables{}{}{_1}
  \item\label{item:26} \resultMutables{}{}{_1}
  \item\label{item:30} \resultResources{}{}{_1}
  \item\label{item:10} \resultThinAir{}{_1}
  \end{enumerate}
  To establish the assumptions for the call \\
  @eval delta_1 pi_1 gamma_2(x-:>r_1) i' e_2@,
  we write $\VEnv_2' = \VEnv (x \mapsto r_1)$.
  \begin{enumerate}[({A2-}1)]
  \item From inversion: $\inferS{C}{\E;\bvar{x}{\tau_1}}{e_2}{\tau_2}$
  \item From~\eqref{eq:17} we have $\SE \vdash \VEnv_2 : \E_2$.
    By store weakening (Lemma~\ref{lemma:store-weakening}) and
    using~\ref{item:8}, we have
    $\SE_1 \vdash \VEnv_2 : \E_2$.
    With~\ref{item:9}, we obtain
    $\SE_1 \vdash \VEnv_2 (x \mapsto r_1 ) : \E_2;\bvar x {\tau_1}$.
  \item Immediate from~\ref{item:8}.
  \item Immediate from~\ref{item:14}.
  \item Show \assumeReachable{_1}{'_2}\\
    From~\ref{item:12}, we have
    \assumeReachable{_1}{_2}
    The extra
    binding $(x \mapsto r_1)$ goes into one of the compartments
    according to its type. We conclude by~\ref{item:13}.
  \item Disjointness holds by assumption for $\VEnv_2$ and it remains
    to discuss $r_1$. But $r_1$ is either a fresh resource, a
    linear/affine resource from $\VEnv_1$ (which is disjoint), or
    unrestricted. In each case, there is no overlap with another
    compartment of the environment.
  \item We need to show
    \assumeIncoming{_1}{_2'}
    The first item holds by assumption, splitting, and (for $r_1$)
    by~\ref{item:14} and~\ref{item:13}.

    The second and third items hold by assumption~\ref{item:15},
    splitting, and framing~\ref{item:16}.
  \end{enumerate}
  Hence, we can apply the induction hypothesis and obtain
  \begin{enumerate}[({R2}-1)]
  \item\label{item:17} \resultOk{_1}{_2}
  \item\label{item:18} \resultEnv{_1}{_2}
  \item\label{item:19} $\SE_2 \vdash r_2 : \tau_2$
  \item\label{item:20} \resultPermDom{_1}{_2}
  \item\label{item:21} \resultReachPerm{_1}{_2}
  \item\label{item:22} \resultFrame{_1}{_2'}{_2}
  \item\label{item:23} \resultImmutables{_1}{_2'}{_2}
  \item\label{item:25} \resultMutables{_1}{_2'}{_2}
  \item\label{item:29} \resultResources{_1}{_2'}{_2}
  \item\label{item:31} \resultThinAir{_1}{_2}
  \end{enumerate}

  It remains to establish the assertions for the let expression.
  \begin{enumerate}[({R}-1)]
  \item \resultOk{_1}{_2}
    \\ Immediate from~\ref{item:17}.
  \item \resultEnv{}{_2}
    \\ Transitivity of store extension
    (Lemma~\ref{lemma:store-extension-transitive}), \ref{item:18},
    and~\ref{item:8}.
  \item  $\SE_2 \vdash r_2 : \tau_2$
    \\ Immediate from~\ref{item:19}.
  \item \resultPermDom{}{_2}
    \\ Immediate from~\ref{item:20}.
  \item \resultReachPerm{}{_2}
    \\ Immediate from~\ref{item:21}
    {because $\Reach{\Store_2}{\VEnv_1}
    \subseteq \Reach{\Store_2}{\VEnv} $ and
    $\Reach{\Store_2}{\Suspended{\VEnv_1}} \subseteq
    \Reach{\Store_2}{\Suspended\VEnv}$.
    Moreover, $\Dom\Store \subseteq \Dom{\Store_1}$.
  }
\item \resultFrame{}{}{_2}
    ~\\ Suppose that $\Loc \in \Dom\Store \setminus
    \Rawloc{\Reach{\Store_2}{\VEnv}}$.
    \\ Then $\Loc \in \Dom\Store \setminus
    \Rawloc{\Reach{\Store_1}{\VEnv_1}}$.
    \\ By~\ref{item:16}, $\Store_1 (\Loc) = \Store (\Loc)$ and for any
    $\Addr$ with $\Rawloc\Addr = \{\Loc\}$: $\Addr \in \Perm
    \Leftrightarrow \Addr \in \Perm_1$.
    \\ But also $\Loc \in \Dom{\Store_1} \setminus
    \Rawloc{\Reach{\Store_2}{\VEnv_2'}}$.
    \\ By~\ref{item:22}, $\Store_2 (\Loc) = \Store_1 (\Loc)$ for
    applicable $\Addr$, $\Addr \in \Perm_1
    \Leftrightarrow \Addr \in \Perm_2$.
    \\ Taken together, we obtain the claim.
  \item \resultImmutables{}{}{_2}
    \\ Follows from~\ref{item:23} or~\ref{item:24} because
    $\ImmutableBorrows\VEnv = \ImmutableBorrows{\VEnv_1} =
    \ImmutableBorrows{\VEnv_2}$.
  \item \resultMutables{}{}{_2}
    \\ Follows from~\ref{item:25},~\ref{item:26}, and framing.
  \item \resultResources{}{}{_2}
    \\ Follows from disjoint splitting of
    $\Active\VEnv$,~\ref{item:29},~\ref{item:30}, and framing.
  \item \resultThinAir{}{_2}
    \\ Immediate from~\ref{item:31}.
  \end{enumerate}

  \newpage{}
  \textbf{Case $e$ of}
  \lstsemrule{sappanf}

  We need to invert rule \TirName{App}:
  \begin{mathpar}
    \ruleSDAIApp
  \end{mathpar}

  We need to establish the assumptions for the recursive call
  @eval \Store' \Perm' \VEnv'(x'-:>r_2) i' e'@.
  We write $\VEnv_2' = \VEnv' (x'\mapsto r_2)$.
  \begin{enumerate}[({A1-}1)]
  \item $\inferS{C'}{\E'; \bvar{x'}{\tau_2}}{e'}{\tau_1}$, for some
    $C'$ and $\E'$\\
    Applying the first premise of \TirName{App} to $r_1 = \VEnv
    (x_1)$, $\SE \vdash \VEnv : \E$, and inversion of result typing
    yields that $r_1 = \Loc_1$ with $\SE (\Loc_1) = \tau_2 \tarr{k}
    \tau_1$.
    By inversion of store typing and storable typing, we find that
    there exist $\E'$ and $C'$ such that
    \begin{enumerate}
    \item $\Store (\Loc_1) = (\VEnv', \lam[k]{x'}{e'})$
    \item $\Disjoint{\E'}$
    \item\label{item:32} $\SE \vdash \VEnv', \E'$
    \item $\inferS{C'}{\E';\bvar {x'}{\tau_2}}{e'}{\tau_1}$
    \item $\addlin{\entail{C'}{\Cleq{\E'}{k}}}$
    \end{enumerate}
  \item $\SE \vdash \VEnv' (x' \mapsto r_2) : \E';
    \bvar{x'}{\tau_2}$\\
    By~\ref{item:32}, assumption on $\VEnv$, and the subtyping premise.
  \item $\vdash \Store' : \SE'$ \\
    by assumption and the released rule of store typing (where we
    write $\SE'=\SE$ henceforth)
  \item \assumeWellformed{'} \\
    the possible removal of a permission does not violate
    wellformedness; the permission is taken away exactly when the
    closure is destroyed
  \item
    \assumeReachable{}{'_2} \\
    as the reach set is a subset of the incoming environment's reach
  \item $\Rawloc{\Active{\VEnv_2'}}$,
    $\Rawloc{\MutableBorrows{\VEnv_2'}}$,
    $\Rawloc{\ImmutableBorrows{\VEnv_2'}}$, and
    $\Rawloc{\Suspended{\VEnv_2'}}$ are all disjoint follows from
    $\Disjoint{\E'}$ and since $r_2= \E (x_2)$ which is an entry
    disjoint from the closure $\E(x_1)$.
  \item \assumeIncoming{'}{_2'}
    The first item holds because of assumption~\ref{item:15}.

    The second item holds because $\REACH' \subseteq \REACH$ from
    assumption~\ref{item:15}.
  \end{enumerate}
  The inductive hypothesis yields that
    $\exists$ $\Store_3$, $\Perm_3$, $r_3$, $\SE_3$ such that
  \begin{enumerate}[({R1-}1)]
  \item \resultOk{}{_3}
  \item \resultEnv{'}{_3}
  \item $\SE_3 \vdash r_3 : \tau_1$
  \item \resultPermDom{}{_3}
  \item \resultReachPerm{}{_3}
  \item \resultFrame{'}{_2'}{_3}
  \item \resultImmutables{'}{_2'}{_3}
  \item \resultMutables{'}{_2'}{_3}
  \item \resultResources{'}{_2'}{_3}
  \item \resultThinAir{'}{_3}
  \end{enumerate}
  The desired results are immediate because $\Dom\Store = \Dom{\Store'}$.

  \newpage{}
  \textbf{Case $e$ of}
  \lstsemrule{sregion}

  We need to invert rule \TirName{Region}
  \begin{mathpar}
    \ruleSDRegion
  \end{mathpar}

  We need to establish the assumptions for the recursive call
  @eval \Store \Perm' \VEnv' i' e'@ where  $\VEnv' =
  \VEnv (x\mapsto \Addr')$.

  \begin{enumerate}[({A1-}1)]
  \item $\inferS{C}{\E'}{e}{\tau}$ \\
    immediate from the inverted premise
  \item $\SE \vdash \VEnv' : \E'$ \\
    the only change of the environments is at $x$; adding the borrow
    modifier $\BORROW$ succeeds due to the second premise; the address $\Addr'$
    stored into $x$ is compatible with its type by store typing
  \item $\vdash \Store : \SE$ \\
    Immediate by outer assumption
  \item\label{item:11} \assumeWellformed{} \\
    Immediate by outer assumption; adding the modifier does not change
    the underlying raw location
  \item\label{item:12} 
    \assumeReachable{}{'} \\
    locations were swapped simultaneously
  \item $\Rawloc{\Active{\VEnv'}}$,
    $\Rawloc{\MutableBorrows{\VEnv'}}$,
    $\Rawloc{\ImmutableBorrows{\VEnv'}}$, and
    $\Rawloc{\Suspended{\VEnv'}}$ are all disjoint \\
    Immediate by assumption
  \item\label{item:15} \assumeIncoming{}{'}
    Immediate by assumption.
  \end{enumerate}

  The induction hypothesis yields the following statements.
  $\exists$ $\Store_1$, $\Perm_1$, $r_1$, $\SE_1$ such that
  \begin{enumerate}[({R1-}1)]
  \item \resultOk{}{_1}
  \item \resultEnv{}{_1}
  \item $\SE_1 \vdash r_1 : \tau$
  \item \resultPermDom{}{_1}
  \item \resultReachPerm{}{_1}
  \item \resultFrame{}{'}{_1}
  \item \resultImmutables{}{'}{_1}
  \item \resultMutables{}{'}{_1}
  \item \resultResources{}{'}{_1}
  \item \resultThinAir{}{_1}
  \end{enumerate}

  It remains to derive the induction hypothesis in the last line. The
  only additional action is the exchange of permissions which
  withdraws the borrow.
  \begin{enumerate}[({R}1)]
  \item \resultOk{}{_1}\\
    Immediate
  \item \resultEnv{}{_1}\\
    Immediate
  \item $\SE_1 \vdash r_1 : \tau$\\
    Immediate
  \item \resultPermDom{}{_1} \\
    The addresses $\Addr$ and $\Addr'$ (as well as the elements of
    $\Perm'$ and $\Perm''$) have the same raw location, so
    exchanging them does not affect wellformedness. The underlying set
    of locations does not change.
  \item \resultReachPerm{}{_1} \\
    This case is critical for region encapsulation. Here we need to
    argue that $\Addr'$ (and hence $\Perm''$) is not reachable from $r_1$ because its type
    $\tau$ is bounded by $\klin_{n-1}$ according to the fourth
    premise. We conclude with Lemma~\ref{lemma:containment}.
  \item \resultFrame{}{'}{_1}
    Immediate
  \item \resultImmutables{}{'}{_1} \\
    Immediate
  \item \resultMutables{}{'}{_1} \\
    Immediate
  \item \resultResources{}{'}{_1} \\
    Immediate
  \item \resultThinAir{}{_1} \\
    Immediate
  \end{enumerate}

  \newpage{}
  \textbf{Case $e$ of}
  \lstsemrule{screateanf}

  We need to invert the corresponding rule
  \begin{mathpar}
    \ruleSDCreate
  \end{mathpar}

  It is sufficient to show that there is some $\SE_1 = \SE(\Loc_1 :
  \tapp\tres\tau)$ such that  
  $\Store_1$, $\Perm_1$, and $r_1 = \Loc_1$ fulfill the following requirements.
  \begin{enumerate}[({R}1)]
  \item \resultOk{}{_1}
  \item \resultEnv{}{_1} \\
    For the last item, we need to show that $\SE (\Loc_1) :
    \tapp\tres\tau$, but this follows from the setting of $w$ to a
    resource storable in the semantics.
  \item $\SE_1 \vdash r_1 : \tapp\tres\tau$ \\
    Immediate from the discussion of the preceding case
  \item \resultPermDom{}{_1} \\
    Follows from the assumption on $\Perm$ and for $\Loc_1$ from the
    allocation of the resource.
  \item \resultReachPerm{}{_1} \\
    Immediate from the assignment to $\Perm_1$.
  \item \resultFrame{}{}{_1} 
    Obvious as no existing location is changed.
  \item \resultImmutables{}{}{_1} \\
    Obvious as no existing location has changed and no permission is withdrawn.
  \item \resultMutables{}{}{_1} \\
    Obvious as no existing location has changed and no permission is
    withdrawn.
  \item \resultResources{}{}{_1} \\
    By the constraint on $\E$ in the \TirName{Create} rule, $\Active\VEnv = \emptyset$.
  \item \resultThinAir{}{_1} \\
    Immediate
  \end{enumerate}
  
  \newpage
  \textbf{Case $e$ of}
  \lstsemrule{sdestroyanf}

  We need to invert rule \TirName{Destroy}.
  \begin{mathpar}
    \ruleSDDestroy
  \end{mathpar}
  It is sufficient to show that $\SE_1 = \SE$,
  $\Store_1$, $\Perm_1$, and $r_1 = ()$ fulfill the following
  requirements.
  \begin{enumerate}[({R}1)]
  \item \resultOk{}{_1}
  \item \resultEnv{}{_1} \\
    Immediate: $\Loc$ was updated to void, which has any type.
  \item $\SE_1 \vdash () : \tunit$
  \item \resultPermDom{}{_1}\\
    By assumption on $\Perm$ and because $\Loc$ was removed.
  \item \resultReachPerm{}{_1} \\
    Immediate because the reach set is empty
  \item \resultFrame{}{}{_1} 
    Only $\Store (\Loc)$ was changed, which is not reachable from the frame.
  \item \resultImmutables{}{}{_1} \\
    Immediate because we updated (destroyed) a resource (in $\Active\VEnv$).
  \item \resultMutables{}{}{_1} \\
    Immediate because we updated (destroyed) a resource (in $\Active\VEnv$).
  \item \resultResources{}{}{_1} \\
    By the constraint on $\E$, $\Loc$ was the only resource passed to
    this invocation of eval. The claimed condition holds as $\Loc$ was
    removed from $\Perm_1$ and the location's contents cleared.
  \item \resultThinAir{}{_1}
    \\ Immediate
  \end{enumerate}

  \newpage
  \textbf{Case $e$ of}
  \lstsemrule{var}

  We need to invert rule \TirName{Var}.
  \begin{mathpar}
    \ruleSDIVar
  \end{mathpar}

  We establish that the claims hold for
  $\Store' = \Store$, $\Perm' = \Perm$, $r = \VEnv (x)$, and $\SE' = \SE$.
  \begin{enumerate}[({R}1)]
  \item \resultOk{}{}
  \item \resultEnv{}{}
    \\ Immediate by reflexivity and assumption.
  \item $\SE \vdash r : \tau$
    \\ Immediate by assumption~\ref{item:32}.
  \item \resultPermDom{}{}
    \\ Immediate by assumption~\ref{item:11}
  \item \resultReachPerm{}{}
    \\ Immediate
  \item \resultFrame{}{}{}
    Immediate as permissions and store stay the same.
  \item \resultImmutables{}{}{}
    \\ Immediate
  \item \resultMutables{}{}{}
    \\Immediate
  \item \resultResources{}{}{}
    \\ As $\Perm$ remains the same, a linear resource in $x$ is
    returned untouched.
  \item \resultThinAir{}{}
    \\ Immediate
  \end{enumerate}

  \newpage
  \textbf{Case $e$ of}
  \lstsemrule{const}

  We need to invert rule \TirName{Const}.
  \begin{mathpar}
    \ruleSDConst
  \end{mathpar}

  We need to establish the claims for $\Store' = \Store$,
  $\Perm'=\Perm$, $r' = c$, and $\SE' = \SE$:
  \begin{enumerate}[({R}1)]
  \item \resultOk{}{}
  \item \resultEnv{}{} \\
    By assumption~\ref{item:33}.
  \item $\SE \vdash c : \CType c$
    \\by result typing.
  \item \resultPermDom{}{}
    \\ By assumption~\ref{item:11}.
  \item \resultReachPerm{}{}
    \\ As $\Reach\Store c = \emptyset$.
  \item \resultFrame{}{}{}
    Immediate
  \item \resultImmutables{}{}{}
    \\ Immediate
  \item \resultMutables{}{}{}
  \item \resultResources{}{}{}
    \\ Immediate as $\Active\VEnv = \emptyset$.
  \item \resultThinAir{}{}
    \\ Immediate
  \end{enumerate}

  \newpage
  \textbf{Case $e$ of}
  \lstsemrule{spairanf}

  We need to invert rule \TirName{Pair}.
  \begin{mathpar}
    \ruleSDAIPair
  \end{mathpar}
  Show that $\Store'$, $\Perm'$, $r' = \Loc'$, $\SE' = \SE (\Loc' : \tyPair[k]{\tau_1}{\tau_2})$ such that
  \begin{enumerate}[({R}1)]
  \item \resultOk{}{'}
  \item \resultEnv{}{'}
  \item $\SE' \vdash \Loc' : \tyPair{\tau_1}{\tau_2}$
  \item \resultPermDom{}{'}
    \\ By assumption~\ref{item:11} and because $\Loc'$ is properly initialized.
  \item \resultReachPerm{}{'}
    \\ By assumption~\ref{item:12}, $\Reach{\Store'}{\Loc'} =
    \Reach\Store{r_1, r_2} \cup \{\Loc'\} \subseteq
    \Reach{\Store}\VEnv \cup \{\Loc'\}$ and $\{\Loc'\} = \Dom{\Store'}
    \setminus \Dom\Store$.
  \item \resultFrame{}{}{'}
    Immediate
  \item \resultImmutables{}{}{'}
    \\ Immediate
  \item \resultMutables{}{}{'}
    \\ Immediate
  \item \resultResources{}{}{'}
    \\ Every such $\Loc$ must be reachable either from $r_1$ or
    $r_2$. So they become reachable from $\Loc'$, as required.
  \item \resultThinAir{}{'}
    \\ Immediate
  \end{enumerate}

  \newpage
  \textbf{Case $e$ of}
  \lstsemrule{lam}

  We need to invert rule \TirName{Abs}
  \begin{mathpar}
    \ruleSDLam
  \end{mathpar}

  Show that $\Store'$, $\Perm'$, $r' = \Loc'$, and $\SE' = \SE (\Loc'
  : \tau_2\tarr{k}\tau_1)$ fulfill
  \begin{enumerate}[({R}1)]
  \item \resultOk{}{'}
  \item \resultEnv{}{'} \\
    Immediate by definition and store typing
  \item $\SE' \vdash r' : \tau_2\tarr{k}\tau_1$ \\
    Immediate by store typing
  \item \resultPermDom{}{'} \\
    Wellformedness holds by assumption on $\Perm$ and because $\Loc'$ is a new location.
    The domain constraint is assumed for $\Perm$ and $\Loc'$ is
    initialized to a  closure.
  \item \resultReachPerm{}{'} 
    \begin{align*}
      \Reach{\Store'}{r'} & = \{ \Loc' \} \cup \Reach{\Store'}{\VEnv} \\
      &= \Dom{\Store'}\setminus\Dom\Store \cup \Reach{\Store'}{\VEnv}
    \end{align*}
    Moreover, the constraint ${\Cleq{\E}{k}}$ implies that
    $\Suspended\VEnv = \emptyset$.
  \item \resultFrame{}{}{'} 
    Immediate
  \item \resultImmutables{}{}{'}   \\
    Immediate
  \item \resultMutables{}{}{'} \\
    Immediate
  \item \resultResources{}{}{'} \\
    The second case is immediately applicable.
  \item \resultThinAir{}{'} \\
    Immediate
  \end{enumerate}

  \newpage
  \textbf{Case $e$ of}
  \lstsemrule{sborrow}

  We have to invert rule \TirName{Borrow}
  \begin{mathpar}
    \ruleSDBorrow
  \end{mathpar}

  Show that $\Store' = \Store$, $\Perm' = \Perm$, $r' = \Addr$, $\SE'
  = \SE$ such that
  \begin{enumerate}[({R}1)]
  \item \resultOk{}{'}
  \item \resultEnv{}{'} \\
    Immediate, no changes.
  \item $\SE' \vdash r' : \borrowty{k}{\tau}$\\
    Immediate by result typing and because the interpreter checks that
    the permissions of the borrow are very restricted.
  \item \resultPermDom{}{'}\\
    Immediate (no change).
  \item \resultReachPerm{}{'} \\
    By typing, $\Addr$ is not in $\Suspended\VEnv$.
    Hence,  the condition is immediate.
  \item \resultFrame{}{}{'} 
    Immediate as no change.
  \item \resultImmutables{}{}{'}\\
    Immediate as no change
  \item \resultMutables{}{}{'} \\
    Immediate
  \item \resultResources{}{}{'} \\
    Immediate because $\REACH$, $\REACH'$ must be empty
  \item \resultThinAir{}{'}\\
    Immediate.
  \end{enumerate}

  \newpage
  \textbf{Case $e$ of}
  \lstsemrule{sobserveanf}

  We have to invert the rule \TirName{Observe}
  \begin{mathpar}
    \ruleSDObserve
  \end{mathpar}
  Show that $\Store' = \Store$, $\Perm' = \Perm$, $r'$, and $\SE' =
  \SE$ fulfill
  \begin{enumerate}[({R}1)]
  \item \resultOk{}{'}
  \item \resultEnv{}{'} \\
    By reflexivity and assumption.
  \item $\SE' \vdash r' : \tau$ \\
    Immediate by store typing
  \item \resultPermDom{}{'} \\
    Immediate: no changes.
  \item \resultReachPerm{}{'} \\
    Immediate
  \item \resultFrame{}{}{'}
    Immediate: no changes.
  \item \resultImmutables{}{}{'} \\
    Immediate: no changes to immutables.
  \item \resultMutables{}{}{'}  \\
    Immediate: one particular $\Addr$ is overwritten, but not freed.
  \item \resultResources{}{}{'} \\
    Immediate because $\REACH = \emptyset$
  \item \resultThinAir{}{'} \\
    Immediate
  \end{enumerate}

  \newpage
  \textbf{Case $e$ of}
  \lstsemrule{supdateanf}

  We need to invert rule \TirName{Update}
  \begin{mathpar}
    \ruleSDUpdate
  \end{mathpar}

  We need to show that $\Store'$, $\Perm'$, $r' = ()$, $\SE' = \SE$ fulfill
  \begin{enumerate}[({R}1)]
  \item \resultOk{}{'}
  \item \resultEnv{}{'} \\
    Immediate by store typing for $\Loc$
  \item $\SE' \vdash r' : \tunit$ \\
    Immediate
  \item \resultPermDom{}{'} \\
    Immediate, as we remove a permission from $\Perm$
  \item \resultReachPerm{}{'}\\
    Immediate, as we only update a reachable $\Loc$
  \item \resultFrame{}{}{'}
    Immediate
  \item \resultImmutables{}{}{'} \\
    Immediate
  \item \resultMutables{}{}{'} \\
    Immediate; for $\Loc$, we observe that it is overwritten, but not  freed.
  \item \resultResources{}{}{'} \\
    Immediate because $\Active\VEnv = \emptyset$ and hence $\REACH = \emptyset$.
  \item \resultThinAir{}{'}
  \end{enumerate}

  \newpage
  \textbf{Case $e$ of}
  \lstsemrule{smatchanf}

  We need to invert rule \TirName{MatchPair}
  \begin{mathpar}
    \ruleSDAIMatchPair
  \end{mathpar}
  The case VMatch corresponds to the match specification
  $\etransfm = \operatorname{id}$.

  Establish the assumptions for the recursive call
  with $\VEnv'_2 = \VEnv_2 (x\mapsto r_1) (x' \mapsto r_1')$ and $\SE' = \SE$:
  \begin{enumerate}[({A1-}1)]
  \item $\inferS{C}{\E_2 (x:\tau_1) (x':\tau'_1)}{e_2}{\tau_2}$ by inversion
  \item $\SE \vdash \VEnv_2 : \E_2$  by assumption; moreover, $\SE
    \vdash r_1 : \tau_1$ and $\SE \vdash r_1' : \tau'_1$ by inversion
    of the store typing for $\Loc$. As $\SE'=\SE$, we have $\SE' \vdash \VEnv_2' : \E_2 (x:\tau_1) (x':\tau_1')$.
  \item $\vdash \Store' : \SE'$ : the only change from assumption is in
    $\Loc$ which potentially maps to $\StFreed$.
  \item \assumeWellformed{'} : permission to $\Loc$ is removed iff
    $\Loc$ is mapped to $\StFreed$.
  \item \assumeReachable{'}{'_2} \\
    by assumption
  \item $\Rawloc{\Active{\VEnv'_2}}$,
    $\Rawloc{\MutableBorrows{\VEnv'_2}}$,
    $\Rawloc{\ImmutableBorrows{\VEnv'_2}}$, and
    $\Rawloc{\Suspended{\VEnv'_2}}$ are all disjoint: by assumption
    and splitting
  \item \assumeIncoming{'}{'_2}
  \end{enumerate}

  Hence the call to @eval@ yields
  $\exists$ $\Store_2$, $\Perm_2$, $r_2$, $\SE_2$ such that
  \begin{enumerate}[({R1-}1)]
  \item \resultOk{}{_2}
  \item \resultEnv{'}{_2}
  \item $\SE_2 \vdash r_2 : \tau_2$
  \item \resultPermDom{}{_2}
  \item \resultReachPerm{}{_2}
  \item \resultFrame{}{}{_2}
  \item \resultImmutables{}{}{'}
  \item \resultMutables{}{}{'}
  \item \resultResources{}{}{'}
  \item \resultThinAir{}{'}
  \end{enumerate}
  As $R_2$ is also returned from the match, these results carry over.
  
  \newpage
  \textbf{Case $e$ of}
  \lstsemrule{matchborrowanf}

  We need to invert rule \TirName{MatchPair}
  \begin{mathpar}
    \ruleSDAIMatchPair
  \end{mathpar}
  The case VMatchborrow corresponds to the match specification
  $\etransfm = \&^{\BORROW}$.
  In contrast to the non-borrowing match, the borrowed pair is never deallocated.

  Establish the assumptions for the recursive call
  with  $\SE'' = \SE$:
  \begin{enumerate}[({A1-}1)]
  \item $\inferS{C}{\E_2 (x:\borrow{\tau_1}) (x':\borrow{\tau'_1})}{e_2}{\tau_2}$ by inversion
  \item $\SE \vdash \VEnv_2 : \E_2$  by assumption; moreover, $\SE
    \vdash r_1'' : \borrow{\tau_1}$ and $\SE \vdash r_2'' : \borrow{\tau'_1}$ by inversion
    of the store typing for $\Addr$. As $\SE''=\SE$, we have $\SE'' \vdash \VEnv_2'' : \E_2 (x:\borrow{\tau_1}) (x':\borrow{\tau_1'})$.
  \item $\vdash \Store' : \SE''$ : the only change from assumption is in
    $\Loc$ which potentially maps to $\StFreed$.
  \item \assumeWellformed{''} : permission to $\Loc$ is removed iff
    $\Loc$ is mapped to $\StFreed$.
  \item \assumeReachable{''}{''_2} \\
    by assumption
  \item $\Rawloc{\Active{\VEnv''_2}}$,
    $\Rawloc{\MutableBorrows{\VEnv''_2}}$,
    $\Rawloc{\ImmutableBorrows{\VEnv''_2}}$, and
    $\Rawloc{\Suspended{\VEnv''_2}}$ are all disjoint: by assumption
    and splitting
  \item \assumeIncoming{''}{'_2}
  \end{enumerate}

  Hence the call to @eval@ yields
  $\Store_2$, $\Perm_2$, $r_2$, $\SE_2$ such that
  \begin{enumerate}[({R1-}1)]
  \item \resultOk{}{_2}
  \item\label{item:1} \resultEnv{''}{_2}
  \item\label{item:3} $\SE_2 \vdash r_2 : \tau_2$
  \item\label{item:4} \resultPermDom{}{_2}
  \item\label{item:5} \resultReachPerm{_2''}{_2}
  \item\label{item:6} \resultFrame{''}{_2''}{_2}
  \item\label{item:7} \resultImmutables{''}{_2''}{_2}
  \item\label{item:34} \resultMutables{''}{_2''}{_2}
  \item\label{item:35} \resultResources{''}{_2''}{_2}
  \item\label{item:36} \resultThinAir{''}{_2}
  \end{enumerate}

  It remains to relate to result with the original call to @eval@.
  \begin{enumerate}[(R1)]
  \item \resultOk{}{_2}
  \item \resultEnv{}{_2} because $\SE'' = \SE$ and \ref{item:1}.
  \item $\SE_2 \vdash r_2 : \tau_2$ by \ref{item:3}
  \item \resultPermDom{}{_2}
    Immediate from~\ref{item:4}.
  \item \resultReachPerm{}{_2} By~\ref{item:5} and because $\Store =\Store''$.
  \item \resultFrame{}{}{_2} Immediate from~\ref{item:6} because
    $\Store = \Store''$
  \item \resultImmutables{}{}{_2}
  \item \resultMutables{}{}{_2}
  \item \resultResources{}{}{_2}
    Immediate by~\ref{item:35} because the borrowing match does not deallocate.
  \item \resultThinAir{}{_2}
  \end{enumerate}

\end{proof}


\end{document}